\documentclass[fleqn,11pt,twoside]{article}

\usepackage{latexsym,amsmath,amsfonts,amssymb,amsthm,amscd,mathrsfs}

\usepackage{accents}

\usepackage{enumitem}

\usepackage{tikz}   
\usetikzlibrary{arrows,arrows.meta}  
\usepackage{caption}

 \usepackage{color}

\usepackage[unicode,colorlinks,plainpages=false,hyperindex=true,bookmarksnumbered=true,bookmarksopen=false,pdfpagelabels]{hyperref}
 \hypersetup{urlcolor=cyan,linkcolor=blue,citecolor=red,colorlinks=true}

\newcommand{\be}{\begin{equation}}
\newcommand{\ee}{\end{equation}}
\newcommand{\bea}{\begin{eqnarray}}
\newcommand{\eea}{\end{eqnarray}}

\newcommand{\ba}{\begin{aligned}}
\newcommand{\ea}{\end{aligned}}

\numberwithin{equation}{section}
\newcounter{thmcounter}
\numberwithin{thmcounter}{section}

\theoremstyle{definition}

\newtheorem{definition}[thmcounter]{Definition}
\newtheorem{remark}[thmcounter]{Remark}

\theoremstyle{plain}
\newtheorem{corollary}[thmcounter]{Corollary}
\newtheorem{lemma}[thmcounter]{Lemma}
\newtheorem{proposition}[thmcounter]{Proposition}
\newtheorem{theorem}[thmcounter]{Theorem}
\newtheorem{conjecture}[thmcounter]{Conjecture}

\def\1{{\boldsymbol 1}}                     %
\def\cH{{\mathcal H}}                       %
\def\ri{{\rm i}}                            %
 \def\C{\mathbb{C}}                          %
\def\fH{\mathfrak{H}}                       %
\def\fI{\mathfrak{I}}
\def\cF{{\mathcal F}}                       %
\def\reg{\mathrm{reg}}                      %
\def\red{\mathrm{red}}                      %
\def\res{\mathrm{res}}
\def\span{{\mathrm{span}}}                  %
\def\cR{{\mathcal R}}                       %
\def\Ad{{\mathrm{Ad}}}                      %
\def\id{{\mathrm{id}}}                      %
\def\cA{{\mathcal A}}                       %
\def\dt {\left.\frac{d}{dt}\right|_{t=0}}   %
\def\fR{{\mathfrak{R}}}                        %
 \def\B{\mathrm{B}}                          %
\def\cM{{\mathcal M}}                       %
\def\sgn{\mathrm{sgn}}

\def\fC{{\mathfrak{C}}}                   %

\def\cE{{\mathcal E}}

\newcommand{\CC}{\ensuremath{\mathbb{C}}}

\newcommand{\R}{\ensuremath{\mathbb{R}}}
\newcommand{\Z}{\ensuremath{\mathbb{Z}}}
\newcommand{\D}{\ensuremath{\mathbb{D}}}
\newcommand{\MM}{\ensuremath{\mathcal{M}}}

\newcommand{\tr}{\operatorname{tr}}

\newcommand{\diag}{\operatorname{diag}}

\newcommand{\Hom}{\operatorname{Hom}}

\newcommand{\Mat}{\operatorname{Mat}}

\newcommand{\Gl}{\operatorname{GL}}
\newcommand{\UU}{\operatorname{U}}
\newcommand{\ad}{\operatorname{ad}}

\newcommand{\Vect}{\operatorname{Vect}}
\newcommand{\disk}{\operatorname{D}}
\newcommand{\gl}{\ensuremath{\mathfrak{gl}}}
\newcommand{\g}{\ensuremath{\mathfrak{g}}}
\newcommand{\uu}{\ensuremath{\mathfrak{u}}}

\newcommand{\Cinf}{\ensuremath{\mathcal{C}^\infty}}
\newcommand\ic{\ensuremath{\mathfrak{i}}}
\newcommand{\att}{\ensuremath{\mathtt{a}}}
\newcommand{\btt}{\ensuremath{\mathtt{b}}}
\newcommand{\ctt}{\ensuremath{\mathtt{c}}}

\newcommand\br[1]{\{ #1 \}}
\newcommand\ip[1]{\langle #1 \rangle}

\def\cC{{\mathrm{C}}}
\newcommand{\bT}{\ensuremath{\mathbb{T}}}

\def\fT{{\mathfrak t}}
\def\End{{\mathrm{End}}}
\def\fN{{\mathfrak{N}}}
\def\fC{{\mathfrak{C}}}
\def\mC{\mathrm{C}}

\def\SU{{\mathrm{SU}}}

\def\cO{{\mathcal O}}
\def\bM{{\mathbb{M}}}
\def\fF{{\mathfrak{F}}}

\begin{document}

\begin{center}
 {\Large\bf Integrable multi-Hamiltonian systems from reduction of
 an extended quasi-Poisson double of $\UU(n)$}
\end{center}

\medskip
\begin{center}
M.~Fairon${}^{a}$ and  L.~Feh\'er${}^{b,c}$
\\

\bigskip
${}^a$Laboratoire de math\'ematiques d'Orsay\\
Universit\'e Paris-Saclay, F-91405 Orsay, France\\
e-mail: maxime.fairon@universite-paris-saclay.fr

\medskip
${}^b$Department of Theoretical Physics, University of Szeged\\
Tisza Lajos krt 84-86, H-6720 Szeged, Hungary\\
e-mail: lfeher@physx.u-szeged.hu

\medskip
${}^c$Institute for Particle and Nuclear Physics\\
Wigner Research Centre for Physics\\
 H-1525 Budapest, P.O.B.~49, Hungary

\end{center}

 \setcounter{tocdepth}{2}

\medskip
\begin{abstract}
We construct a master dynamical system on a $\UU(n)$ quasi-Poisson manifold, $\cM_d$, built from the double
$\UU(n) \times \UU(n)$ and $d\geq 2$ open balls in $\CC^n$, whose
 quasi-Poisson structures are obtained from $T^* \R^n$  by  exponentiation.
A pencil of quasi-Poisson bivectors $P_{\underline{z}}$ is defined on $\cM_d$ that depends on $d(d-1)/2$ arbitrary
real parameters and gives rise to pairwise compatible Poisson brackets on the $\UU(n)$-invariant functions.
The master system on $\cM_d$ is a quasi-Poisson analogue of the degenerate integrable system of free motion
on the extended cotangent bundle $T^*\!\UU(n) \times \CC^{n\times d}$. Its
commuting Hamiltonians are  pullbacks of the class functions on one of the $\UU(n)$ factors.
We prove that the master system descends to a degenerate integrable system on a dense open subset of
the smooth component  of the quotient space $\cM_d/\UU(n)$ associated with the principal orbit type.
Any reduced Hamiltonian  arising from a class function generates the same flow  via any
of the compatible  Poisson structures stemming from the bivectors $P_{\underline{z}}$.
The restrictions of the reduced system on minimal symplectic leaves parameterized by generic elements
of the center of $\UU(n)$ provide a new real form of the complex, trigonometric
spin Ruijsenaars--Schneider model of Krichever and Zabrodin.
This generalizes the derivation of the  compactified trigonometric RS model found previously  in the $d=1$ case.
\end{abstract}

 \newpage

{\linespread{0.7}\tableofcontents}

 \newpage

\section{Introduction}
\label{sec:I}

The present paper is intended to advance the reduction approach
to integrable Hamiltonian systems.
The application of this method always starts with a manifestly integrable `master system' having a large symmetry group.
Then, taking quotient by the symmetry group or its suitable subgroup leads to the reduced system.
The quotient space gets equipped with a Poisson structure, and can be decomposed into
a disjoint union of symplectic leaves. Integrability is usually inherited by
the restrictions of the reduced system on the generic symplectic leaves, and on particular symplectic leaves
one may find models that have important applications.

Notable successes of the  method include
the  reduction treatment of integrable many-body models of
Calogero--Moser--Sutherland and Ruijsenaars--Schneider type \cite{AFM,FK2,FK,FGNR,KKS,N,OP,Pu}.
The study of this celebrated family of integrable systems and their extensions by so-called spin variables
started decades ago \cite{C,Mo,Su,RS,GH,KZ,LX,Wo} and still attracts considerable attention \cite{ARe,Fa,FG,Fe1,Fe2,KLOZ,KMK,Re3,TZ}.
The prototype of the spin-particle models was introduced by Gibbons and Hermsen \cite{GH} using
Hamiltonian reduction in a complex holomorphic setting.
Interesting generalizations of the Gibbons--Hermsen model are the spin Ruijsenaars--Schneider (RS) models
of Krichever and Zabrodin \cite{KZ}, who found their models working  at the level of the equations of motion.
The Hamiltonian structure of the holomorphic trigonometric/hyperbolic model
of Krichever and Zabrodin has been revealed by one of us in collaboration
with Chalykh \cite{CF2} relying on  complex quasi-Poisson geometry.
Soon after, another reduction treatment of the same model was independently developed
by Arutyunov and Olivucci building on the theory of Poisson--Lie groups \cite{AO}.
Then, a trigonometric real form of the model was derived in our joint paper with Marshall \cite{FFM},
again using Poisson--Lie groups.
As a complex integrable system may have several real forms possessing very different
 properties, it is natural
to ask if new real forms of the Krichever--Zabrodin model can be obtained by utilizing
 quasi-Poisson techniques in a real setting.

\medskip

Since it provides our basic motivation, let us recall the reduction behind the
trigonometric real form of the Gibbons--Hermsen model.
In this case, the `master phase space'
is the symplectic manifold
\be
M_d:= T^*\!\UU(n) \times \C^{n\times d} \simeq \UU(n) \times \uu(n) \times \C^{n\times d} = \{ (g,J,\sigma)\},
\qquad \forall d\geq 2,
\label{I1}\ee
where the cotangent bundle $T^*\!\UU(n)$ is trivialized by right-translations and
an identification $\uu(n)^* \simeq \uu(n)$ is made.
The entries of $\sigma\in\C^{n\times d}$ encode independent copies  of the
symplectic vector space $\R^2\simeq \C$.
The commuting Hamiltonians of an integrable master system on $M_d$ are defined
by the conjugation invariant functions on $\uu(n)$.
A Hamiltonian action of $\UU(n)$ on $M_d$ is
engendered by the maps
\be
A_\eta: (g,J, \sigma) \mapsto (\eta g \eta^{-1}, \eta J \eta^{-1}, \eta \sigma),\quad \forall \eta \in \UU(n).
\label{I2}\ee
This action is generated by the classical moment map $\phi: T^*\!\UU(n) \times \C^{n\times d}\to \uu(n)$ with
\be
\phi(g,J,\sigma)=J - g^{-1} J g + \ic \sum_{\alpha=1}^d \sigma_\alpha \sigma_\alpha^\dagger,
\label{phi}\ee
where the $\sigma_\alpha$ are the columns of the matrix $\sigma$.
Denoting by $\cO\subset \uu(n)^*$ a co-adjoint orbit, $\phi^{-1}(\cO)/\UU(n)$
yields a Poisson subspace of the full reduced phase space $M_d/\UU(n)$.
Under further conditions, $\phi^{-1}(\cO)/\UU(n)$ is actually a symplectic leaf.
Picking the smallest non-trivial orbits, $\cO = \{ \ic \gamma \1_n\}$
with an arbitrary positive constant $\gamma$, one obtains the Gibbons--Hermsen model~\cite{GH}, whose
main Hamiltonian descends from the `kinetic energy of free motion' furnished by
$-\frac{1}{2} \tr(J^2)$.
As is well known \cite{KKS},
the $d=1$ case leads to the spinless Sutherland model with coupling parameter $\gamma$.

A generalization of the above sketched construction
 \cite{FFM}  begins by replacing the cotangent bundle by the so-called
Heisenberg double \cite{STS}, which as a (real) manifold is provided by $\Gl(n,\C)$.
Moreover, the column vectors that constitute $\sigma\in \C^{n\times d}$ are equipped
with a Poisson structure that is $\UU(n)$-covariant, with respect to the  standard Poisson--Lie group structure on $\UU(n)$.
The Poisson--Lie counterpart of $\uu(n)^*$ is the dual Poisson--Lie group $\B(n)$ of $\UU(n)$, alias
the group of $n\times n$ upper-triangular matrices having positive diagonal elements.
The moment map $\phi$ \eqref{phi} is replaced by a Poisson map $\Lambda$ from
$\bM_d:=\Gl(n,\C) \times \C^{n\times d}$ into $\B(n)$ having the product form
$\Lambda = \Lambda_L \Lambda_R b_1 \cdots b_d$,
where $\Lambda_L$ and $\Lambda_R$ depend only on $\Gl(n,\C)$,
$b_\alpha$ depends only on the $\alpha$-th column of $\sigma\in \C^{n\times d}$,
and all the factors of $\Lambda$ are Poisson maps into $\B(n)$.
The Poisson--Lie moment map $\Lambda$
generates a Poisson action of $\UU(n)$ on $\bM_d$, and  this was used for reducing
the integrable master system whose commuting Hamiltonians are the
conjugation invariant functions of $\Lambda_R \Lambda_R^\dagger$.
Then,  a trigonometric real form of the Krichever--Zabrodin model was found  on the symplectic
leaf $\Lambda^{-1}(e^\gamma \1_n)/\UU(n)  \subset \bM_d/\UU(n)$ of the full quotient space, for any $\gamma>0$ and $d\geq 2$.
The Krichever--Zabrodin equations of motion are associated with the reduced Hamiltonian
descending from  $\tr (\Lambda_R \Lambda_R^\dagger)$.
 In  the $d=1$ case,  this construction yields
the spinless trigonometric RS model \cite{FK2}.

The unreduced master systems are integrable
in the degenerate sense both in the cotangent bundle and in the Heisenberg double cases.
Let us recall \cite{J,Nek,Re2} that a degenerate integrable system on a symplectic manifold of dimension $2m$ consists
of a ring of commuting Hamiltonians (with complete flows)  and the ring of their joint constants of motion, such
that the functional dimensions of the two rings sum up to $2m$.
More precisely, it is further assumed that the number, $r$, of independent commuting Hamiltonians satisfies $r<m$.
In the case of Liouville integrability $r=m$ and the two rings coincide.
It was proved in the papers \cite{AO,CF2,FFM} that degenerate integrability is inherited by the special,
Krichever--Zabrodin type reduced systems for any $d\ge 2$. It is also known that
the spinless reduced systems that result  for $d=1$  are (`only') Liouville integrable \cite{Ru}.

The principal goal of the present paper is to explore the generic reductions of a different master system,
relying on quasi-Poisson \cite{AKSM} and quasi-Hamiltonian \cite{AMM} techniques.
In addition to the previously mentioned papers, our work also  builds on a recent study of reductions
of doubles of compact Lie groups \cite{Fe}.

Our starting phase space will be the manifold
\be
\cM_d := \UU(n) \times \UU(n) \times \disk(x_1) \times \cdots \times \disk(x_d) = \{(A,B,v_1,\dots, v_d)\},
\label{I4}\ee
where $x_1,\dots, x_d$ are non-zero real parameters and $\disk(x_\alpha):= \{ v_\alpha \in \CC^n\mid \vert v_\alpha\vert^2 < 2\pi/ |x_\alpha|\}$.
The building block $\UU(n)\times \UU(n)$ of $\cM_d$ is the so-called internally fused
double,
 a quasi-Poisson analogue
of the cotangent bundle $T^*\!\UU(n)$.  Every open ball $\disk(x_\alpha)$ is endowed with a $\UU(n)$ quasi-Poisson structure,
which blows up at the boundary of the ball and is
obtained from  $T^*\R^n$ by exponentiation \cite{AKSM} and a rescaling of the variables.
Then, a quasi-Poisson bivector $P_{\cM_d}$ on $\cM_d$ results from these building blocks by  application of
the fusion procedure of \cite{AKSM}.
What is more, $\cM_d$ will be equipped with a `pencil' of
quasi-Poisson bivectors $P_{\underline{z}}$ of the form\footnote{In the wording of our presentation it is often
assumed that $d\geq 2$,  but the $d=1$ case, when the bivector $P_{\underline{z}}$ \eqref{I5} collapses to $P_{\cM_d}$, is also interesting.}
\be
P_{\underline{z}} = P_{\cM_d} + \psi_{\underline{z}},
\label{I5}\ee
which depend not only on the fixed parameters $x_1,\dots, x_d$, but also on further
arbitrary parameters $z_{\alpha \beta}\in \R$, for  $1\leq \alpha < \beta \leq d$.
All these bivectors
are invariant with respect to the $\UU(n)$-action operating via the maps
\be
\cA_\eta : (A,B,v_1,\dots, v_d) \mapsto (\eta A \eta^{-1}, \eta B \eta^{-1}, \eta v_1,\dots, \eta v_d), \quad \forall \eta\in \UU(n),
\label{I6}\ee
and  give rise to genuine Poisson brackets on the  $\UU(n)$-invariant functions.
For any two choices of the $z_{\alpha\beta}$, the corresponding Poisson brackets of the invariant functions are
 compatible in the usual sense of Poisson geometry, that is, their arbitrary linear combinations satisfy the Jacobi identity.
Thus, by identifying its space of smooth functions with the $\UU(n)$-invariant functions on $\cM_d$,
 the quotient space $\cM_d/\UU(n)$ inherits a family of compatible Poisson structures. It will also be explained   that the
 map $\Phi: \cM_d \to \UU(n)$, defined by
\be
\Phi(A,B,v_1,\dots, v_d) =  A B A^{-1} B^{-1} e^{\ic x_1 v_1 v_1^\dagger} \cdots e^{\ic x_d v_d v_d^\dagger},
\label{I7}
\ee
is a group valued moment map,  giving $(\cM_d, P_{\underline{z}}, \Phi)$ the structure of  a `Hamiltonian quasi-Poisson manifold' in the sense of \cite{AKSM}.
Moreover, $\cM_d$  admits a quasi-Hamiltonian structure as well; the 2-form
$\omega_{\underline{z}}$ compatible with $P_{\underline{z}}$ and $\Phi$ will be given explicitly. (We shall also explain how this pencil
can be adapted to the extended cotangent  bundle $M_d$ \eqref{I1}.)
It follows from the general theory that, for  any conjugacy class $\mC \subset \UU(n)$,
$\Phi^{-1}(\mC)/\UU(n)$ becomes a Poisson subspace for any of the compatible Poisson structures on $\cM_d/\UU(n)$.
If the $\UU(n)$-action is free on $\Phi^{-1}(\mC)$, then this quotient is a symplectic manifold.

The phase space $\cM_d$  carries a dynamical system
quite akin to a degenerate
integrable system on a symplectic manifold. To describe this, remark that the bivector $P_{\underline{z}}$ corresponds to a
quasi-Poisson bracket on $\Cinf(\cM_d)$, and every function $H\in \Cinf(\cM_d)$  gives rise to a vector field $X_H$.
Our master system is generated by the Hamiltonians of the form
\be
H(A,B,v_1,\dots, v_d):= h(A) \quad \hbox{with any}\quad h\in \Cinf(\UU(n))^{\UU(n)}.
\label{I8}\ee
These Hamiltonians are in involution with respect to the quasi-Poisson bracket and form a ring of functional dimension $n$.
It turns out that $A$ and the $v_\alpha$ are constant along the integral curves of $X_H$, while $B$
develops according to $B(t) = B(0) \exp( -t \nabla h(A(0)) )$, where $\nabla h$ is the $\uu(n)$-valued `differential' of $h$.
This is very similar to the unreduced dynamics on $T^*\!\UU(n)$.
The ring of joint constants of motion of the Hamiltonians \eqref{I8} has functional dimension $\dim(\cM_d) -n$,
and is closed under the quasi-Poisson bracket.
We shall present a neat characterization of the ring of constants of motion
in terms of a (non-surjective) map $\Psi: \cM_d \to \cM_d$,
\be
\Psi (A,B, v_1,\dots, v_d) = (A, B A B^{-1}, v_1,\dots, v_d),
\label{I9}\ee
which is a quasi-Poisson map
with respect to $P_{\underline{z}}$ on its domain and  a certain degenerate
quasi-Poisson structure on its target space, and is constant along the integral curves.

The quotient space $\cM_d^\red:= \cM_d/\UU(n)$ carries the Poisson algebra of the smooth functions
 $\Cinf(\cM_d)^{\UU(n)}$.
However, $\cM_d^\red$  is not a smooth manifold, but a stratified Poisson space \cite{SL,Sn}.
We shall restrict our study of the reduction  to the dense open subset of the quotient space
corresponding to the subset  $\cM_{d*}\subset \cM_{d}$ on which the $\UU(n)$-action is free.
Then, $\cM_{d*}^\red :=\cM_{d*}/\UU(n)$ is a smooth Poisson manifold, having generic symplectic leaves of codimension $n$.
We prove that the functional dimension of the reduced commuting Hamiltonians on $\cM_{d*}^\red$
that descend from our master system remains $n$, and can also show that their Hamiltonian vector fields
span an $n$-dimensional subspace of the tangent space over a dense open subset of $\cM_{d*}^\red$.
These reduced Hamiltonian vector fields and their flows are the projections
of the unreduced quasi-Hamiltonian vector fields $X_H$ and their flows, which stay in $\cM_{d*}$.
All this would mean \cite{LGMV} that our master system reduces to a degenerate integrable system on
the Poisson manifold $\cM_{d*}^\red$, if we could exhibit a ring of smooth constants of motion inside $\Cinf(\cM_{d*}^\red)$
of functional dimension $\dim(\cM_{d*}^\red)-n$.
At present, we are able to construct such constants of motion after restriction to a smaller dense open subset
$\cM_{d**}^\red \subset \cM_{d*}^\red$.

It is worth stressing that the commuting reduced flows are multi-Hamiltonian since the reduced Hamiltonian coming from $H$  in \eqref{I8}
 generates the same flow by means of each
element of the pencil of compatible  Poisson structures on $\cM_{d*}^\red$  stemming from the bivectors $P_{\underline{z}}$ \eqref{I5}.

After the general study of the reduced master system, we inspect
the smallest symplectic leaf in $\cM_{d*}^\red$   given by the quotient space
 $\Phi^{-1}(e^{\ic \gamma} \1_n)/\UU(n)$, which indeed lies inside $\cM_{d*}^\red$ if $k \gamma \notin 2\pi \Z$ for any $k=1,\dots, n$.
 At the level of the reduced equations of motion, we demonstrate that this
  symplectic leaf  carries a new real form of the complex trigonometric spin RS system
 of Krichever and Zabrodin \cite{KZ}, for any $d\geq 2$.
 However,  we encountered unexpected technical difficulties in trying to prove the degenerate integrability of the reduced system
 restricted on this interesting symplectic leaf.
 Actually this is only one of  the challenging open problems that remain for future work.
 Note that for $d = 1$ the corresponding reduced systems were identified in \cite{FK,FKl} as (spinless)
 compactified trigonometric RS systems.

\subsubsection*{Layout} The goal of Section \ref{S:Back} consists in presenting the quasi-Poisson structure of the manifolds that are at the core of this paper.
We start by reviewing the notions of quasi-Poisson and quasi-Hamiltonian manifolds,
their relation, and the method of `exponentiation' of Hamiltonian  manifolds \cite{AMM,AKSM}.
The first main example, the internally fused double $\D(\UU(n))$ of $\UU(n)$, is recalled in Proposition \ref{Pr:qP-intDble}; we also present another key quasi-Poisson structure on the underlying smooth manifold $\UU(n)\times \UU(n)$  in Proposition \ref{Pr:qP-AltDble}.
We then exhibit as part of Proposition \ref{Pr:qP-Cn} the third main example: the quasi-Poisson ball $\disk(x)$, which depends on a parameter $x\in \R\setminus \{0\}$. It is constructed from the Hamiltonian manifold $\CC^n\simeq T^\ast \R^n$ by the exponentiation method; the quasi-Poisson manifold $\disk(x)$ admits a corresponding quasi-Hamiltonian structure which is exhibited in Proposition \ref{Pr:qSymp} based on a result of Hurtubise, Jeffrey and Sjamaar \cite{HJS}.

In Section \ref{S:Master}, we define the master phase space $\MM_d$ as a Hamiltonian quasi-Poisson manifold (with respect to the Lie group $\UU(n)$ and the moment map $\Phi$ from \eqref{I7}) by performing fusion of the internally fused double $\D(\UU(n))$ and $d\geq 1$ quasi-Poisson balls, see Theorem \ref{Thm:qH-Master}; the corresponding quasi-Hamiltonian structure is spelled out in Corollary \ref{cor:qHamMd}.
The quasi-Poisson bivector $P_{\MM_d}$ constructed in this way is used to build a pencil of Hamiltonian quasi-Poisson structures $(\MM_d,P_{\underline{z}},\Phi)$ depending on a tuple of parameters $\underline{z}\in \R^{d(d-1)/2}$ as part of Proposition \ref{Pr:Pencil}. We prove in Proposition \ref{Pr:om-Pencil} that for each quasi-Poisson bivector $P_{\underline{z}}$ there is a compatible quasi-Hamiltonian structure. An analogous result also holds for the extended cotangent bundle $M_d$ \eqref{I1}, see Remark \ref{Rem:Cotang}.

The dynamical considerations of the paper start in Section \ref{S:Integr-Md}.
We introduce the ring $\fH$ \eqref{B3} by pulling back the smooth invariant functions on $\UU(n)$ through the projection map $\cE_1:\MM_d\to \UU(n)$ on the first component as in \eqref{I8}. We explicitly write down the integral curves of the quasi-Hamiltonian vector fields associated with the functions in $\fH$ as part of Theorem \ref{Thm:B1}, and we collect the corresponding first integrals in a $\UU(n)$-equivariant smooth map $\Psi:\MM_d\to \MM_d$.
In Theorem \ref{Thm:B5} we note that $\Psi$ is a quasi-Poisson map for a natural but \emph{different} quasi-Poisson structure on the target $\MM_d$. Let us emphasize that this result holds for any choice of quasi-Poisson bivector $P_{\underline{z}}$ from the pencil constructed previously.
We finish Section \ref{S:Integr-Md} by explaining the notion of degenerate integrability, which fits the master system constructed in this section.

In Section \ref{S:Red}, we perform reduction of the master system through the action of $\UU(n)$ on $\MM_d$.
Denoting by $\MM_{d\ast}\subset \MM_d$ the subset on which $\UU(n)$ acts freely, we show in Proposition \ref{Pr:LR3} that a dense open
subset of $\MM_{d\ast}^{\red}:=\MM_{d\ast}/\UU(n)$ is filled by symplectic leaves of codimension $n$,
having the form $\MM_{d*}^\red(\cC) := (\Phi^{-1}(\cC)\cap \MM_{d*})/\UU(n)$
for a conjugacy class $\cC\subset \UU(n)_{\reg}$ of regular elements.
Next, we define the subset $\MM_{d\ast}^{2,\red}\subset \MM_{d\ast}^{\red}$ by requiring
that the second $\UU(n)$ component, $B$ in \eqref{I4},  belongs to $\UU(n)_\reg$, and introduce
 a  partial gauge fixing by diagonalizing
$B$.
We then determine the vector fields $Y_H$ on the resulting `gauge slice' that
induce the reduced Hamiltonian vector fields on $\MM_{d\ast}^{2,\red}$ of the elements $H\in \fH$, see  Proposition \ref{Pr:LR6}.
More importantly, we prove in Theorem \ref{Thm:DegInt} that a suitable quotient of
 the map $\Psi:\MM_d\to \MM_d$ yields enough constants of motion implying the degenerate integrability of the
 reduced system on a dense open subset $\MM_{d\ast \ast}^{\red}$ of the Poisson manifold $\MM_{d\ast}^{\red}$.
  We are also able to show that degenerate integrability is inherited on the largest symplectic leaves (of codimension $n$)
   inside $\MM_{d\ast\ast}^{\red}$, see Corollary \ref{Cor:IntSympl}.
We finish Section \ref{S:Red} by noting that the reduced spaces $\Phi^{-1}(\cC)/\UU(n)$ are compact if and only if $d=1$.

We investigate in Section \ref{S:spinRS} the smallest symplectic leaves $\Phi^{-1}(e^{\ic \gamma} \1_n)/\UU(n)$, where  $e^{\ic \gamma}$ is not a $k$-th root of unity for $k=1,\ldots,n$. By writing down the reduced equations of motion for two suitably chosen elements of $\fH$ in Theorem \ref{Thm:Dyn-RS}, we exhibit the connection between our model and the spin Ruijsenaars--Schneider systems of Krichever and Zabrodin \cite{KZ}.

The main text ends with Section \ref{S:ccl} where we gather some open questions for future work.
There are three appendices: Appendix \ref{sec:A} provides an alternative derivation of the quasi-Poisson ball
$\disk(x)$ based on an ansatz stemming from a \emph{holomorphic} quasi-Poisson manifold; Appendix \ref{sec:J}
collects auxiliary results needed in Section \ref{S:Red}; Appendix \ref{sec:M} contains a study of
a conjecturally important algebra of first integrals.

\medskip

Let us finish by  emphasizing our main results.
Firstly, we constructed a pencil of  Hamiltonian quasi-Poisson structures and corresponding quasi-Hamiltonian structures on $\MM_d$  (Propositions \ref{Pr:Pencil} and \ref{Pr:om-Pencil}).
Secondly, we proved that the master system living on $\MM_d$  induces a degenerate integrable system on a dense open subset of the reduced phase space $\MM_d^\red$ (Theorem \ref{Thm:DegInt}) and on its maximal symplectic leaves  (Corollary \ref{Cor:IntSympl}).
Thirdly, we have related the reduced system on special minimal symplectic leaves
to the spin Ruijsenaars--Schneider models of Krichever and Zabrodin \cite{KZ} (Section \ref{S:spinRS}).

\section{Quasi-Poisson and quasi-Hamiltonian geometry}
\label{S:Back}

We collect background material about quasi-Poisson and quasi-Hamiltonian geometry which is necessary to navigate throughout this paper.
We fix a real Lie group $G$ with Lie algebra $\g$ and exponential map $\exp:\g\to G$. Denote by $\operatorname{Ad}$  the adjoint actions of $G$ on itself and on its Lie algebra, while $\ad:\g\to \Hom(\g)$ is given by $\ad(\xi)(\zeta)=[\xi,\zeta]$ for all $\xi,\zeta\in \g$.
Introduce for any $\xi \in \g$ the left- and right-invariant vector fields $\xi^L$ and $\xi^R$ on $G$.
They act as derivations of the smooth functions $F\in \Cinf(G)$ according to
\begin{equation}
  \begin{aligned} \label{EqinfLR}
    \xi^L(F)(g)=\left.\frac{d}{dt}\right|_{t=0} \, F\left(g \cdot \exp(t \xi) \right)\,, \quad
\xi^R(F)(g)=\left.\frac{d}{dt}\right|_{t=0} \, F\left( \exp(t \xi) \cdot g \right)\,,
\quad \,\,g\in G\,.
  \end{aligned}
\end{equation}

Assume that $\g$ is endowed with an inner product $\ip{-,-}:\g\times\g\to \R$ which is $\Ad$-invariant, i.e.
$\ip{\Ad_g\xi ,\Ad_g\zeta}=\ip{\xi,\zeta}$ for all $\xi,\zeta\in \g$, $g\in G$. By differentiating the last identity, we get the $\ad$-invariance $\ip{\ad(\rho)(\xi) ,\zeta}+\ip{\xi,\ad(\rho)(\zeta)}=0$ for $\xi,\zeta,\rho\in \g$.
We also fix a basis $(e_a)_{a\in \mathtt{A}}$ of $\g$  orthonormal with respect to $\ip{-,-}$. Then define the Cartan trivector\footnote{Wedge products do not contain a normalization factor, e.g.
$e_1\wedge e_2 \wedge e_3=\sum_{\sigma \in S_3} \sgn(\sigma) e_{\sigma(1)}\otimes e_{\sigma(2)}\otimes e_{\sigma(3)},$ and all unadorned tensor products or wedge products are taken over $\R$.} $\phi \in\wedge^3\g$ by
\begin{equation} \label{Eq:Cartan3}
 \phi =\frac{1}{12} \sum_{a,b,c\in \mathtt{A}} C_{abc} \, e_a\wedge e_b \wedge e_c\,, \quad \text{where }C_{abc}=\ip{e_a,[e_b,e_c]}\,.
\end{equation}
The Cartan trivector is $\Ad$-invariant because the bilinear form is.
There is a $3$-form similar to the trivector $\phi$.
For $\theta^L$ and $\theta^R$ the left- and right-invariant Maurer-Cartan elements ($g^{-1} dg$ and $dg\,g^{-1}$ in matrix notations) which are $\g$-valued $1$-forms, we let $\theta_a^{L,R}:=\ip{e_a,\theta^{L,R}}$. We can then define the $3$-form
\begin{equation} \label{Eq:eta3}
 \eta=\frac{1}{12} \ip{\theta^R\stackrel{\wedge}{,}[\theta^R\stackrel{\wedge}{,}\theta^R]}
 =\frac{1}{12} \sum_{a,b,c\in \mathtt{A}}C_{abc}\, \theta^R_a \wedge \theta^R_b \wedge \theta^R_c\,.
\end{equation}
Note that $\eta$ is closed, invariant under the action by right multiplication and also invariant for the action by left multiplication (due to $ \eta=\frac{1}{12} \ip{\theta^L\stackrel{\wedge}{,}[\theta^L\stackrel{\wedge}{,}\theta^L]}$).

When we consider a smooth manifold $M$ equipped with a left action of $G$, the infinitesimal action of $\g$ on $M$ assigns to each  $\xi\in \g$ a vector field $\xi_M\in \Vect(M)$ over $M$ defined for any smooth function $f\in \Cinf(M)$ by
\begin{equation} \label{EqinfVectM}
 \xi_M(f)(m)=\left.\frac{d}{dt}\right|_{t=0} f(\exp(-t\xi)\cdot m)\,, \quad \,\,m\in M\,.
\end{equation}
The map $\xi\mapsto \xi_M$ extends for any $k\geq 1$  in such a way that it returns a $k$-vector field $\psi_M$ for any $\psi\in\wedge^k\g$. This extension is equivariant and preserves Schouten brackets.

Throughout the first two subsections, the reader can put $G=\UU(n)$ and $\g=\uu(n)$ as the unitary case is the main focus of our construction, see \S\ref{ss:Prel-Un}.
General results are taken from the original works of Alekseev, Kosmann-Schwarzbach and Meinrenken \cite{AKSM} or Alekseev, Malkin and Meinrenken \cite{AMM} regarding the quasi-Poisson and quasi-Hamiltonian setting, respectively.

\subsection{Definitions and fusion}

\subsubsection{Quasi-Poisson manifolds}

Let us fix a smooth manifold $M$ equipped with a left action of $G$. Any bivector field $P\in \wedge^2 \Vect(M)$ on $M$ induces an
antisymmetric $\R$-bilinear map $\br{-,-}$ on $M$ which is  a derivation in each argument through
$\br{f_1,f_2}:= P(df_1, df_2)$ for any smooth functions $f_1,f_2$. We call the operation $\br{-,-}:\Cinf(M)\times \Cinf(M)\to \Cinf(M)$ (and its restriction to subsets) a \emph{bracket}; we say that $\br{-,-}$ is the bracket associated with $P$.
By $\Cinf(M)$ we mean real functions.  The bracket extends naturally to complex functions in an $\C$-bilinear manner,
 and then it satisfies the reality condition $\br{\bar f_1,\bar f_2} = \overline{\br{f_1,f_2}}$.

\begin{definition}[\cite{AKSM}]
Denote by $\br{-,-}$ a $G$-invariant bracket on $M$.
We say that $\br{-,-}$ is a \emph{quasi-Poisson bracket} if for any  smooth functions $f_1,f_2,f_3\in \Cinf(M)$, we have
 \begin{equation} \label{Eq:JacPhi}
\br{f_1,\br{f_2,f_3}} +\br{f_2,\br{f_3,f_1}} + \br{f_3,\br{f_1,f_2}} = \frac12 \phi_M(f_1,f_2,f_3)\,,
 \end{equation}
where $\phi_M$ is induced using the infinitesimal action by the Cartan trivector \eqref{Eq:Cartan3}.
 The couple $(M,\br{-,-})$ is a \emph{quasi-Poisson manifold}.
\end{definition}
If $\br{-,-}$ corresponds to $P\in \wedge^2 \Vect(M)$, then \eqref{Eq:JacPhi} is equivalent to $[P,P]=\phi_M$, where $[-,-]$ denotes the Schouten bracket on $M$. If $G$ is Abelian, then $\phi=0$ and $\br{-,-}$ is a Poisson bracket on $M$ as the Jacobi identity holds by \eqref{Eq:JacPhi}.
Note also that the right-hand side of \eqref{Eq:JacPhi} vanishes if at least one of the functions $f_i$ is $G$-invariant, hence the subalgebra $\Cinf(M)^G\subset \Cinf(M)$ is equipped with a Poisson bracket.

\begin{definition}[\cite{AKSM}] \label{def:qHam}
Let $\Phi:M \to G$ be a smooth map intertwining the action of $G$ on $M$ with the adjoint action of $G$ on itself.
We say that $\Phi$ is a \emph{(Lie group valued) moment map} if, for any $F\in \Cinf(G)$, we have the following equality of vector fields on $M$:
\begin{equation} \label{Eq:momap}
  \br{F \circ \Phi,-}=\frac12 \sum_{a\in \mathtt{A}}\,\Phi^\ast\left((e_a^L+e_a^R)(F) \right) \,\,(e_a)_M\,.
\end{equation}
The triple $(M,\br{-,-},\Phi)$ is a \emph{Hamiltonian quasi-Poisson manifold}.
\end{definition}
The following operation allows one to build Hamiltonian quasi-Poisson manifolds recursively.

\begin{proposition}[Fusion product, \cite{AKSM}] \label{Pr:Fus}
Let $(M,\br{-,-}_M,\Phi_M)$ and $(N,\br{-,-}_N,\Phi_N)$ be Hamiltonian quasi-Poisson manifolds for some actions of $G$ on $M$ and $N$. Then the space
$$M\circledast N:=(M\times N, \br{-,-}^{\textrm{\emph{fus}}}_{M,N},\Phi_M \Phi_N)$$
is a Hamiltonian quasi-Poisson manifold for the diagonal action of  $G\hookrightarrow G\times G$ on $M\times N$, where
$\br{-,-}^{\textrm{\emph{fus}}}_{M,N}:=\br{-,-}_M+\br{-,-}_N-\br{-,-}_{P_{M,N}}$ using
 the bracket $\br{-,-}_{P_{M,N}}$ associated with the bivector field
\begin{equation}\label{Eq:Fus-Psi}
 P_{M,N}:=\frac12 \sum_{a\in \mathtt{A}} (e_a)_M\wedge (e_a)_N\,.
\end{equation}
\end{proposition}

Note that $M\circledast N$ and $N\circledast M$ have different quasi-Poisson structures and moment maps, but they are (non-trivially) isomorphic as Hamiltonian quasi-Poisson manifolds, see \cite[Prop. 5.7]{AKSM}.
Hence fusion is not commutative, but it is obviously associative.

\subsubsection{Quasi-Hamiltonian manifolds}

Fix a smooth manifold $M$ equipped with a left action of $G$. Recall the existence of the $3$-form $\eta$ on $G$ defined by \eqref{Eq:eta3}.

\begin{definition} [\cite{AMM}]  \label{def:om-qHam}
Let $\omega$ be a $G$-invariant $2$-form on $M$ and let $\Phi:M \to G$ be a smooth map intertwining the $G$-action on $M$ with the adjoint action of $G$ on itself.
We say that the triple $(M,\omega,\Phi)$ is a \emph{quasi-Hamiltonian manifold} (or \emph{quasi-Hamiltonian $G$-manifold}) and that $\Phi$ is its moment map if the following three conditions are satisfied:
\begin{itemize}
 \item[(B1)] $d\omega=\Phi^\ast \eta$;
 \item[(B2)] $\iota_{(e_a)_M}\omega=\frac12 \Phi^\ast (\theta_a^R+\theta_a^L)$;
  \item[(B3)] for all $m\in M$, $\operatorname{ker} \omega_m
  =\{\xi_{M,m} \mid \xi \in \operatorname{ker}(\operatorname{Ad}_{\Phi(m)}+\id_\g)\}$.
\end{itemize}
\end{definition}
Let us emphasize that we follow the convention of \cite{AKSM} which has an unimportant sign difference compared to the original definition of \cite{AMM}. The operation of fusion also exists in the quasi-Hamiltonian setting.

\begin{proposition}[Fusion product, \cite{AMM}] \label{Pr:Fus-omega}
Assume that $(M,\omega_M,\Phi_M)$ and $(N,\omega_N,\Phi_N)$ are quasi-Hamiltonian manifolds for some actions of $G$ on $M$ and $N$. Then the space
$$M\circledast N:=(M\times N, \omega^{\textrm{\emph{fus}}}_{M,N},\Phi_M \Phi_N)$$
is a quasi-Hamiltonian $G$-manifold for the diagonal action $G\hookrightarrow G\times G \curvearrowright M\times N$, where
\begin{equation}
\omega^{\textrm{\emph{fus}}}_{M,N}:=\omega_M+\omega_N- \frac12 \sum_{a\in \mathtt{A}} \Phi_M^\ast \theta_a^L \wedge \Phi_N^\ast \theta_a^R\,.
\end{equation}
\end{proposition}

\subsection{Relation between the quasi-Poisson and quasi-Hamiltonian cases} \label{ss:Exp}

Assume that $(M,\br{-,-}_0,\Phi_0)$ is a Hamiltonian Poisson manifold for an action of $G$.
This means that $\br{-,-}_0$ is a Poisson bracket associated with a bivector field $P_0\in \wedge^2\Vect(M)$,
and $\Phi_0:M\to \g$ is a \emph{moment map}  turning the  infinitesimal
action \eqref{EqinfVectM} of $\xi\in \g$ on $M$ into a Hamiltonian vector field:
$\xi_M=\br{\ip{\Phi_0,\xi},-}_0$ for all $\xi\in \g$.
Following \cite[Sec. 7]{AKSM}, there is a way to turn $\exp\circ \Phi_0:M\to \g$ into a Lie group
valued moment map as in Definition \ref{def:qHam}. Introduce the meromorphic function $\varphi$ on $\CC$  by
\begin{equation}
 \label{Eq:Mero-phi}
 \varphi(s):=\frac{1}{s}-\frac12 \coth\left(\frac{s}{2}\right)\,.
\end{equation}
This function is holomorphic at $0$ and away from the points $s\in 2\pi \ic \Z^\times$.
Next, on the open subset $\g^\natural\subset \g$ where the derivative of $\exp:\g\to G$ is invertible, we define the bivector field
\begin{equation}
 \label{Eq:biv-T}
 T=\frac12 \sum_{a,b\in \mathtt{A}} T_{ab} \,e_a\wedge e_b\in \mathcal{C}^\infty(\g^\natural,\g\wedge \g)\,, \qquad
 T_{ab}(\xi)=\ip{e_a, \varphi(\ad(\xi)) e_b}\quad  \forall \xi\in \g^\natural\,.
\end{equation}
The operator $\varphi(\ad(\xi))$ is well-defined, since $\g^\natural$ consists of the $\xi\in \g$ for which the set of the eigenvalues of the complexification of $\ad(\xi)$ does not intersect $2\pi \ic \Z^\times$, see e.g. \cite[p.25]{DK}.
For later use, let us remark the identities (the second one follows by $\ad$-invariance of $\ip{-,-}$)
\begin{equation}
 \label{Eq:Tab-id}
 \sum_{a\in \mathtt{A}} T_{ab}(\xi) e_a = \varphi(\ad(\xi))e_b\,, \qquad
 \sum_{b\in \mathtt{A}} T_{ab}(\xi) e_b = \varphi(-\ad(\xi)) e_a\,.
\end{equation}

\begin{theorem}[\cite{AKSM}] \label{thm:AKSM-exp}
 If $(M,\br{-,-}_0,\Phi_0)$ is a Hamiltonian Poisson manifold for an action of $G$ and $\Phi_0(M)\subset \g^\natural$,
 then $(M,\br{-,-},\Phi)$ is a Hamiltonian quasi-Poisson manifold for the same action of $G$ with
quasi-Poisson bracket $\br{-,-}$ associated with $P:=P_0-(\Phi_0^\ast T)_M$ and
the moment map $\Phi=\exp\circ \Phi_0$.
\end{theorem}

If the Poisson bracket $\br{-,-}_0$ on $M$ is non-degenerate, there is a $G$-invariant symplectic form $\omega\in \Omega^2(M)$ corresponding to
the defining bivector field $P_0$ through $P_0^\sharp \circ \omega_0^\flat =\id_{TM}$. Here and below, we use the conventions
\begin{equation} \label{Eq:conv-music}
 P^\sharp(\alpha)(\beta)=P(\alpha ,\beta), \quad
 \omega^\flat (X)=\iota_X \omega\,,\qquad \alpha,\beta\in \Omega^1(M),\quad X\in \Vect(M)\,.
\end{equation}
(Explicitly, contraction of a $2$-form reads $\iota_X(\alpha\wedge \beta)=\alpha(X)\, \beta - \beta(X)\,\alpha$.)
This makes $(M,\omega_0,\Phi_0)$ a Hamiltonian symplectic manifold.
The exponentiation process in that case follows the method from \cite[\S3.3]{AMM} which relies on the $2$-form
\begin{equation} \label{Eq:varpi}
\varpi:=\frac12 \int_0^1 \langle \exp_s^\ast \theta^R \stackrel{\wedge}{,} \frac{\partial}{\partial s} \exp_s^\ast \theta^R\rangle\,ds\,,\quad \text{ for }
\exp_s:\g\to G,\,\,\exp_s(\lambda)=\exp(s \lambda)\,.
\end{equation}
Note that $\varpi$ satisfies $d\varpi=-\Phi^\ast \eta$ with $\eta$ in \eqref{Eq:eta3}.
\begin{theorem}[\cite{AMM}] \label{thm:AMM-exp}
 If $(M,\omega_0,\Phi_0)$ is a Hamiltonian symplectic manifold for an action of $G$ and $\Phi_0(M)\subset \g^\natural$,
 then $(M,\omega,\Phi)$ is a quasi-Hamiltonian manifold for the same action of $G$ with
$\omega:=\omega_0 - \Phi_0^\ast \varpi$ and
the moment map $\Phi=\exp\circ \Phi_0$.
\end{theorem}
The structures of Hamiltonian quasi-Poisson and quasi-Hamiltonian $G$-manifold obtained
through Theorems \ref{thm:AKSM-exp} and \ref{thm:AMM-exp} are compatible \cite{AKSM} through the identity
\begin{equation} \label{Eq:corrPOm}
 P^\sharp \circ \omega^\flat = \id_{TM}- \frac14 \sum_{a\in \mathtt{A}} (e_a)_M \otimes \Phi^\ast(\theta_a^L-\theta_a^R)\,.
\end{equation}
Furthermore, the Hamiltonian quasi-Poisson manifold $(M,\br{-,-},\Phi)$ is \emph{non-degenerate}: for all $m\in M$,
$T_mM= \span_\R \{P_m^\sharp(\alpha_m),(\xi_M)_m \mid \alpha_m\in T^\ast_m M \text{ and } \xi \in \ker (\id_\g+\operatorname{Ad}_{\Phi(m)})\}$.
In full generality, the following fundamental result holds.
\begin{theorem}[\cite{AKSM}] \label{Thm:Corr}
If $(M,\br{-,-},\Phi)$ is a non-degenerate Hamiltonian quasi-Poisson $G$-manifold, there exists a unique $2$-form $\omega\in \Omega^2(M)$ such that $(M,\omega,\Phi)$ is a quasi-Hamiltonian $G$-manifold and the compatibility condition \eqref{Eq:corrPOm} holds for $P$ the bivector field defining $\br{-,-}$.

If $(M,\omega,\Phi)$ is a quasi-Hamiltonian $G$-manifold, there exists a unique bivector field $P\in \wedge^2\Vect(M)$ with associated bracket denoted $\br{-,-}$ such that $(M,\br{-,-},\Phi)$ is a non-degenerate Hamiltonian quasi-Poisson $G$-manifold and the compatibility condition \eqref{Eq:corrPOm} holds.
\end{theorem}

\begin{remark}
In the second part of Theorem \ref{Thm:Corr}, we can simply consider the existence of a unique bivector $P$ compatible with $\omega$ through the identity \eqref{Eq:corrPOm}. Indeed, it is known that the identity \eqref{Eq:corrPOm} guarantees that $(M,\br{-,-},\Phi)$ is non-degenerate.
\end{remark}

\subsection{Examples for \texorpdfstring{$\UU(n)$}{U(n)} actions}  \label{ss:Prel-Un}

We look at the case of the unitary group and its Lie algebra which are given by
\begin{equation*}
 G=\UU(n):=\{A\in \Gl(n,\CC) \mid AA^\dagger=\1_n\}\,, \quad
 \g=\uu(n):=\{A\in \gl(n,\CC) \mid A+A^\dagger=0_n\}\,.
\end{equation*}
Hereafter, $A^\dagger$ is for the conjugate transpose of any matrix $A\in \Mat(n\times m,\CC)$.
The adjoint action of $\UU(n)$ is given by conjugation of matrices.
We use the bilinear map
\begin{equation} \label{Eq:ipU}
 \ip{-,-}:\uu(n)\times \uu(n)\to \R\,, \quad \ip{\xi,\zeta}=\tr(\xi \zeta^\dagger)=-\tr(\xi \zeta)\,,
\end{equation}
as an invariant real inner product on $\uu(n)$. In particular, for any orthonormal basis $(e_a)_{a\in \mathtt{A}}$, we recover the well-known identity
\begin{equation} \label{Eq:shortEE}
 \sum_{a\in \mathtt{A}} e_a \otimes_\CC e_a = -\sum_{1\leq k,l \leq n} E_{kl}\otimes_\CC E_{lk}\,,
\end{equation}
where $E_{kl}\in \gl(n,\CC)$ is the elementary matrix such that $(E_{kl})_{ij}=\delta_{ik}\delta_{jl}$ for $1\leq i,j\leq n$.

\subsubsection{Internally fused double of \texorpdfstring{$\UU(n)$}{U(n)}} \label{ss:qDouble}

We recall a well-known quasi-Poisson structure on $\UU(n)\times \UU(n)= \{(A,B) \mid A,B\in \UU(n)\}$, colloquially called the internally fused double of $\UU(n)$.
Below, we define the complex-valued `coordinate functions'  $A_{ij},B_{ij}\in \Cinf(\UU(n)\times \UU(n),\CC)$ for $1\leq i,j\leq n$ by setting  $A_{ij}(m)=A'_{ij}$ and $B_{ij}(m)=B'_{ij}$ at  $m=(A',B')\in \UU(n)\times \UU(n)$.

\begin{proposition}[\cite{AMM,AKSM}, \emph{Internally fused double}]  \label{Pr:qP-intDble}
Consider $\UU(n)\times \UU(n)$ with the action of $\UU(n)$ by simultaneous conjugation:
$g\cdot (A,B)=(g A g^{-1} , g B g^{-1})$.
Then, the bracket which is determined by  the formulae
\begin{equation}\label{Eq:intDble}
\begin{aligned}
 \br{A_{ij} , A_{kl} } &= -\frac12 (A^2)_{kj} \delta_{il} + \frac12 \delta_{kj} (A^2)_{il}\,,  \\
\br{B_{ij} , B_{kl} } &=+\frac12 (B^2)_{kj} \delta_{il} - \frac12 \delta_{kj} (B^2)_{il}\,, \\
\br{A_{ij} , B_{kl} } &= -\frac12 \Big(\delta_{kj} (AB)_{il} + (BA)_{kj} \delta_{il} + B_{kj} A_{il} - A_{kj} B_{il}\Big)\,,
\end{aligned}
\end{equation}
is a quasi-Poisson bracket admitting $\Phi:(A,B)\mapsto ABA^{-1}B^{-1}$ as a moment map.
Furthermore, this Hamiltonian quasi-Poisson manifold is non-degenerate, and its compatible quasi-Hamiltonian structure is provided by
\begin{equation} \label{Eq:om-DUn}
\begin{aligned}
\hspace{-0.1cm}  \omega_{\D(\UU(n))}=& \frac12 \tr\!\left(
A^{-1} dA \!\wedge\! dB\, B^{-1}
\!+\! dA\, A^{-1} \!\wedge\! B^{-1} dB
\!-\!(AB)^{-1} d(AB) \!\wedge\! (BA)^{-1}d(BA)  \right).
\end{aligned}
\end{equation}
We denote both triples $(\UU(n)\times \UU(n),\br{-,-},\Phi)$ and $(\UU(n)\times \UU(n),\omega_{\D(\UU(n))},\Phi)$ by $\D(\UU(n))$.
\end{proposition}
\begin{proof}
The result follows from \cite[Example 10.5]{AKSM} by performing fusion of its Hamiltonian quasi-Poisson structure or its quasi-Hamiltonian structure, noting that the correspondence \eqref{Eq:corrPOm} is compatible with fusion by \cite[Proposition 10.7]{AKSM}. The explicit form of the quasi-Poisson bracket is directly obtained from the corresponding  bivector field
\begin{equation}  \label{Eq:P-DUn}
P_{\D(\UU(n))}=\frac12 \sum_{a\in \mathtt{A}} \Big( (e_a,0)^L \wedge (0,e_a)^R + (e_a,0)^R \wedge (0,e_a)^L
 - (e_a,0)_{\UU(n)\times \UU(n)} \wedge (0,e_a)_{\UU(n)\times \UU(n)} \Big)\,,
\end{equation}
where the infinitesimal action of $(\xi,\zeta)\in \uu(n)\times \uu(n)$ comes from
the action of $\UU(n)\times \UU(n)$ on itself by $(g,h)\cdot (A,B)=(g A h^{-1} , h B g^{-1})$.
\end{proof}

\subsubsection{A degenerate quasi-Poisson structure on $\UU(n) \times \UU(n)$}  \label{ss:alternativeDouble}

Later we shall also  need another quasi-Poisson structure on $\UU(n)\times \UU(n)$, which is given as follows.
\begin{proposition} \label{Pr:qP-AltDble}
Consider $\UU(n)\times \UU(n) = \{(A,\tilde{B}) \mid A,\tilde{B}\in \UU(n)\}$ with the action of $\UU(n)$ by simultaneous conjugation:
$g\cdot (A,\tilde{B})=(g A g^{-1} , g \tilde{B} g^{-1})$.
There is a quasi-Poisson bracket $\{\ ,\ \}_c$ on $\UU(n)\times \UU(n)$ that is characterized by the formulae
\begin{equation}\label{Eq:AltDble}
\begin{aligned}
 \br{A_{ij} , A_{kl} }_c &= -\frac12 (A^2)_{kj} \delta_{il} + \frac12 \delta_{kj} (A^2)_{il}\,,  \\
\br{\tilde{B}_{ij},\tilde{B}_{kl}}_c &=+\frac12(\tilde{B}^2)_{kj} \delta_{il} - \frac12\delta_{kj} (\tilde{B}^2)_{il}\,, \\
\br{A_{ij} , \tilde{B}_{kl} }_c &= -\frac1 2 \Big(\delta_{kj} (A\tilde{B})_{il} + (\tilde{B}A)_{kj} \delta_{il} - \tilde{B}_{kj} A_{il} - A_{kj} \tilde{B}_{il} \Big)\,,
\end{aligned}
\end{equation}
and this bracket admits $\tilde{\Phi}:(A,\tilde{B})\mapsto A\tilde{B}^{-1}$ as a moment map.
\end{proposition}

\begin{proof}
It is well-known that  $M:= \{A\in \UU(n)\}$ with the action of $\UU(n)$ by conjugation is a Hamiltonian quasi-Poisson manifold whose quasi-Poisson bracket is determined
by the first equation  in \eqref{Eq:AltDble} and its moment map is $\id_{\UU(n)}$, see \cite[Proposition 3.1]{AKSM}. Taking the opposite quasi-Poisson bracket results in  $N:= \{\tilde{B}\in \UU(n)\}$
 with the same action, the quasi-Poisson bracket given by the second equality in \eqref{Eq:AltDble}, and the moment map $\tilde{B}\mapsto \tilde{B}^{-1}$.

It is clear that the fusion product $M \circledast N$ ($\simeq \UU(n)\times \UU(n)$ as a manifold) constructed using Proposition \ref{Pr:Fus} admits $\tilde{\Phi}(A,\tilde{B})=A\tilde{B}^{-1}$ as its moment map.
For the quasi-Poisson bracket, it amounts to adding $-P_{M,N}$ \eqref{Eq:Fus-Psi} to
the quasi-Poisson bivectors of $M$ and $N$; this does not change the first two equalities in \eqref{Eq:AltDble}.
Finally, we use that for any $\xi\in \uu(n)$, $\xi_M(A_{ij})=-(\xi A)_{ij}+(A\xi)_{ij}$ and $\xi_M(\tilde{B}_{ij})=0$
(the same holds for $\xi_N$ under swapping $A$ and $\tilde{B}$) together with \eqref{Eq:shortEE} to compute that
\begin{equation*}
 \begin{aligned}
\br{A_{ij} , \tilde{B}_{kl} }_c &=-\frac12 \sum_{a\in \mathtt{A}} (e_a)_M(A_{ij}) (e_a)_N(\tilde{B}_{kl})\\
&=\frac12 \sum_{1\leq u,v\leq n} \left[-(E_{uv} A)_{ij}+(A E_{uv})_{ij}\right]  \left[-(E_{vu} \tilde{B})_{kl}+(\tilde{B} E_{vu})_{kl}\right] \\
&= \frac12 \Big(A_{kj} \tilde{B}_{il} + \tilde{B}_{kj} A_{il} -\delta_{kj} (A\tilde{B})_{il} - (\tilde{B}A)_{kj} \delta_{il} \Big)\,,
 \end{aligned}
\end{equation*}
which is the third equality in \eqref{Eq:AltDble}.
\end{proof}

\begin{remark} \label{Rem:Deg}
 The quasi-Poisson bivector $P_c$ associated with the structure defined on $M=\UU(n)\times \UU(n)$ in Proposition \ref{Pr:qP-AltDble} is \emph{not} non-degenerate. If it was, there would
  exist a $2$-form $\omega_c$ such that \eqref{Eq:corrPOm} holds with the moment map $\tilde \Phi$, and therefore the transposed equation
 \begin{equation} \label{Eq:corrPOm-T}
 \omega_c^\flat \circ P_c^\sharp = \id_{T^\ast M}- \frac14 \sum_{a\in \mathtt{A}} \tilde \Phi^\ast(\theta_a^L-\theta_a^R) \otimes (e_a)_M
\end{equation}
would hold as well.
However, we can easily check from \eqref{Eq:AltDble} that $\br{\tr(A^\ell),A_{ij}}=0$ and
$\br{\tr(A^\ell),\tilde{B}_{ij}}_c=0$ for any $1\leq i,j\leq n$ and $\ell \in \Z$, hence $P_c^\sharp(d\tr(A^\ell))=0$. Any function $\tr(A^\ell)$ is invariant, so that
$\xi_M(\tr(A^\ell))=0$ for any $\xi\in \uu(n)$ which contradicts \eqref{Eq:corrPOm-T}.
Incidentally, let us note that the functions $\tr(\tilde B^\ell)$ also belong to the center of the quasi-Poisson bracket $\br{-,-}_c$.
\end{remark}

\subsubsection{The quasi-Poisson ball by exponentiation}
\label{ss:Disk}

Consider $\CC^n$ with the action of $\UU(n)$ by left multiplication, i.e. $g\cdot v=gv$ for all $g\in \UU(n)$ and $v\in \CC^n$.
It is standard that $\CC^n\simeq \R^{2n}$ is a Hamiltonian Poisson manifold with
\begin{equation}
\br{v_i,v_k}_0=0,\,\, \br{\bar{v}_i,\bar{v}_k}_0=0, \,\, \br{v_i,\bar{v}_k}_0=\frac{\ic}{x}\delta_{ik}, \quad
\Phi_0:\CC^n\to \uu(n),\,\, \Phi_0(v)=\ic x vv^\dagger\,. \label{Eq:CnP-2}
\end{equation}
We can write explicitly its non-degenerate Poisson
bivector $P_0=\frac{\ic}{x} \sum_{j=1}^{n} \frac{\partial}{\partial v_j} \wedge \frac{\partial}{\partial \bar{v}_j}$
and the associated symplectic form
$\omega_0=\ic x \sum_{j=1}^{n} dv_j \wedge d\bar{v}_j$
(so $P_0^\sharp \circ \omega_0^\flat =\id_{TM}$ using \eqref{Eq:conv-music}).
Below, we present the corresponding Hamiltonian quasi-Poisson structure obtained by the method of exponentiation given in \S\ref{ss:Exp}.
The approach that we originally used to find this structure is described in Appendix \ref{sec:A}.
To state the result, introduce from \eqref{Eq:Mero-phi} the following function
\begin{equation} \label{Eq:ab-fct}
\btt(t)= 2\ic \varphi(-\ic t)= \cot\left(\frac{t}{2}\right)-\frac{2}{t}\,,
\end{equation}
which is real analytic on $\{t\in \R \mid |t|<2\pi\}$ and each $\{t\in \R \mid 2k\pi<|t|<2(k+1)\pi\}$ for $k=1,2,\ldots$
The next two propositions are respectively proved in \S\ref{ss:Proof1} and \S\ref{ss:Proof2}.

\begin{proposition} \label{Pr:qP-Cn}
Fix $x\in \R\setminus \{0\}$, and let\footnote{Our notation is similar to the one used in \cite[Lemma B.1]{HJS}, where such a space is called a ``disc''. To avoid confusion, we refer to these spaces as (open) balls since they are not products of discs if $n>1$.} $\disk(x):=\{v\in \CC^n \mid |v|^2 < \frac{2\pi}{|x|}\}\subset \CC^n\simeq \R^{2n}$.
Then $\disk(x)$ endowed with the action of $\UU(n)$ by left multiplication admits a structure of non-degenerate Hamiltonian quasi-Poisson manifold obtained by exponentiation. Explicitly, for $\btt(t)$ given in \eqref{Eq:ab-fct}, the quasi-Poisson bracket satisfies for any $1\leq i,k\leq n$
\begin{equation}
 \br{v_i,v_k}_{\disk(x)}=0\,, \quad \br{\bar{v}_i,\bar{v}_k}_{\disk(x)}=0\,, \quad
\br{v_i,\bar{v}_k}_{\disk(x)}
=\frac{\ic}{x} \delta_{ik}
+ \frac{\ic}{2} \btt(x |v|^2) \left[ |v|^2 \delta_{ik} - v_i\bar{v}_k \right] .
\label{Eq:qPB-C2}
\end{equation}
The moment map is given by $\Phi:\disk(x)\to \UU(n),\,\, \Phi(v)=\exp(\ic x vv^\dagger)$.
\end{proposition}

\begin{proposition} \label{Pr:qSymp}
 The structure of quasi-Hamiltonian manifold on $\disk(x)$ corresponding uniquely as in Theorem \ref{Thm:Corr} to the structure of Hamiltonian quasi-Poisson manifold from Proposition \ref{Pr:qP-Cn} is given by the same moment map $\Phi:\disk(x)\to \UU(n)$ with the $2$-form
 \begin{equation}
 \begin{aligned} \label{Eq:qSymp}
  \omega_{\disk(x)}
&=\ic \frac{\sin(x |v|^2)}{|v|^2} \sum_{j=1}^n dv_j \wedge d\bar{v}_j
+ \left(\ic x - \frac{\ic \sin(x |v|^2)}{|v|^2} \right) \sum_{j=1}^n \frac{|v_j|^2}{|v|^2} dv_j \wedge d\bar{v}_j \\
&\quad + \left(\ic x - \frac{\ic \sin(x |v|^2)}{|v|^2} \right)
\sum_{j \neq k} \left(\frac{\bar{v}_j \bar{v}_k}{|v|^2} dv_j \wedge dv_k + \frac{v_jv_k}{|v|^2} d\bar{v}_j \wedge d\bar{v}_k
+ \frac{\bar{v}_jv_k}{|v|^2} d v_j \wedge d\bar{v}_k \right).
 \end{aligned}
 \end{equation}
\end{proposition}

\begin{remark} \label{Rem:qP-disk}
The bivector that defines the quasi-Poisson bracket \eqref{Eq:qPB-C2} is given by
\begin{equation} \label{Eq:qBiv}
P_{\disk(x)}=\frac{\ic}{2x} [2+x|v|^2\btt(x |v|^2)]\sum_{1\leq j\leq n} \frac{\partial}{\partial v_j}\wedge \frac{\partial}{\partial \bar{v}_j}
 - \frac{\ic}{2}\btt(x |v|^2) |v|^2 \sum_{1\leq j,k\leq n}\frac{\partial}{\partial v_j}\wedge \frac{\partial}{\partial \bar{v}_k} \,.
\end{equation}
Furthermore, the results of Propositions \ref{Pr:qP-Cn} and \ref{Pr:qSymp} hold on any `annulus' $\operatorname{A}_k(x)=\{v\in \CC^n \mid \frac{2k\pi}{|x|}<|v|^2<\frac{2(k+1)\pi}{|x|}\}$ with $k=1,2,\ldots$, or any union of these annuli and the ball $\disk(x)$.
This follows because the non-zero eigenvalues of (the complexification of) $\ad(\Phi_0(v))$ are $\pm \ic x|v|^2$, and therefore $\Phi_0(v)\in \uu(n)^\natural$ if and only if $v$ belongs either to $\disk(x)$ or to $\operatorname{A}_k(x)$ for some $k$.
Let us finally note that we have checked explicitly the correspondence \eqref{Eq:corrPOm} for $P_{\disk(x)}$ \eqref{Eq:qBiv} and $\omega_{\disk(x)}$ \eqref{Eq:qSymp}; it suffices to do this at the points of the form $v=(z,0,\ldots,0)$.
\end{remark}

\subsubsection{Proof of Proposition \ref{Pr:qP-Cn}} \label{ss:Proof1}

As noted in Remark \ref{Rem:qP-disk},  $\Phi_0$ restricted to $\disk(x)$ takes value in $\uu(n)^\natural$, hence we can use Theorem \ref{thm:AKSM-exp}.
The only non-trivial task consists in checking that the bivector $P_0-(\Phi_0^\ast T)_{\disk(x)}$ gives the desired quasi-Poisson bracket in the form \eqref{Eq:qPB-C2}.
Here, $T$ is defined in \eqref{Eq:biv-T} and the infinitesimal action $\zeta_{\disk(x)}$ of any $\zeta\in \uu(n)$ comes from the action of $\UU(n)$ by left multiplication on $\disk(x)$, which satisfies $\zeta_{\disk(x)}(v_i)=-(\zeta v)_i$ and $\zeta_{\disk(x)}(\bar{v}_i)=(v^\dagger\zeta)_i$.

Fix $\zeta\in \uu(n)$.
We start by computing the left multiplication of $v\in \Cinf(\disk(x),\CC^n)$ by  $\varphi(\ad(\Phi_0))(\zeta)\in \Cinf(\disk(x),\uu(n))$.
Since for any $n\geq 1$ and $\xi\in \uu(n)$ we have the binomial identity $\ad(\xi)^n(\zeta)=\sum_{k=0}^n (-1)^k \binom{n}{k} \xi^{n-k} \zeta \xi^k$, we can write
\begin{equation*}
  \ad(\Phi_0)^n(\zeta)v
=- (-\ic x |v|^2)^{n} \frac{vv^\dagger}{|v|^2}\,\zeta v
 + (-\ic x |v|^2)^n  \zeta v.
\end{equation*}
Using the analytic expansion of $\varphi$ at $0\in (-2\pi \ic,2\pi \ic)$ noting that $\varphi(0)=0$, we can write in terms of the function $\btt(t)$ defined by \eqref{Eq:ab-fct} that
\begin{equation} \label{Eq:varphi-v}
 \varphi(\ad(\Phi_0))(\zeta)\,v
=-\varphi(-\ic x |v|^2) \frac{vv^\dagger}{|v|^2}\, \zeta v + \varphi(-\ic x |v|^2) \, \zeta v
=\frac{\ic}{2} \btt(x |v|^2)  \left( \frac{v^\dagger \zeta v}{|v|^2} v \,-\, \zeta v\right)\,.
\end{equation}

Next, we note from \eqref{Eq:Tab-id} and the fact that $\varphi$ is odd that we can write the bivector obtained from the infinitesimal action of $\Phi_0^\ast T$ in the following form
\begin{equation*}
 (\Phi_0^\ast T)_{\disk(x)}
 = \sum_{b\in \mathtt{A}}\, (\varphi(\ad(\Phi_0^L))e_b)_{\disk(x)}\otimes (e_b)_{\disk(x)}\,.
\end{equation*}
We denote by $\br{-,-}_{\exp}$ the associated bracket, and we want to see that $\br{-,-}_0-\br{-,-}_{\exp}$ (where $\br{-,-}_0$ is the Poisson bracket from \eqref{Eq:CnP-2}) coincides with the quasi-Poisson bracket $\br{-,-}_{\disk(x)}$ stated in \eqref{Eq:qPB-C2}.
On the coordinate functions $(v_i,v_k)$, $1\leq i,k\leq n$, we have by successively applying \eqref{Eq:varphi-v} and  \eqref{Eq:shortEE}
\begin{equation*}
\begin{aligned}
 \br{v_i,v_k}_{\exp}
 &=\frac{\ic}{2} \btt(x |v|^2) \sum_{b\in \mathtt{A}}\,  \left( \frac{v^\dagger e_b v}{|v|^2} v_i \,-\, (e_b v)_i\right) (e_b v)_k \\
&=-\frac{\ic}{2} \btt(x |v|^2) \sum_{1\leq j,l\leq n}\,  \left( \frac{v^\dagger E_{jl} v}{|v|^2} v_i \,-\, (E_{jl} v)_i\right) (E_{lj} v)_k=0\,.
\end{aligned}
\end{equation*}
Hence $\br{v_i,v_k}_0-\br{v_i,v_k}_{\exp}=0$, in agreement with the first identity in \eqref{Eq:qPB-C2}.
We get $\br{\bar v_i,\bar v_k}_0-\br{\bar v_i,\bar v_k}_{\exp}=0$ by a similar argument or by reality of the bracket.
In the same way, we obtain
\begin{equation*}
\begin{aligned}
 \br{v_i,\bar{v}_k}_{\exp}
 &=-\frac{\ic}{2} \btt(x |v|^2) \sum_{b\in \mathtt{A}}\,  \left( \frac{v^\dagger e_b v}{|v|^2} v_i \,-\, (e_b v)_i\right) (v^\dagger e_b)_k  =-\frac{\ic}{2} \btt(x |v|^2) (\delta_{ik} |v|^2- v_i \bar{v}_k)\,.
\end{aligned}
\end{equation*}
We directly see that $\br{v_i,\bar{v}_k}_0-\br{v_i,\bar{v}_k}_{\exp}$ agrees with \eqref{Eq:qPB-C2}.  \qed

\subsubsection{Proof of Proposition \ref{Pr:qSymp}} \label{ss:Proof2}

The quasi-Hamiltonian structure is obtained by exponentiation using Theorem \ref{thm:AMM-exp} from the symplectic form $\omega_0$ given at the beginning of \S\ref{ss:Disk}. The fact that it corresponds to the quasi-Poisson structure from Proposition \ref{Pr:qP-Cn} is a consequence of Theorem \ref{thm:AMM-exp} (in fact, the correspondence \eqref{Eq:corrPOm} implies that the bivector $P$ is non-degenerate).
To compute $\omega_{\disk(x)}$, we build on the analogous derivation\footnote{The reference \cite{HJS} follows the conventions of \cite{AMM} regarding quasi-Hamiltonian manifolds, while we follow \cite{AKSM}. This results in a sign change in addition to a difference of factor in the chosen inner product on $\uu(n)$.} due to Hurtubise, Jeffrey and Sjamaar \cite[Lemma B.1]{HJS}.

We start with the following general result (see e.g. \cite[\S2]{HJS}) regarding the $2$-form $\varpi$ defined in \eqref{Eq:varpi} : for any $\lambda\in \uu(n)$ and $\xi,\zeta\in T_\lambda\uu(n)$, we have
 \begin{equation}  \label{Eq:varpi1}
  \varpi_\lambda(\xi,\zeta) = \left\langle\frac{\ad(\lambda) - \sinh(\ad(\lambda))}{\ad(\lambda)^2}\xi, \zeta\right\rangle\,.
 \end{equation}
In the last equation, the function appearing in the inner product admits the analytic expansion
\begin{equation} \label{Eq:varpi-pf2}
 g(z):=\frac{z-\sinh(z)}{z^2}=-\sum_{k\geq 1} \frac{z^{2k-1}}{(2k+1)!}\,,
\end{equation}
and we can write at $\lambda=\Phi_0(v)=\ic x vv^\dagger$ for $\xi\in T_\lambda \uu(n)$
\begin{equation*}
g(\ad(\ic x vv^\dagger))\xi =-  \sum_{k\geq 1} \frac{1}{(2k+1)!} \sum_{j=0}^{2k-1} \binom{2k-1}{j} (\ic x vv^\dagger)^j \xi (-\ic x vv^\dagger)^{2k-j-1} \,.
 \end{equation*}
We then find after elementary calculations that $g(\ad(\ic x vv^\dagger))\xi = g(\ic x |v|^2)\frac{1}{|v|^2} (- \xi vv^\dagger +  vv^\dagger \xi )$. Combined with \eqref{Eq:varpi1}, this yields
\begin{equation*}
\varpi_{\Phi_0(v)}(\xi,\zeta) = -\frac{g(\ic x |v|^2)}{|v|^2} \, \tr(vv^\dagger [\xi,\zeta])\,.
\end{equation*}
To get $(\Phi_0^\ast \varpi)_v$, we need to consider this formula for $\xi, \zeta$ of the form $(d\Phi_0)_vw$ with $w$ a tangent vector at $v\in \disk(x)$.
Noting that $(d\Phi_0)_v=\ic x d(vv^\dagger)$
together with the identity $g(\ic x |v|^2)=\frac{-\ic}{x |v|^2} + \frac{\ic}{x^2 |v|^4} \sin(x |v|^2)$ obtained from  \eqref{Eq:varpi-pf2}, we get
 \begin{equation*}
(\Phi_0^\ast\varpi)_v=\,-\left(\ic x -\ic \frac{\sin(x |v|^2)}{|v|^2} \right)\,\, \tr\left(\frac{vv^\dagger}{|v|^4} \,d(vv^\dagger)\wedge d(vv^\dagger)\right)\,.
 \end{equation*}
This can be used to write $\omega_{\disk(x)}=\omega_0-\Phi_0^\ast\varpi$ as
\begin{equation}
\omega_{\disk(x)}
=\ic x \sum_{j=1}^n dv_j \wedge d\bar{v}_j + \left(\ic x - \frac{\ic \sin(x |v|^2)}{|v|^2} \right)
\tr\left[ \frac{vv^\dagger}{|v|^4} (dv\, v^\dagger + v dv^\dagger) \wedge (dv\, v^\dagger + v dv^\dagger) \right] \,,
\end{equation}
and after simplifications this is \eqref{Eq:qSymp}.  \qed

\section{Master phase space}
\label{S:Master}

Here, we first present the explicit form of a Hamiltonian quasi-Poisson structure on the master phase space $\MM_d$
that is obtained  by performing fusion of the  double $\D(\UU(n))$ from
Proposition \ref{Pr:qP-intDble} with $d\geq 1$ copies of the Hamiltonian quasi-Poisson ball $\disk(x)$ from Proposition \ref{Pr:qP-Cn}.  Then, for $d\geq 2$, we construct a whole pencil of quasi-Poisson brackets  admitting
 \emph{the same} moment map $\Phi:\MM_d\to \UU(n)$ given in \eqref{Eq:Mast-phi-1}.
We prove that any quasi-Poisson bracket from the pencil is non-degenerate by  exhibiting the corresponding quasi-Hamiltonian structure.

\subsection{Main quasi-Poisson structure}

For $d\geq1$, let us fix $x_1,\ldots,x_d\in \R\setminus\{0\}$ which are the parameters of $d$ quasi-Poisson balls.

\begin{theorem} \label{Thm:qH-Master}
The fusion space
\begin{equation} \label{Eq:Mast-Md}
 \MM_d:=\MM_d(x_1,\ldots,x_d):=\D(\UU(n))\circledast \disk(x_1)\circledast \cdots \circledast \disk(x_d)
\end{equation}
parameterized by $(A,B)\in \D(\UU(n))$ and $v_\alpha \in \disk(x_\alpha)$, $1\leq \alpha \leq d$,
with the $\UU(n)$-action
 \begin{equation} \label{Eq:act-Mast}
 \mathcal{A}:\UU(n)\times \MM_d \to \MM_d\,, \quad
g\cdot (A,B,v_1,\ldots,v_d)=(gAg^{-1},gBg^{-1},g v_1,\ldots,g v_d)\,,
 \end{equation}
admits a structure of Hamiltonian quasi-Poisson manifold for the moment map
\begin{equation}
\Phi:\MM_d\to \UU(n), \quad
\Phi(A,B,v_\alpha)=ABA^{-1}B^{-1}\exp(\ic x_1 v_1v_1^\dagger)\cdots \exp(\ic x_d v_dv_d^\dagger)\,. \label{Eq:Mast-phi-1}
\end{equation}
The corresponding quasi-Poisson bracket satisfies for $1\leq i,j,k,l\leq n$,
\begin{equation}
 \begin{aligned}  \label{Eq:br-AB}
\br{A_{ij} , A_{kl} } =& -\frac12 (A^2)_{kj} \delta_{il} + \frac12 \delta_{kj} (A^2)_{il}\,,  \quad
\br{B_{ij} , B_{kl} } = +\frac12 (B^2)_{kj} \delta_{il} - \frac12 \delta_{kj} (B^2)_{il}\,,  \\
\br{A_{ij} , B_{kl} } =& -\frac12 \Big(\delta_{kj} (AB)_{il} + (BA)_{kj}\delta_{il} + B_{kj}A_{il} - A_{kj}B_{il} \Big)\,,
 \end{aligned}
\end{equation}
as well as for any $1\leq \alpha \leq d$,
\begin{equation}
 \begin{aligned}  \label{Eq:br-ABv}
\br{A_{ij},(v_\alpha)_{k}}= \frac12 A_{kj} (v_\alpha)_i - \frac12 \delta_{kj} (Av_\alpha)_i \,, \quad
\br{B_{ij},(v_\alpha)_{k}}= \frac12 B_{kj} (v_\alpha)_i - \frac12 \delta_{kj} (Bv_\alpha)_i  \,,
 \end{aligned}
\end{equation}
and for any $1\leq \alpha,\beta \leq d$,
\begin{equation} \label{Eq:br-VVall}
 \begin{aligned}
\br{(v_\alpha)_i,(v_\beta)_{k}}=& \frac12 \sgn(\beta-\alpha)\,   (v_\alpha)_k (v_\beta)_i  \,, \quad
\br{(\bar{v}_\alpha)_j,(\bar{v}_\beta)_{l}}= \frac12 \sgn(\beta-\alpha)\,   (\bar{v}_\alpha)_l (\bar{v}_\beta)_j\,, \\
\br{(v_\alpha)_i,(\bar{v}_\beta)_l}=&
\frac{\ic}{x_\alpha}\delta_{\alpha\beta} \delta_{il}
+ \frac{\ic}{2}\delta_{\alpha\beta} \btt(x_\alpha |v_\alpha|^2) \left[ |v_\alpha|^2 \delta_{il} - (v_\alpha)_i (\bar{v}_\beta)_l \right]
-\frac12 \sgn(\beta-\alpha)\,  \delta_{il} (v_\beta^\dagger v_\alpha) \,,
 \end{aligned}
\end{equation}
where for any $N\in \Z$ we set $\sgn(N)=+1$ if $N>0$, $\sgn(N)=-1$ if $N<0$, and $\sgn(0)=0$.
\end{theorem}

\begin{proof}
We prove the result by induction on $d$. Allowing $d=0$ as a base case, the statement reduces to the structure of the internally fused double $\D(\UU(n))$ given in Proposition \ref{Pr:qP-intDble}.
Then, we apply the fusion method of Proposition \ref{Pr:Fus} to go
from $\MM_{d-1}\times \disk(x_d)$ to $\MM_d=\MM_{d-1}\circledast \disk(x_d)$.
Indeed, we have the desired moment map $\Phi$ \eqref{Eq:Mast-phi-1} when we take
$M=\MM_{d-1}$, $N=\disk(x_d)$ and $G=\UU(n)$ in Proposition \ref{Pr:Fus}.
So we only need to check the formula for the quasi-Poisson bracket.
By Proposition \ref{Pr:Fus}, it is given by taking the sum of the quasi-Poisson brackets on $\MM_{d-1}$ (which is in the statement of the theorem by induction) and on $N=\disk(x_d)$ (given by \eqref{Eq:qPB-C2} with $v=v_d$ and $x=x_d$) to which we add the bracket corresponding to the bivector
\begin{equation} \label{Eq:Mast-Psi}
-\psi_{\MM_{d-1},\disk(x_d)}= -\frac12 \sum_{a\in \mathtt{A}} (e_a)_{\MM_{d-1}}\wedge (e_a)_{\disk(x_d)} \,.
\end{equation}
Here, $(e_a)_{a\in \mathtt{A}}$ is an orthonormal basis of $\uu(n)$.
Thus,  we only need to check those brackets in the statement that involve the entries of $v_d$ (or $\bar{v}_d$) and the entries of one of the matrices forming $(A,B,v_1,\ldots,v_{d-1})\in \MM_{d-1}$, which arise from the added bivector \eqref{Eq:Mast-Psi}.

Thanks to \eqref{Eq:act-Mast}, the infinitesimal action of $\xi\in \uu(n)$ on $\MM_{d-1}$ is such that for $1\leq \alpha < d$,
\begin{equation} \label{Eq:Mast-InfAct}
\xi_{\MM_{d-1}}(A_{ij})=(A \xi)_{ij}-(\xi A)_{ij}\,, \quad
\xi_{\MM_{d-1}}((v_\alpha)_k)=-(\xi v_\alpha)_k\,, \quad
\xi_{\MM_{d-1}}((\bar{v}_\alpha)_k)=(v_\alpha^\dagger \xi)_k \,.
\end{equation}
The formulas for each $B_{ij}$ can be used verbatim by replacing $A$ with $B$.
The infinitesimal action on $\disk(x_d)$ is similar:
$\xi_{\disk(x_d)}((v_d)_k)=-(\xi v_d)_k$ and $\xi_{\disk(x_d)}((\bar{v}_d)_k)=(v_d^\dagger \xi)_k$.
It remains to compute the fusion brackets corresponding to the bivector \eqref{Eq:Mast-Psi}, which is easily done when one uses the identity  \eqref{Eq:shortEE}.

To establish the first equality in \eqref{Eq:br-ABv} with $\alpha=d$, we compute
\begin{equation}
\begin{aligned}
 \br{A_{ij},(v_d)_k}=&-\psi_{\MM_{d-1},\disk(x_d)}(A_{ij},(v_d)_k)
 =\frac12 \sum_{1\leq b,c\leq n}  [(A E_{bc})_{ij}-(E_{bc} A)_{ij}]\, [-(E_{cb} v_d)_k] \\
 =&- \frac12 \left[\sum_{1\leq b \leq n}\delta_{kj} A_{ib} (v_d)_b - A_{kj} (v_d)_i\right]
 =\frac12 A_{kj} (v_d)_i - \frac12 \delta_{kj} (A v_d)_i\,.
\end{aligned}
\end{equation}
This is the required identity, and the second equality in \eqref{Eq:br-ABv} is obtained by replacing $A$ with $B$ in the computations.
To get the first equality in \eqref{Eq:br-VVall} with $1\leq \alpha<d$ and $\beta=d$, we calculate
\begin{equation}
 \br{(v_\alpha)_i,(v_d)_k}
 =\frac12 \sum_{1\leq b,c\leq n}  (E_{bc} v_\alpha)_i\, (E_{cb} v_d)_k
 = \frac12 (v_\alpha)_k (v_d)_i\,.
\end{equation}
This agrees with the desired equality, and the case $\alpha=d$ and $1\leq \beta<d$ follows by antisymmetry.
The second equality in \eqref{Eq:br-VVall} follows from a similar computation, or by reality of the bracket. We finally have for  $1\leq \alpha<d$
\begin{equation}
 \br{(v_\alpha)_i,(\bar{v}_d)_l}
=-\frac12 \sum_{1\leq b,c\leq n}  (E_{bc} v_\alpha)_i\, (v^\dagger_d E_{cb})_l
 = - \frac12 \delta_{il} (v_d^\dagger v_\alpha)\,.
\end{equation}
This is the third equality in \eqref{Eq:br-VVall} with $\alpha<\beta=d$. A similar computation (or reality of the bracket) yields the case $d=\alpha>\beta$.
\end{proof}

\begin{remark} \label{Rem:Mast-Biv}
Due to the unitarity of $A,B$ and the reality of the quasi-Poisson bracket on $\MM_d$ defined in Theorem \ref{Thm:qH-Master}, let us remark that  \eqref{Eq:br-ABv} yields
\begin{equation} \label{Eq:br-ABvbar}
 \begin{aligned}
\br{A_{ij},(\bar{v}_\alpha)_{l}}= \frac12 A_{il} (\bar{v}_\alpha)_j - \frac12 \delta_{il} (v_\alpha^\dagger A)_j \,, \quad
\br{B_{ij},(\bar{v}_\alpha)_{l}}= \frac12 B_{il} (\bar{v}_\alpha)_j - \frac12 \delta_{il} (v_\alpha^\dagger B)_j  \,.
 \end{aligned}
\end{equation}
In fact, we can explicitly write down the quasi-Poisson bivector $P_{\MM_d}$ obtained in Theorem \ref{Thm:qH-Master}.
If $P_{\D(\UU(n))}$ and $P_{\disk^\alpha}$ denote the quasi-Poisson bivectors on $\D(\UU(n))$ and the $\alpha$-th ball $\disk(x_\alpha)$ constituting $\MM_d$ for $1\leq \alpha \leq d$ as given in Propositions \ref{Pr:qP-intDble} and \ref{Pr:qP-Cn}, respectively, we see from \eqref{Eq:Mast-Psi} and the inductive proof of Theorem \ref{Thm:qH-Master} that we can write
\begin{equation}  \label{Eq:qPMd}
P_{\MM_d}= P_{\D(\UU(n))}+ \sum_{1\leq \alpha \leq d} P_{\disk^\alpha}
-\frac12 \sum_{1\leq \alpha \leq d} \sum_{a\in \mathtt{A}} (e_a)_{\MM_{\alpha-1}}\wedge (e_a)_{\disk^\alpha}\,.
\end{equation}
Here $(e_a)_{a\in \mathtt{A}}$ is an orthonormal basis of $\uu(n)$ and for any $\xi\in \uu(n)$,
$\xi_{\disk^\alpha}$ is the infinitesimal action obtained from left multiplication on the $\alpha$-th ball, while
$\xi_{\MM_{\alpha-1}}\in \Vect(\MM_d)$ is a short-hand notation for the infinitesimal action obtained from the $\UU(n)$-action
 \begin{equation} \label{Eq:act-Res}
\begin{aligned}
  &\mathcal{A}^{\res}_{\alpha-1}:\UU(n)\times \MM_d \to \MM_d\,, \\
&g\cdot (A,B,v_1,\ldots,v_d)=(gAg^{-1},gBg^{-1},g v_1,\ldots,g v_{\alpha-1},v_{\alpha},\ldots, v_d)\,.
\end{aligned}
 \end{equation}
This is the restriction of the action \eqref{Eq:act-Mast} to $\D(\UU(n))$ and the first $(\alpha-1)$ balls, which only acts by simultaneous conjugation on $\D(\UU(n))$ for $\alpha=1$. The infinitesimal action of $\mathcal{A}^{\res}_{\alpha-1}$ can be split into the infinitesimal actions associated with simultaneous conjugation on $\D(\UU(n))$ and left multiplication on the first $(\alpha-1)$ balls. Therefore, we can present $P_{\MM_d}$ as
\begin{equation}\label{Eq:qPMd-b}
 \begin{aligned}
P_{\MM_d}=& P_{\D(\UU(n))}+P_{\text{mix}}+ P_{\text{spin}}  \,, \qquad \text{ where }\\
P_{\text{mix}}=&-\frac12 \sum_{1\leq \alpha \leq d} \sum_{a\in \mathtt{A}} (e_a)_{\D(\UU(n))}\wedge (e_a)_{\disk^\alpha}\,, \\
P_{\text{spin}}=& \sum_{1\leq \alpha \leq d} P_{\disk^\alpha}
-\frac12 \sum_{1\leq \alpha<\beta \leq d} \sum_{a\in \mathtt{A}} (e_a)_{\disk^\alpha}\wedge (e_a)_{\disk^\beta}\,.
 \end{aligned}
\end{equation}
\end{remark}

\begin{corollary} \label{cor:qHamMd}
 The Hamiltonian quasi-Poisson manifold $\MM_d$ presented in Theorem \ref{Thm:qH-Master} is non-degenerate.
 The corresponding quasi-Hamiltonian structure is given by $(\MM_d,\omega_{\MM_d},\Phi)$ for
\begin{equation} \label{Eq:om-Md}
 \omega_{\MM_d}:= \omega_{\D(\UU(n))}+ \sum_{1\leq \alpha \leq d} \omega_{\disk^\alpha}
-\frac12 \sum_{1\leq \alpha \leq d} \sum_{a\in \mathtt{A}} \Phi^\ast_{\MM_{\alpha-1}}\theta_a^L\wedge \Phi^\ast_{\alpha}\theta_a^R\,.
\end{equation}
Here, $\omega_{\D(\UU(n))}$ is defined in \eqref{Eq:om-DUn}, $\omega_{\disk^\alpha}$ is given by \eqref{Eq:qSymp} on the $\alpha$-th ball for any $1\leq \alpha \leq d$, and we set
\begin{equation*}
\Phi_{\MM_{\alpha-1}}= ABA^{-1}B^{-1}\exp(\ic x_1 v_1v_1^\dagger)\cdots \exp(\ic x_{\alpha-1} v_{\alpha-1} v_{\alpha-1}^\dagger)\,, \quad
\Phi_{\alpha}=\exp(\ic x_{\alpha} v_{\alpha} v_{\alpha}^\dagger)\,.
\end{equation*}
\end{corollary}
\begin{proof}
To get $\omega_{\MM_d}$, it suffices to proceed by induction as in the proof of Theorem \ref{Thm:qH-Master} with the quasi-Hamiltonian fusion product of Proposition \ref{Pr:Fus-omega}.
The fact that $\omega_{\MM_d}$ and $P_{\MM_d}$ are compatible follows because compatibility is preserved by fusion \cite[Proposition 10.7]{AKSM}. This entails that $(\MM_d,P_{\MM_d},\Phi)$ is non-degenerate.
\end{proof}

So far, we allowed the parameters $x_1,\ldots,x_d\in \R\setminus\{0\}$ to be completely arbitrary. However, the quasi-Poisson structure on $\disk(x)$ from Proposition \ref{Pr:qP-Cn} can be transformed into the one of $\disk(\sgn(x))$ through $v\mapsto v \sqrt{|x|}$.
After fusion, this observation can be rephrased as follows.
\begin{proposition} \label{Pr:IsoSign}
There is an isomorphism of Hamiltonian quasi-Poisson manifolds
\begin{equation*}
 \MM_d(x_1,\ldots,x_d) \mapsto \MM_d(\sgn(x_1),\ldots,\sgn(x_d))\,, \quad (A,B,v_\alpha)\mapsto (A,B,v_\alpha\sqrt{|x_\alpha|})\,.
\end{equation*}
\end{proposition}

Although as smooth manifolds $\disk(1)$ and $\disk(-1)$ are the same open ball of radius $2\pi$ centered at $0\in \CC^n$, their quasi-Poisson brackets given in Proposition \ref{Pr:qP-Cn} are opposite of each other; this entails that their moment maps are inverses of each other.
There is no such relation \emph{after fusion} between the $2^d$ Hamiltonian quasi-Poisson structures defined on the isomorphic manifolds $\MM_d(s_1,\ldots,s_d)$, $(s_\alpha)\in \{\pm 1\}^d$ (which are the most fundamental according to Proposition \ref{Pr:IsoSign}).
For example, if we consider
$\MM_d(1,\ldots,1)$ and $\MM_d(s_1,\ldots,s_d)$ with $(s_\alpha)\in \{\pm 1\}^d \setminus\{(1,\ldots,1)\}$, their quasi-Poisson brackets only differ by the $\alpha=\beta$ case in the third identity of \eqref{Eq:br-VVall} for each $\alpha$ such that $s_\alpha=-1$.

\subsection{Pencil of compatible quasi-Poisson brackets} \label{ss:Pencil}

Fix a manifold $M$ with an action of a Lie group $G$ and a quasi-Poisson bivector $P \in \wedge^2 \Vect(M)$.
For $k\geq 1$, consider $k$ bivectors $Q_1,\ldots,Q_k\in \wedge^2 \Vect(M)$ independent from $P$, i.e. the elements $P_m,(Q_1)_m,\ldots,(Q_k)_m\in \wedge^2 T_mM$ are linearly independent at some point $m\in M$.
We say that these bivectors span a $k$-dimensional \emph{pencil of compatible quasi-Poisson bivectors centered} at $P$ if $P+\sum_{1\leq j \leq k}c_j Q_j$ is a quasi-Poisson bivector for any $c_1,\ldots,c_k\in \R$.

When restricting a pencil of compatible quasi-Poisson bivectors to invariant functions, we obtain a `usual' pencil of compatible Poisson brackets on $\Cinf(M)^G$. Indeed,  $c_0 P+\sum_{1\leq j \leq k}c_j Q_j$ defines a Poisson bivector for any $c_0,c_1,\ldots,c_k\in \R$.
However, before restriction to $\Cinf(M)^G$, there are important differences between our definition and the standard notion of compatibility on Poisson manifolds. If $P_1,P_2$ belong to a pencil centered at $P$, then $c_1 P_1+c_2 P_2$ is not a quasi-Poisson bivector for all $c_1,c_2\in \R$ and we do not have $[P_1,P_2]=0$.
As an example, if $Q\in \wedge^2 \Vect(M)$ is a nonvanishing Poisson bivector such that $[P,Q]=0$, then $P_1=P+Q$ belongs to the pencil $\{P+c Q \mid c\in \R\}$ centered at $P$.
However, if $\phi_M\neq 0$, $P-P_1$ is certainly not a quasi-Poisson bivector and we have $[c_1P_1+c_2 P,c_1P_1+c_2 P]=(c_1+c_2)^2 \phi_M$ which is a quasi-Poisson bivector only for $(c_1+c_2)^2=1$.

\medskip

Our aim is to exhibit a pencil of compatible quasi-Poisson bivectors on $\MM_d$, for any $d\geq 2$.
Consider the Hamiltonian quasi-Poisson structure of Theorem \ref{Thm:qH-Master}, and let $P_{\MM_d}$ \eqref{Eq:qPMd} be the corresponding quasi-Poisson bivector.  For any $1\leq \beta \leq d$, define an action of $\UU(1)$ on $\MM_d$ by
\begin{equation}  \label{Eq:act-U1}
\mathcal{A}_\beta^{\UU(1)}:\UU(1)\times \MM_d \to \MM_d\,, \qquad  \lambda \cdot (A,B,v_\alpha)=(A,B,\lambda^{-\delta_{\alpha \beta}}v_\alpha)\,.
\end{equation}
Then, the  infinitesimal action of $\ic z\in \ic \R \simeq \uu(1)$ is given by the vector field
$(\ic z)_{\MM_d;\beta}\in \Vect(\MM_d)$ satisfying
\begin{equation}  \label{Eq:act-U1-inf}
(\ic z)_{\MM_d;\beta}(A)=(\ic z)_{\MM_d;\beta}(B)=0\,, \quad
(\ic z)_{\MM_d;\beta}(v_\alpha)=\ic z \delta_{\alpha\beta} v_\alpha\,, \quad
(\ic z)_{\MM_d;\beta}(v_\alpha^\dagger)=-\ic z \delta_{\alpha\beta} v_\alpha^\dagger\,.
\end{equation}

\begin{proposition} \label{Pr:Pencil}
There exists a $d(d-1)/2$-dimensional pencil of compatible quasi-Poisson bivectors centered at $P_{\MM_d}$ which define a structure of Hamiltonian quasi-Poisson $\UU(n)$-manifold on $\MM_d$ with the moment map $\Phi:\MM_d\to \UU(n)$ given by \eqref{Eq:Mast-phi-1}.
Explicitly, for any fixed parameters $\underline{z}=(z_{\alpha\beta})_{\alpha<\beta} \in \R^{d(d-1)/2}$, the bivector
 \begin{equation} \label{Eq:PencCor}
 P_{\underline{z}}:=P_{\MM_d}+\psi_{\underline{z}}\,, \qquad
 \psi_{\underline{z}}:=\sum_{1\leq \alpha<\beta \leq d} z_{\alpha\beta} \, \ic_{\MM_d;\alpha} \wedge \ic_{\MM_d;\beta}\,,
\end{equation}
turns $(\MM_d,P_{\underline{z}},\Phi)$ into a Hamiltonian quasi-Poisson manifold for the $\UU(n)$-action \eqref{Eq:act-Mast}.
In particular, we can decompose the quasi-Poisson bracket associated with $P_{\underline{z}}$ as
$$\br{-,-}_{\underline{z}}=\br{-,-}+\br{-,-}_{\psi_{\underline{z}}}\,,$$
where $\br{-,-}$ is given in Theorem \ref{Thm:qH-Master}, while $\br{-,-}_{\psi_{\underline{z}}}$ satisfies
\begin{equation} \label{Eq:underz-1}
 \br{F,G}_{\psi_{\underline{z}}}=0\,, \qquad \forall F=F(A,B)\in \Cinf(\D(\UU(n))),\,\, G\in \Cinf(\MM_d),
\end{equation}
and for any $1\leq i,k \leq n$ and  $1\leq \alpha<\beta \leq d$,
\begin{equation} \label{Eq:underz-vv}
 \begin{aligned}
\br{(v_\alpha)_i,(v_\beta)_k}_{\psi_{\underline{z}}}=&
- z_{\alpha\beta}  (v_\alpha)_i(v_\beta)_k\,,  \quad
\br{(\bar{v}_\alpha)_i,(\bar{v}_\beta)_k}_{\psi_{\underline{z}}}=
- z_{\alpha\beta}  (\bar{v}_\alpha)_i(\bar{v}_\beta)_k \,, \\
\br{(v_\alpha)_i,(\bar{v}_\beta)_k}_{\psi_{\underline{z}}}=&\, +z_{\alpha\beta} (v_\alpha)_i(\bar{v}_\beta)_k\,, \quad
\br{(v_\beta)_i,(\bar{v}_\alpha)_k}_{\psi_{\underline{z}}}=\, -z_{\alpha\beta} (v_\beta)_i(\bar{v}_\alpha)_k \,.
 \end{aligned}
\end{equation}
\end{proposition}
We will prove Proposition \ref{Pr:Pencil} in \S\ref{ss:ProofPen} after gathering some preliminary results.

\subsubsection{Quasi-Poisson structures with an Abelian component}

\begin{lemma} \label{Lem:qP-abel1}
Assume that $(M,P)$ is a quasi-Poisson manifold for an action of a Lie group $G$.
Assume in addition that there is an action on $M$ by an Abelian Lie group $K$ such that $P$ is $K$-invariant.
Given $2k$ elements $\underline{\xi}:=(\xi_i,\xi_i')_{i=1}^k$ in the Lie algebra $\mathfrak{K}$ of $K$, the bivector
$\psi_{\underline{\xi}}=\sum_{i=1}^k (\xi_i)_M \wedge (\xi_i')_M$ turns $(M,P+\psi_{\underline{\xi}})$ into a quasi-Poisson manifold for the action of $G$.
Furthermore, if $(M,P,\Phi)$ is Hamiltonian with moment map $\Phi:M\to G$ which is $K$-invariant, then $(M,P+\psi_{\underline{\xi}},\Phi)$ is also Hamiltonian.
\end{lemma}
\begin{proof}
Since $K$ is Abelian and $M$ is a quasi-Poisson $G$-manifold, we have $[P,P]=\phi_M$ for $\phi\in \wedge^3 \mathfrak{g}$ the Cartan trivector corresponding to $G$.
Any two infinitesimal vector fields $\xi_M,\zeta_M$ commute for $\xi,\zeta\in \mathfrak{K}$ because $K$ is Abelian.
The $K$-invariance of $P$ implies that $[\xi_M,P]=0$ for any $\xi\in \mathfrak{K}$. Hence
$[P+\psi_{\underline{\xi}},P+\psi_{\underline{\xi}}]=[P,P]=\phi_M$, which proves the first claim.

For the second claim, note that by $K$-invariance of $\Phi$, $\xi_M(\Phi)=0$ for any $\xi\in \mathfrak{K}$. Thus the moment map property given by \eqref{Eq:momap} with respect to $G$ is still satisfied if we consider the quasi-Poisson bracket corresponding to $P+\psi_{\underline{\xi}}$.
\end{proof}

A direct consequence of Lemma \ref{Lem:qP-abel1} is given by the following result.

\begin{lemma} \label{Lem:qP-abel2}
Under the assumptions of Lemma \ref{Lem:qP-abel1}, suppose that there exists $\ell \geq 2$ elements $\xi^{(1)},\ldots,\xi^{(\ell)} \in \mathfrak{K}$, such that the corresponding infinitesimal vector fields $\xi^{(j)}_M$ induce linearly independent tangent vectors at least at one point of $M$.
Then they define a pencil of compatible quasi-Poisson brackets  of dimension $\ell(\ell-1)/2$ centered at $P$: any bivector
\begin{equation*}
 P_{\underline{z}}:=P+\psi_{\underline{z}}\,, \quad \psi_{\underline{z}}:=\sum_{1\leq i<j\leq \ell} z_{ij}\, \xi^{(i)}_M\wedge \xi^{(j)}_M\,, \quad
\text{with }\underline{z}:=(z_{ij}), \,\, z_{ij}\in \R \text{ for }1\leq i<j\leq \ell\,,
\end{equation*}
is a quasi-Poisson bracket on $M$ for the action of $G$.
Furthermore, if $(M,P,\Phi)$ is Hamiltonian, $\Phi$ is a moment map for $(M,P_{\underline{z}})$.
\end{lemma}

\subsubsection{Proof of Proposition \ref{Pr:Pencil}} \label{ss:ProofPen}

We want to apply Lemma \ref{Lem:qP-abel1}, hence we have to check that $P_{\MM_d}$ \eqref{Eq:qPMd} and $\Phi$ \eqref{Eq:Mast-phi-1} are invariant under the $\UU(1)^d$-action defined as the product of the $d$ actions in \eqref{Eq:act-U1}. The invariance of $\Phi$ is easily seen. For the invariance of $P_{\MM_d}$, we use its decomposition according to \eqref{Eq:qPMd}.
First, note that each of the $d$ actions \eqref{Eq:act-U1} does not act on the pair $(A,B)\in \D(\UU(n))$, so $P_{\D(\UU(n))}$ is trivially invariant.
Moreover, each action $\mathcal{A}_\beta^{\UU(1)}$ commutes with the restriction action $\mathcal{A}^{\res}_{\alpha-1}$ \eqref{Eq:act-Res} and the $\UU(n)$-action by left multiplication on the $\alpha$-th ball $\disk(x_\alpha)$, hence the third type of bivectors in \eqref{Eq:qPMd} is invariant.
It remains to check that for any $1\leq \beta \leq d$, $P_{\disk^\beta}$ is invariant under the action $\mathcal{A}_\beta^{\UU(1)}$ defined in \eqref{Eq:act-U1}, and this is easily seen from the corresponding quasi-Poisson bracket \eqref{Eq:qPB-C2}.

Applying Lemma \ref{Lem:qP-abel1} to $(\MM_d,P_{\MM_d},\Phi)$, we see that any bivector $P_{\underline{z}}:=P_{\MM_d}+\psi_{\underline{z}}$ as in \eqref{Eq:PencCor}
endows $\MM_d$ with a Hamiltonian quasi-Poisson structure for the action of $\UU(n)$ given by \eqref{Eq:act-Mast} and the moment map $\Phi$.
In fact, the $d$ infinitesimal vector fields $\ic_{\MM_d;1},\ldots,\ic_{\MM_d;d}$ are easily seen from \eqref{Eq:act-U1-inf} to be linearly independent tangent vectors at generic points of $\MM_d$.
Thus, by Lemma \ref{Lem:qP-abel2}, we get on $\MM_d$ a pencil of compatible quasi-Poisson brackets  of dimension $d(d-1)/2$ centered at $P_{\MM_d}$  which are Hamiltonian for $\Phi:\MM_d \to \UU(n)$ given by \eqref{Eq:Mast-phi-1}.

It remains to check the formula of the bracket $\br{-,-}_{\underline{z}}$ associated with $P_{\underline{z}}$.
We only need to understand the bracket $\br{-,-}_{\psi_{\underline{z}}}$ associated with $\psi_{\underline{z}}$, and we easily obtain \eqref{Eq:underz-1} and \eqref{Eq:underz-vv} from the infinitesimal actions \eqref{Eq:act-U1-inf}.  \qed

\subsubsection{Comments on the pencil and non-degeneracy} \label{sss:Comm-pen}

The assumptions of Lemma \ref{Lem:qP-abel1} are satisfied whenever we are given a (Hamiltonian) quasi-Poisson manifold $M$ for an action of $G\times K$ with $K$ Abelian. Indeed, the Cartan trivector of $K$ trivially vanishes hence we can view $M$ as a (Hamiltonian) quasi-Poisson $G$-manifold. Some quasi-Poisson bivectors that can be found in the pencil of  Proposition \ref{Pr:Pencil} on $\MM_d$ are precisely obtained in that way, as we now explain.

The internally fused double $\D(\UU(n))$ can be regarded as a Hamiltonian quasi-Poisson $\UU(n)\times \UU(1)$-manifold for the trivial action of $\UU(1)$ by $\lambda\cdot (A,B)=(A,B)$ and the corresponding moment map $\widehat{\Phi}:\D(\UU(n))\to \UU(1)$, $\widehat{\Phi}(A,B)=1$.
Building on Remark \ref{Rem:App-U1} with arbitrary $x\in \R\setminus\{0\}$, we can also view $\disk(x)$ as a Hamiltonian quasi-Poisson $\UU(n)\times \UU(1)$-manifold for the $\UU(1)$-action by $\lambda\cdot v=v \lambda^{-1}$ and the moment map $\widehat{\Phi}:\disk(x)\to \UU(1)$, $\widehat{\Phi}(v)=e^{-\ic x |v|^2}$.
(In both cases, the quasi-Poisson bracket is unchanged and is the one of Propositions \ref{Pr:qP-intDble} and \ref{Pr:qP-Cn}.)
Thus, we can start with $\D(\UU(n))\times \disk(x_1) \times \cdots \times \disk(x_d)$, which is nothing else than $\MM_d$ before fusion, seen as a Hamiltonian quasi-Poisson $(\UU(n)^{d+1}\times \UU(1)^{d+1})$-manifold. Repeating the proof of Theorem \ref{Thm:qH-Master} where we fuse the $\UU(n)$-actions in the same order while we allow ourselves to fuse some of the $\UU(1)$-actions (in arbitrary orders), we end up with a structure of Hamiltonian quasi-Poisson $(\UU(n)\times \UU(1)^{\ell})$-manifold on $\MM_d$ for some $1\leq \ell \leq d+1$.
The case $\ell=d+1$ amounts not to perform any fusion of the $\UU(1)$-actions, hence it gives the same quasi-Poisson bivector $P_{\MM_d}$ as in \eqref{Eq:qPMd}. The case $\ell=1$ where we fuse the $\UU(1)$-actions recursively in the same order as we fused the $\UU(n)$-actions leads to the moment map $\widehat{\Phi}:\MM_d\to \UU(1)$, $\widehat{\Phi}(v)=e^{-\ic x_1 |v_1|^2}\cdots e^{-\ic x_d |v_d|^2}$, and the quasi-Poisson bivector
\begin{equation} \label{Eq:P-FullFus}
 P_{\MM_d}-\frac12 \sum_{1\leq \alpha <\beta \leq d} \ic_{\MM_{d};\alpha}\wedge \ic_{\MM_{d};\beta}\,.
\end{equation}
Any other case yields a quasi-Poisson bivector of the form
\begin{equation*}
 P_{\MM_d}+ \sum_{1\leq \alpha <\beta \leq d} z_{\alpha \beta}\, \ic_{\MM_{d};\alpha}\wedge \ic_{\MM_{d};\beta}\,, \quad
 z_{\alpha\beta} \in \Big\{0,\pm \frac12\Big\}\,,
\end{equation*}
and they are all easily seen to be special instances of \eqref{Eq:PencCor}. By Lemma \ref{Lem:qP-abel1}, for any choice of fusion of the $d+1$ actions of $\UU(1)$, we can forget the residual $\UU(1)$-actions on $\MM_d$ and obtain a structure of Hamiltonian quasi-Poisson manifold on $\MM_d$ for the action of $\UU(n)$ given by \eqref{Eq:act-Mast}. Any such structure falls within the pencil constructed as part of Proposition \ref{Pr:Pencil}.

The quasi-Poisson bivectors obtained by performing these extra fusions admit corresponding compatible $2$-forms in view of Proposition \ref{Pr:Fus-omega} and \cite[Prop. 10.7]{AKSM}. For example, we can compute using Proposition \ref{Pr:Fus-omega} that the $2$-form compatible with \eqref{Eq:P-FullFus} is
\begin{equation*}
 \omega_{\MM_d}-\frac12 \sum_{1\leq \alpha <\beta \leq d} x_\alpha x_\beta\, d|v_\alpha|^2\wedge d|v_\beta|^2\,.
\end{equation*}
Based on this example, we seek to define a compatible $2$-form for each quasi-Poisson bivector $P_{\underline{z}}$ from the pencil given in \eqref{Eq:PencCor}.

\begin{proposition} \label{Pr:om-Pencil}
Fix $\underline{z}=(z_{\alpha\beta})_{\alpha<\beta} \in \R^{d(d-1)/2}$. Then the $2$-form
\begin{equation}
 \label{Eq:Om-PencCor}
\omega_{\underline{z}}:=\omega_{\MM_d}+\widetilde{\omega}_{\underline{z}}\,, \qquad
\widetilde{\omega}_{\underline{z}}:=\sum_{1\leq \alpha<\beta \leq d} x_\alpha x_\beta\, z_{\alpha\beta} \, d|v_\alpha|^2\wedge d|v_\beta|^2\,,
\end{equation}
turns $(\MM_d,\omega_{\underline{z}},\Phi)$ into a quasi-Hamiltonian $\UU(n)$-manifold.
Furthermore, $\omega_{\underline{z}}$ and $P_{\underline{z}}$ \eqref{Eq:PencCor}  satisfy the compatibility condition \eqref{Eq:corrPOm}.
In particular, any quasi-Poisson bivector presented in Proposition \ref{Pr:Pencil} is non-degenerate.
\end{proposition}
\begin{proof}
The final part follows from the rest of the statement and Theorem \ref{Thm:Corr}.

We check that $(\MM_d,\omega_{\underline{z}},\Phi)$ is a quasi-Hamiltonian manifold, noting that it holds for $\underline{z}=0$ by Corollary \ref{cor:qHamMd}.
It suffices to verify (B2)-(B3) in Definition \ref{def:om-qHam} since (B1) is automatic as  $d\widetilde{\omega}_{\underline{z}}=0$.
By \eqref{Eq:act-Mast}, for any $1\leq \alpha \leq d$ the action of $\UU(n)$ on $|v_\alpha|^2$ is trivial and thus $\iota_{\xi_{\MM_d}}d|v_\alpha|^2=0$  for any $\xi\in \uu(n)$; this yields $\iota_{\xi_{\MM_d}}\omega_{\underline{z}}=\iota_{\xi_{\MM_d}}\omega_{\MM_d}$ and (B2) holds.
The last condition (B3) will follow from the compatibility \eqref{Eq:corrPOm}.
Indeed, if $X\in \Vect(\MM_d)$ is such that $X_m\in \ker \omega_{\underline{z},m}$  at $m\in \MM_d$, then \eqref{Eq:corrPOm} entails that $X_m=\xi_{\MM_d,m}$
for some $\xi\in \uu(n)$.
As we remarked that $\iota_{\xi_{\MM_d}}\omega_{\underline{z}}=\iota_{\xi_{\MM_d}}\omega_{\MM_d}$ for all $\xi\in \uu(n)$, we get $\ker \omega_{\underline{z},m}=\ker \omega_{\MM_d,m}$ and therefore (B3) holds.

\medskip

We need to verify the compatibility condition \eqref{Eq:corrPOm} for
$\omega_{\underline{z}}$ \eqref{Eq:Om-PencCor} and $P_{\underline{z}}$ \eqref{Eq:PencCor}, i.e.
\begin{equation*}
P_{\MM_d}^\sharp \circ \omega_{\MM_d}^\flat + P_{\MM_d}^\sharp \circ \widetilde{\omega}_{\underline{z}}^\flat
+\psi_{\underline{z}}^\sharp \circ \omega_{\MM_d}^\flat + \psi_{\underline{z}}^\sharp \circ \widetilde{\omega}_{\underline{z}}^\flat
= \id_{T\MM_d}- \frac14 \sum_{a\in \mathtt{A}} (e_a)_{\MM_d} \otimes \Phi^\ast(\theta_a^L-\theta_a^R)\,.
\end{equation*}
Using that this equality holds for $\underline{z}=0$ by Corollary \ref{cor:qHamMd} together with the fact that $\psi_{\underline{z}}^\sharp \circ \widetilde{\omega}_{\underline{z}}^\flat=0$ because the infinitesimal $\uu(1)$-actions \eqref{Eq:act-U1-inf} send each $|v_\alpha|^2$ to zero, we are left to check
\begin{equation} \label{Eq:pf-PencCor}
P_{\MM_d}^\sharp \circ \widetilde{\omega}_{\underline{z}}^\flat
+\psi_{\underline{z}}^\sharp \circ \omega_{\MM_d}^\flat =0\,.
\end{equation}
In fact, it is sufficient to prove that \eqref{Eq:pf-PencCor} holds when evaluated on $X=DJ_{\D}(X')$ where $X'\in \Vect(\D(\UU(n)))$ with the inclusion $J_{\D}:\D(\UU(n))\to \MM_d$, or on $X=\partial/\partial (v_\gamma)_j$ and on $X=\partial/\partial (\bar{v}_\gamma)_j$ with $1\leq \gamma \leq d$, $1\leq j\leq n$.

\noindent $\bullet$ \textit{Case 1: $X=DJ_{\D}(X')$ with $X'\in \Vect(\D(\UU(n)))$.}
Clearly $\iota_X d|v_\alpha|^2=0$ for any $1\leq \alpha \leq d$, hence $\widetilde{\omega}_{\underline{z}}^\flat(X)=0$ and we are left with the second term in \eqref{Eq:pf-PencCor}.
Expanding $\omega_{\MM_d}$ as in \eqref{Eq:om-Md} and using the notation therein, we can write
\begin{equation*}
\begin{aligned}
\omega_{\MM_d}^\flat(X)=&
 \omega_{\D(\UU(n))}^\flat(X) -\frac12 \sum_{1\leq \alpha \leq d} \sum_{a\in \mathtt{A}}
 \left(\Phi^\ast_{\MM_{\alpha-1}}\theta_a^L\wedge \Phi^\ast_{\alpha}\theta_a^R\right)^\flat(X) \\
=&\sum_{1\leq k,l\leq n} \left(f^{(1)}_{ij} dA_{ij} + f^{(2)}_{ij} dB_{ij} + \sum_{1\leq \alpha \leq d} g^{(\alpha)}_{ij} d(v_\alpha v_\alpha^\dagger)_{ij} \right)\,,
\end{aligned}
\end{equation*}
for some $f^{(1)}_{ij},f^{(2)}_{ij},g^{(\alpha)}_{ij}\in \Cinf(\MM_d)$. We easily see from \eqref{Eq:act-U1-inf} that
\begin{equation}  \label{Eq:pf-PencCor2}
\ic_{\MM_d;\beta}(A_{ij})=0\,, \quad
\ic_{\MM_d;\beta}(B_{ij})=0\,, \quad
\ic_{\MM_d;\beta}((v_\alpha v_\alpha^\dagger)_{ij})=0\,,
\end{equation}
for all possible indices, hence applying $\psi_{\underline{z}}^\sharp$ to the above expression yields zero.

\noindent $\bullet$ \textit{Case 2: $X=\partial/\partial (v_\gamma)_j$.}
We start by noting that \eqref{Eq:br-ABv} and \eqref{Eq:br-VVall} entail
\begin{equation*}
\begin{aligned}
  P_{\MM_d}^\sharp(d|v_\alpha|^2)=&
 \sum_{1\leq k \leq n} \left(\br{|v_\alpha|^2,(\bar{v}_\alpha)_k} \frac{\partial}{\partial (\bar{v}_\alpha)_k}
 +\br{|v_\alpha|^2,(v_\alpha)_k}  \frac{\partial}{\partial (v_\alpha)_k}\right) \\
=& \sum_{1\leq k \leq n} \left(\frac{\ic}{x_\alpha} (\bar{v}_\alpha)_k  \frac{\partial}{\partial (\bar{v}_\alpha)_k}
 - \frac{\ic}{x_\alpha} (v_\alpha)_k  \frac{\partial}{\partial (v_\alpha)_k}\right)
 =- \frac{1}{x_\alpha} \ic_{\MM_d;\alpha}\,,
\end{aligned}
\end{equation*}
where we used \eqref{Eq:act-U1-inf} in the last equality. Therefore
\begin{equation}
\begin{aligned}  \label{Eq:pf-PencCor3}
  P_{\MM_d}^\sharp\circ \widetilde{\omega}_{\underline{z}}^\flat \left(X\right)
=& \sum_{\gamma<\beta \leq d} x_\gamma x_\beta\, z_{\gamma\beta} \, (\bar{v}_\gamma)_j \,  P_{\MM_d}^\sharp(d|v_\beta|^2)
-\sum_{1\leq \alpha<\gamma} x_\alpha x_\gamma\, z_{\alpha\gamma} \,(\bar{v}_\gamma)_j \,  P_{\MM_d}^\sharp(d|v_\alpha|^2) \\
=&- \sum_{\gamma<\beta \leq d} x_\gamma \, z_{\gamma\beta} \, (\bar{v}_\gamma)_j \,  \ic_{\MM_d;\beta}
+\sum_{1\leq \alpha<\gamma} x_\gamma\, z_{\alpha\gamma} \,(\bar{v}_\gamma)_j \,  \ic_{\MM_d;\alpha}\,.
\end{aligned}
\end{equation}
To compute the second term from \eqref{Eq:pf-PencCor}, we expand  $\omega_{\MM_d}$ as in \eqref{Eq:om-Md} and we write
\begin{equation} \label{Eq:pf-PencCor4}
\begin{aligned}
 \psi_{\underline{z}}^\sharp \circ \omega_{\MM_d}^\flat(X)
=&\psi_{\underline{z}}^\sharp \circ \omega_{\D(\UU(n))}^\flat(X)
+ \sum_{1\leq \alpha \leq d} \psi_{\underline{z}}^\sharp\circ \omega_{\disk^\alpha}^\flat(X) \\
&-\frac12 \sum_{1\leq \alpha \leq d} \sum_{a\in \mathtt{A}} \psi_{\underline{z}}^\sharp\circ (\Phi^\ast_{\MM_{\alpha-1}}\theta_a^L\wedge \Phi^\ast_{\alpha}\theta_a^R)^\flat(X) \,.
\end{aligned}
\end{equation}
The first term in \eqref{Eq:pf-PencCor4} is trivially zero as $\omega_{\D(\UU(n))}$ does not depend on $v_\gamma$, and the summands with $\alpha\neq \gamma$ in the second term vanish too. For the third term, we notice as in Case 1 that
\begin{equation*}
(\Phi^\ast_{\MM_{\alpha-1}}\theta_a^L\wedge \Phi^\ast_{\alpha}\theta_a^R)^\flat(X)
=\sum_{1\leq k,l\leq n} \left(f^{(1)}_{ij} dA_{ij} + f^{(2)}_{ij} dB_{ij} + \sum_{1\leq \beta \leq \alpha} g^{(\beta)}_{ij} d(v_\beta v_\beta^\dagger)_{ij} \right)\,,
\end{equation*}
for some $f^{(1)}_{ij},f^{(2)}_{ij},g^{(\beta)}_{ij}\in \Cinf(\MM_d)$, hence applying $\psi_{\underline{z}}^\sharp$ to this expression gives zero by \eqref{Eq:pf-PencCor2}.
Thus \eqref{Eq:pf-PencCor4} simplifies to
\begin{equation} \label{Eq:pf-PencCor5}
\begin{aligned}
 &\psi_{\underline{z}}^\sharp \circ \omega_{\MM_d}^\flat(X)
=\psi_{\underline{z}}^\sharp\circ \omega_{\disk^\gamma}^\flat(X) \\
=&\ic \frac{\sin(x_\gamma |v_\gamma|^2)}{|v_\gamma|^2} \, \psi_{\underline{z}}^\sharp(d(\bar{v}_\gamma)_j)
+\ic \left( x_\gamma - \frac{\sin(x_\gamma |v_\gamma|^2)}{|v_\gamma|^2} \right)
\sum_{1\leq k \leq n}  \frac{(\bar{v}_\gamma)_j(v_\gamma)_k}{|v_\gamma|^2}  \psi_{\underline{z}}^\sharp(d(\bar{v}_\gamma)_k)\,,
\end{aligned}
\end{equation}
where we used \eqref{Eq:qSymp} with $v=v_\gamma$ and $x=x_\gamma$. We can finally use  \eqref{Eq:act-U1-inf} to find
\begin{equation}
 \psi_{\underline{z}}^\sharp \circ \omega_{\MM_d}^\flat(X)=
\sum_{\gamma<\beta \leq d} x_\gamma \, z_{\gamma\beta} \, (\bar{v}_\gamma)_j \,  \ic_{\MM_d;\beta}
-\sum_{1\leq \alpha<\gamma} x_\gamma\, z_{\alpha\gamma} \,(\bar{v}_\gamma)_j \,  \ic_{\MM_d;\alpha}\,.
\end{equation}
This is the opposite of \eqref{Eq:pf-PencCor3}, therefore \eqref{Eq:pf-PencCor} is zero when applied to $X=\partial/\partial (v_\gamma)_j$, as desired.

\noindent $\bullet$ \textit{Case 3: $X=\partial/\partial (\bar{v}_\gamma)_j$.} This is similar to Case 2. Hence we conclude that \eqref{Eq:pf-PencCor} holds identically, and the compatibility between $\omega_{\underline{z}}$ and $P_{\underline{z}}$ follows.
\end{proof}

\begin{remark} \label{Rem:Reg}
Any bivector $P_{\underline{z}}$ from the pencil satisfies the following property: on a dense subset of $\MM_d$ the rank of $P_{\underline{z}}$ is equal to $\dim(\MM_d)=2n^2+2nd$.
To see this, first note that the property only needs to be shown at one point because  $P_{\underline{z}}$ is real-analytic.
Then remark that the quasi-Poisson bivector of $\disk(x)$ has rank $2n$ at the origin for any $x\in \R\!\setminus\!\{0\}$; this is readily seen from \eqref{Eq:qPB-C2} as $\btt(0)=0$.
At a point $m=(A,B,0,\ldots,0)\in \MM_d$ we note that the third summand of $P_{\MM_d}$ in \eqref{Eq:qPMd} and $\psi_{\underline{z}}$ \eqref{Eq:PencCor} vanish, which thus implies that
$$\operatorname{rank}((P_{\underline{z}})_m)=\operatorname{rank}((P_{\D(\UU(n))})_{(A,B)})+2nd\,.$$
Since $\D(\UU(n))$ is non-degenerate by Proposition \ref{Pr:qP-intDble},
$\operatorname{rank}((P_{\D(\UU(n))})_{(A,B)})=2n^2$ at any point $(A,B)\in \D(\UU(n))$ where the linear map $\id_{\uu(n)}+\Ad_{ABA^{-1}B^{-1}}:\uu(n)\to \uu(n)$ is invertible.
Invertibility is direct at $(\1_n,\1_n)\in \D(\UU(n))$, which yields our claim.
\end{remark}

\begin{remark} \label{Rem:Cotang}
The above constructions can be easily adapted to suitable non-degenerate Poisson manifolds, such as the extended cotangent bundle $M_d$ \eqref{I1}.
To make the comparison transparent and depending on $d\geq 2$ parameters $x_1,\ldots,x_d\in \R\setminus\{0\}$ in the later case, we parametrize $M_d$ as $\UU(n) \times \uu(n) \times (\CC^n)^d=\{(g,J,v_1,\ldots,v_d)\}$ and write its Poisson bivector as
\begin{equation}
 P_{M_d}=P_{T^*\!\UU(n)} + \sum_{1\leq \alpha \leq d} \frac{\ic}{x_\alpha} \sum_{j=1}^n \frac{\partial}{\partial (v_\alpha)_j} \wedge \frac{\partial}{\partial (\bar{v}_\alpha)_j}\,.
\end{equation}
The moment map is then given by $\phi(g,J,v_\alpha)=J-g^{-1}Jg+\ic \sum_{\alpha} x_\alpha v_\alpha v_\alpha^\dagger$, and the corresponding symplectic form reads
\begin{equation}
 \omega_{M_d}=\omega_{T^*\!\UU(n)}
+ \sum_{1\leq \alpha \leq d} \ic x_\alpha \sum_{j=1}^n d(v_\alpha)_j \wedge d(\bar{v}_\alpha)_j\,.
\end{equation}
As in \eqref{Eq:act-U1}--\eqref{Eq:act-U1-inf},
we get for each $1\leq \beta\leq d$ a map $\uu(1)\to \Vect(M_d)$, $\ic z\mapsto (\ic z)_{M_d;\beta}$, after differentiating the action of $\UU(1)$ on $M_d$ by $\lambda \cdot (g,J,v_\alpha):=(g,J,\lambda^{-\delta_{\alpha \beta}}v_\alpha)$.
Let us fix parameters $\underline{z}=(z_{\alpha\beta})_{\alpha<\beta} \in \R^{d(d-1)/2}$.
By adapting the proof of Proposition \ref{Pr:Pencil}, we see that the bivector
 \begin{equation}
 P_{\underline{z}}:=P_{M_d}+\psi_{\underline{z}}\,, \qquad
 \psi_{\underline{z}}:=\sum_{1\leq \alpha<\beta \leq d} z_{\alpha\beta} \, \ic_{M_d;\alpha} \wedge \ic_{M_d;\beta}\,,
\end{equation}
turns $(M_d,P_{\underline{z}},\phi)$ into a Hamiltonian Poisson manifold.
Indeed, the $\UU(1)$-actions pairwise commute and leave $P_{M_d}$ and $\phi$ invariant,
so that $[P_{\underline{z}},P_{\underline{z}}]=[P_{M_d},P_{M_d}]=0$ while for any $\xi\in \uu(n)$, $P_{\underline{z}}^\sharp(d\langle\phi,\xi\rangle)=P_{M_d}^\sharp(d\langle\phi,\xi\rangle)=\xi_{M_d}$.
Furthermore, we can show that the closed $2$-form
\begin{equation}
\omega_{\underline{z}}:=\omega_{M_d}+\widetilde{\omega}_{\underline{z}}\,, \qquad
\widetilde{\omega}_{\underline{z}}:=\sum_{1\leq \alpha<\beta \leq d} x_\alpha x_\beta\, z_{\alpha\beta} \, d|v_\alpha|^2\wedge d|v_\beta|^2\,,
\end{equation}
is the corresponding symplectic form
by adapting the proof of Proposition \ref{Pr:om-Pencil}. To do so, we only need to check that $P_{\underline{z}}^\sharp \circ \omega_{\underline{z}}^\flat=\id_{TM_d}$. This holds for $\underline{z}=0$, hence we are left to verify
\begin{equation} \label{Eq:CompCot}
P_{M_d}^\sharp \circ \widetilde{\omega}_{\underline{z}}^\flat
+\psi_{\underline{z}}^\sharp \circ \widetilde{\omega}_{\underline{z}}^\flat
+\psi_{\underline{z}}^\sharp \circ \omega_{M_d}^\flat = 0\,.
\end{equation}
The middle term identically vanishes because each function $|v_\alpha|^2$ is invariant under the $d$ $\UU(1)$-actions.
Thus we can conclude if the first and third terms in \eqref{Eq:CompCot} cancel out, and this is verified by a straightforward calculation.
\end{remark}

\section{Integrable system on the master phase space}

\label{S:Integr-Md}

We are going to present a dynamical system on $\MM_d=\MM_d(x_1,\ldots,x_d)$ \eqref{Eq:Mast-Md} that possesses quite similar properties to a
degenerate integrable system on a symplectic manifold.
For this reason, we deem it justified to call it an \emph{integrable system}.
Our system will be associated with the ring of class function on $\UU(n)$.
Later we shall apply Hamiltonian reduction to this system.

\subsection{Unreduced dynamics} \label{ss:Unred}

Let $\cE_1$, $\cE_2$ be the projections from $\cM_d$ onto the first and second $\UU(n)$ factors of
$\D(\UU(n))$, and let $\cE(\alpha)$ for $\alpha =1,\dots, d$ be the projection from $\cM_d$  onto the $\alpha$-th ball $\disk(x_\alpha)$.
Then any vector field $X$ over $\cM_d$ is characterized by the derivatives of  these
(matrix and vector valued) functions, which we denote by
$X(\cE_i)$ and  $X(\cE(\alpha))$.
For any real function $H\in \Cinf(\cM_d)$,  we introduce the quasi-Hamiltonian vector field $X_H$ by
\begin{equation}
X_H(F):= \{ F, H\}_{\underline{z}} := P_{\underline{z}}(dF,dH),\quad \forall F\in \Cinf(\cM_d),
\label{B1}\end{equation}
using $P_{\underline{z}}$  defined  in \eqref{Eq:PencCor} in terms of fixed parameters
$\underline{z}=(z_{\alpha\beta})_{\alpha<\beta} \in \R^{d(d-1)/2}$.
Note that the $d=1$ case is also included, in which case $P_{\underline{z}}$ simply equals $P_{\cM_d}$.
We are interested in the quasi-Hamiltonian vector fields of the Hamiltonians having the form
\begin{equation}
H=  h\circ \cE_1
\quad \hbox{with}\quad  h\in \Cinf(\UU(n))^{\UU(n)}.
\label{B2}\end{equation}
These Hamiltonians constitute the ring
\begin{equation}
\fH:= \cE_1^*(\Cinf(\UU(n))^{\UU(n)}).
\label{B3}\end{equation}
It is well known that $\Cinf(\UU(n))^{\UU(n)}$ has functional dimension $n$.
Consequently, this holds for $\fH$, too.
As a final piece of preparation, for any $f\in \Cinf(\UU(n))$ we introduce
the $\uu(n)$-valued derivatives $\nabla f$ and $\nabla' f$ by
\begin{equation}
 \langle \xi, \nabla f(g) \rangle + \langle \xi', \nabla' f(g) \rangle := \dt f(e^{t\xi} g e^{t\xi'}),
 \quad \forall g\in \UU(n), \, \xi, \xi' \in \uu(n).
\label{B4}\end{equation}
For $\cF\in \Cinf(\UU(n) \times \UU(n))$ we define $\nabla_1 \cF$ and $\nabla_1' \cF$ by
\begin{equation}
 \langle \xi,  \nabla_1 \cF(A,B) \rangle + \langle \xi', \nabla_1' \cF(A,B) \rangle := \dt \cF(e^{t\xi} A e^{t\xi'},B),
 \quad \forall g\in \UU(n), \, \xi, \xi' \in \uu(n).
\label{B4+}\end{equation}
and similarly for the derivatives $\nabla_2 \cF$ and $\nabla_2'\cF$ with respect to the second variable.
In these definitions, we use the pairing on $\uu(n)$  given by \eqref{Eq:ipU}.

\begin{theorem}\label{Thm:B1}
The quasi-Hamiltonian vector field $X_H$ on $\cM_d$ associated with a Hamiltonian $H$ in \eqref{B2}
by means of \eqref{B1}
is characterized by the relations
\begin{equation}
X_H(\cE_2) = - \cE_2 \cdot (\nabla h \circ \cE_1), \quad X_H(\cE_1) = 0, \quad X_H(\cE(\alpha))=0,\,\, \forall \alpha=1,\dots, d.
\label{B5}\end{equation}
The integral curves $(A(t), B(t), v_1(t),\ldots, v_d(t))$ of $X_H$ are provided by
\begin{equation}
B(t) = B(0) \exp\left( - t \nabla h(A(0))\right), \quad A(t) = A(0), \quad v_\alpha(t) = v_\alpha(0),
\label{B6}\end{equation}
where the zero in the arguments refers to the  initial value.
The $\UU(n)$-equivariant smooth map $\Psi: \cM_d \to \cM_d$ defined by
\begin{equation}
\Psi: (A,B, v_1,\ldots, v_d) \mapsto (A, \tilde B, v_1,\dots, v_d) \quad\hbox{with}\quad \tilde B:= B A B^{-1}
\label{B7}\end{equation}
is constant along the integral curves \eqref{B6}, and thus the elements of $\Psi^*(\Cinf(\cM_d))$ are first
integrals. These first integrals form a ring of functional dimension $\dim(\cM_d) -n$.
\end{theorem}

In \eqref{B5} the dot denotes matrix multiplication.
The proof of the theorem relies on  two lemmas.

\begin{lemma}
\label{Lm:B2}
Let us consider the fusion product of two quasi-Poisson $G$-manifolds $M$ and $N$,
and calculate the quasi-Poisson bracket between an arbitrary function $F \in \Cinf(M\times N)$
and a function $\pi_M^*\cH$ for some $\cH \in \Cinf(M)^G$, where $\pi_M$ is the obvious projection.
Then, we have
\begin{equation}
\{F, \pi_M^*\cH\}_{M,N}^{\textrm{\emph{fus}}} = \{ F, \pi_M^* \cH\}_M,
\label{B8}\end{equation}
i.e., the bracket is determined by the pushforward of the bivector $P_M$ from $M$ to $M\times N$.
\end{lemma}
\begin{proof}
According to Proposition \ref{Pr:Fus}, we have
\be
\{F, \pi_M^*\cH\}_{M,N}^{\textrm{{fus}}}= \{ F,\pi_M^*\cH\}_M + \{ F,\pi_M^*\cH\}_N- \{F, \pi_M^* \cH\}_{P_{M,N}}.
\label{B9}\ee
The second term vanishes since $\pi_M^*\cH$ does not depend on the second factor in $M\times N$, and the third
term vanishes because of the invariance property of $\cH$.
\end{proof}

The statement of the next lemma is taken from \cite[Proposition 4.1]{Fe}.
Here, we prefer to  also present its proof.

\begin{lemma}\label{Lm:B3}
Let $\pi_1$ and $\pi_2$ be the projections from $\D(\UU(n))$ onto its first and second $\UU(n)$ factors,
and consider the function $\pi_1^* h$ for some $h\in \Cinf(\UU(n))^{\UU(n)}$.
Denoting the corresponding quasi-Hamiltonian vector field on $\D(\UU(n))$ by $Y_{\pi_1^* h}$, we have
\begin{equation}
Y_{\pi_1^* h}(\pi_1) =0 \quad \hbox{and}\quad Y_{\pi_1^* h}(\pi_2) = - \pi_2 \cdot (\nabla h \circ \pi_1).
\label{B10}\end{equation}
\end{lemma}
\begin{proof}
Let us first spell out the bivector field $P_{\D(\UU(n))}$ from \eqref{Eq:P-DUn} in detail.
Recalling the notations from the beginning of Section \ref{S:Back}, we now set
\begin{equation}
e_a^{1,L} := (e_a,0)^L, \quad
e_a^{2,L}:= (0,e_a)^L,\quad
e_a^{1,R}:= (e_a,0)^R, \quad
e_a^{2,R}: = (0,e_a)^R.
\label{B12}\end{equation}
Then, the infinitesimal action of $\uu(n)\times \uu(n)$ that comes from
the action of $\UU(n)\times \UU(n)$ on itself by $(g,h)\cdot (A,B)=(g A h^{-1} , h B g^{-1})$ can be decomposed as
\begin{equation}
(e_a, 0)_{\UU(n) \times \UU(n)} = e_a^{2,L} - e_a^{1,R},
\quad (0, e_a)_{\UU(n) \times \UU(n)} = e_a^{1,L} - e_a^{2,R}.
\label{B13}\end{equation}
By collecting terms, the formula \eqref{Eq:P-DUn} leads to
 \be
P_{\D(\UU(n))}= \frac{1}{2} \sum_{a\in \mathtt{A}}\left(e_a^{1,R}\wedge e_a^{1,L} - e_a^{2,R}\wedge e_a^{2,L} +
e_a^{1,L}\wedge (e_a^{2,L}+ e_a^{2,R}) +
 e_a^{1,R}\wedge (e_a^{2,L} - e_a^{2,R}) \right).
\label{B14}\ee
Since $\pi_1^*h$ does not depend on the second $\UU(n)$ factor, and $(e_a^{1,L} - e_a^{1,R})(\pi_1^* h)=0$ because of the invariance
property of $h$, we obtain for all $\cF\in \Cinf( \D(\UU(n)) )$
\begin{equation}
\{\cF,\pi_1^*h\}_{\D(\UU(n))} = - \sum_{a\in \mathtt{A}} e_a^{2,L}(\cF) \, e_a^{1,R}(\pi_1^* h) + \frac{1}{2} \sum_{a\in \mathtt{A}} (e_a^{1,R} - e_a^{1,L})(\cF) \, e_a^{1,R}(\pi_1^*h) .
\label{B15}\end{equation}
Now, notice from the definition \eqref{B4+} that (for $i=1,2$)
\begin{equation}
\nabla_i \cF = \sum_{a\in \mathtt{A}} e_a e_a^{i,R}(\cF), \quad  \nabla'_i \cF = \sum_{a\in \mathtt{A}} e_a e_a^{i,L}(\cF), \quad
\nabla h \circ \pi_1 =  \nabla_1 \pi_1^* h = \sum_{a\in \mathtt{A}} e_a e_a^{1,R}(\pi_1^*h).
\label{B16}\end{equation}
Therefore we can write
\begin{equation}
\sum_{a\in \mathtt{A}} (e_a^{1,R} - e_a^{1,L})(\cF) e_a^{1,R}(\pi_1^*h) = \langle \nabla_1 \cF - \nabla_1' \cF, \nabla h \circ \pi_1 \rangle.
\label{B17}\end{equation}
This vanishes since
\begin{equation}
 \nabla_1 \cF(A,B)  = A (\nabla_1' \cF(A,B)) A^{-1}, \quad \forall (A,B) \in \UU(n) \times \UU(n),
\label{B18}
\end{equation}
 by the definition \eqref{B4+}, and
$A^{-1}  \nabla h(A)  A = \nabla h(A)$,
because of the invariance property of $h$.
Then \eqref{B15} gives
\begin{equation}
\{\cF,\pi_1^*h\}_{\D(\UU(n))} = - \langle \nabla_2' \cF, \nabla h \circ \pi_1\rangle, \quad
\forall \cF\in \Cinf( \D(\UU(n)) ),
\label{B19}\end{equation}
which is easily seen to be equivalent to the claimed formula \eqref{B10} of the quasi-Hamiltonian vector field $Y_{\pi_1^* h}$.
\end{proof}

\begin{proof}[Proof of Theorem \ref{Thm:B1}]
The function $H$ in \eqref{B2} can be viewed as the pullback to $\cM_d$ of the function $\pi_1^* h\in \Cinf(\D(\UU(n)))^{\UU(n)}$.
For the fusion is associative, the quasi-Poisson bracket \eqref{B1} on $\cM_d$ satisfies
\begin{equation}
\br{-,-}_{\underline{z}}= \br{- ,-}_{M,N}^{\textrm{fus}} + \br{-,-}_{\psi_{\underline{z}}}
\label{B11}\end{equation}
with $M= \D(\UU(n))$ and $N= \disk(x_1)\circledast\cdots \circledast  \disk(x_d)$.
Referring to Lemma \ref{Lm:B2}, we get that
\begin{equation}
 X_H(F)=\{ F, \cE_1^* h \}_{\underline{z}} = \{ F, \cE_1^*h\}_{\D(\UU(n))},\quad \forall F\in \Cinf(\cM_d),
\label{Eq:XHF}
\end{equation}
where we also used that $\{F ,\cE_1^* h \}_{\psi_{\underline{z}}} =0$, as is clear from the form of $\psi_{\underline{z}}$ \eqref{Eq:PencCor}.
We see from \eqref{Eq:XHF} that the vector field $X_H$ on $\cM_d$  is the pushforward of the vector field
$Y_{\pi_1^* h}$ on $\D(\UU(n))$. Hence we obtain \eqref{B5} from \eqref{B10}.
The integral curves of the vector field \eqref{B5} satisfy
\begin{equation}
\dot{B}(t) =- B(t) \nabla h(A(t)), \quad \dot{A}(t)=0,\quad \dot{v}_\alpha(t) =0,
\end{equation}
which implies \eqref{B6}. One checks that the map $\Psi$ \eqref{B7} is constant along the integral curves by using that
$ A \nabla h(A) A^{-1} = \nabla h(A)$ holds for $h\in \Cinf(\UU(n))^{\UU(n)}$.
The statement about the functional dimension follows since $A$ and $\tilde B$
are arbitrary elements of $\UU(n)$ except that their eigenvalues coincide, which generically yields
$n$ relations. In other words, if $\UU(n)_{\reg}\subset \UU(n)$ is the dense open subset of regular elements,
then the image by $\Psi$ of $\UU(n)_{\reg} \times \UU(n) \times \disk(x_1)\times \cdots \times \disk(x_d)$ is a codimension $n$
submanifold of $\cM_d$.
\end{proof}

\begin{remark}
The Hamiltonian $H= h \circ \cE_2$ with $h\in \Cinf(\UU(n))^{\UU(n)}$ leads to
\begin{equation}
X_H(\cE_1) =\cE_1 \cdot (\nabla h \circ \cE_2), \quad X_H(\cE_2) = 0, \quad X_H(\cE(\alpha))=0,\,\, \forall \alpha=1,\dots, d,
\end{equation}
and one can easily write down the corresponding integral curves and constants of motion, similarly to Theorem \ref{Thm:B1}.
\end{remark}

\subsection{Algebra of constants of motion}

A key aspect of dynamics on Poisson manifolds is that the constants of motion
of any Hamiltonian form a Poisson subalgebra of the Poisson algebra of all smooth functions.
This statement generalizes to the constants of motion of a $G$-invariant Hamiltonian $H$ on any quasi-Poisson $G$-manifold $M$.
Indeed, it follows from \eqref{Eq:JacPhi} that the commutant of any $H\in \Cinf(M)^G$  forms a closed \emph{quasi-Poisson subalgebra}
of $(\Cinf(M), \br{-,-})$.
Our next result entails that the joint constants of motion $\Psi^*(\Cinf(\cM_d))$ found in Theorem \ref{Thm:B1}
form a quasi-Poisson subalgebra
in $(\Cinf(\cM_d), \br{-,-}_{\underline{z}})$, and it provides a neat characterization of this algebra.

\begin{theorem} \label{Thm:B5}
Consider the manifold $\cM_d = \UU(n) \times \UU(n) \times \disk(x_1) \times \cdots \times \disk(x_d)$
equipped with the same $\UU(n)$-action as before, but replace its bivector $P_{\underline{z}}$ \eqref{Eq:PencCor} by
\begin{equation}
P_{\underline{z},c} := P_{\D(\UU(n))}^c + P_{\mathrm{mix}} + P_{\mathrm{spin}} + \psi_{\underline{z}},
\label{Pmod}\end{equation}
where $P_{\mathrm{mix}}$ and $P_{\mathrm{spin}}$ are taken from \eqref{Eq:qPMd-b}, but $P_{\D(\UU(n))}$ \eqref{B14} is replaced by
\be
 P_{\D(\UU(n))}^{c} =\frac{1}{2} \sum_{a\in \mathtt{A}}\left(e_a^{1,R}\wedge e_a^{1,L} - e_a^{2,R}\wedge e_a^{2,L}
  - (e_a^{1,L} - e_a^{1,R})\wedge (e_a^{2,L} - e_a^{2,R}) \right).
\label{Pdeg} \ee
 This bivector defines a new quasi-Poisson structure on $\cM_d$ that admits the moment map
 \begin{equation}
\tilde{\Phi}:\MM_d\to \UU(n), \quad
\tilde{\Phi}(A,\tilde{B},v_\alpha)=A\tilde{B}^{-1}\exp(\ic x_1 v_1v_1^\dagger)\cdots \exp(\ic x_d v_dv_d^\dagger). \label{Eq:Tild-phi-1+L}
\end{equation}
 Then, $\Psi:(\MM_d, P_{\underline{z}},\Phi)\to (\MM_d, P_{\underline{z},c},\tilde{\Phi})$ given in \eqref{B7} is a quasi-Poisson map, i.e. it satisfies
\begin{equation}
\{ f_1\circ \Psi, f_2\circ \Psi \}_{\underline{z}} =\{f_1, f_2\}_{\underline{z},c} \circ \Psi, \quad \forall f_1, f_2 \in \Cinf(\MM_d),
\label{lacirel}
\end{equation}
with the respective quasi-Poisson brackets,
and it intertwines also the moment maps: $\tilde \Phi \circ \Psi= \Phi$.
 \end{theorem}
\begin{proof}
We repeat the previous construction of the quasi-Poisson structure of $\cM_d$ with the only modification
that instead of the standard structure on $\D(\UU(n))$ we start with the degenerate quasi-Poisson structure on
$\UU(n) \times \UU(n)$ exhibited in Proposition \ref{Pr:qP-AltDble}. It is easily checked that \eqref{Pdeg} is the bivector field
corresponding to this structure.
Therefore the modified construction yields a quasi-Poisson structure based on the bivector \eqref{Pmod}.
Recalling  that the moment map
of the degenerate structure on $\UU(n) \times \UU(n)$ sends $(A,\tilde B)\in \UU(n)\times \UU(n)$ to $A \tilde B^{-1}$,
we see that
the moment map of the quasi-Poisson structure  \eqref{Pmod} is given by \eqref{Eq:Tild-phi-1+L},
and the relation $\Phi = \tilde \Phi \circ \Psi$ then follows immediately from the formula \eqref{B7} of the map $\Psi$.

Now it remains to verify \eqref{lacirel}, and it is enough to do this for all matrix element and vector component functions.
  Comparison of the bivectors $P_{\underline{z}}$ and  $P_{\underline{z},c}$  shows that the only non-trivial cases are
  \begin{equation}
  \br{A_{ij}\circ \Psi  , \tilde{B}_{kl} \circ \Psi}_{\underline{z}} =
  \br{A_{ij}  , \tilde{B}_{kl}}_{\underline{z},c} \circ \Psi
  \quad\hbox{and}\quad
  \br{\tilde{B}_{ij}\circ \Psi  , \tilde{B}_{kl} \circ \Psi}_{\underline{z}} =
  \br{\tilde{B}_{ij}  , \tilde{B}_{kl}}_{\underline{z},c} \circ \Psi.
  \label{reqidL}\end{equation}
  We have $(\tilde B_{ij}\circ \Psi)(A,B) = (B A B^{-1})_{ij}$ and (by an obvious abuse of notation) $A_{ij} \circ \Psi = A_{ij}$.
Then the equalities \eqref{reqidL} are verified by a straightforward calculation based on the explicit formulae   \eqref{Eq:intDble}  of the
standard quasi-Poisson bracket on $\D(\UU(n))$ and their counterparts for the degenerate structure displayed in
\eqref{Eq:AltDble}, which are valid also
for the corresponding matrix element functions on $\cM_d$.
The details
of the required trivial calculation are omitted.
\end{proof}

\begin{remark}
As a corollary to the proof of Theorem \ref{Thm:B5}, the
map $\Psi:\D(\UU(n))\to \UU(n)\times \UU(n)$ given by $\Psi(A,B)=(A,BAB^{-1})$ is a quasi-Poisson map intertwining
moment maps when we endow these manifolds with the Hamiltonian quasi-Poisson structures of
Propositions \ref{Pr:qP-intDble} and \ref{Pr:qP-AltDble}.
Furthermore, we can adapt the argument from Remark \ref{Rem:Deg} to show that
the quasi-Poisson bivectors $P_{\underline{z},c}$ \eqref{Pmod} are  \emph{not} non-degenerate.
\end{remark}

\subsection{Recall and generalization of the notion of degenerate integrability} \label{ss:Def-Int}

We now explain that our master system characterized by Theorems \ref{Thm:B1} and \ref{Thm:B5}
possesses similar properties as do degenerate integrable systems
living on symplectic manifolds. Degenerate integrability, also called non-Abelian integrability,
was first discussed in the papers by Nekhoroshev \cite{Nek} and Mischenko and Fomenko \cite{MF}.
In recent papers by Reshetikhin {\it et al.} \cite{ARe,Re1,Re2,Re3,ReS} it is  referred to as superintegrability, a term originally
introduced for related quantum mechanical systems. In these and other works \cite{J,RudS} one finds slightly
different, but essentially equivalent definitions, and below we recall a variant that appears to us
as the most convenient. This subsection also prepares the ground for Section \ref{S:Red}, where
we study the reduction of our master system, and enquire about integrability after reduction.

Let us recall that the functional dimension of a ring $\fR$ of functions on a manifold $M$ is $r$
if the exterior derivatives of the elements of $\fR$ \emph{generically}
span an $r$-dimensional subspace of the cotangent space, meaning that there exists
a dense, open submanifold of $M$ where this is true.

Consider a \emph{symplectic} manifold $M$ and two subrings $\fH$ and $\fF$ of $\Cinf(M)$
obeying the following conditions:
\begin{enumerate}[itemsep=0pt]
\item{
The ring $\fH$ has functional dimension $r$ and $\fF$ has functional dimension $s$ such that
$r + s = 2m$ and  $r<m$ with $2m:= \dim(M)$.}
\item{Both $\fH$ and $\fF$ form Poisson subalgebras of $\Cinf(M)$,  satisfying
$\fH\subset \fF$  and $ \{\cF, \cH\}=0$ for all $\cF\in \fF$, $\cH\in \fH$.}
\item{The Hamiltonian vector fields of the elements of $\fH$ are complete.}
\end{enumerate}
If these conditions are satisfied,  then $(M, \br{-,-}, \fH, \fF)$ is called a \emph{degenerate
integrable system of rank $r$}, with
 ring of  Hamiltonians  $\fH$ and ring of constants of motion (first integrals)  $\fF$.

A few remarks are now in order.
First, degenerate integrability of a single Hamiltonian $\cH$ is  understood to mean
that there exist rings $\fH$ and $\fF$ with the above properties such that $\cH\in \fH$.
Second, observe that $\fF$ is either equal to or can
be enlarged to the commutant of $\fH$ in the Poisson algebra $(\Cinf(M), \br{-,-})$, keeping integrability.
Third, the condition on the completeness of the flows is superfluous, if the joint level surfaces
of the elements of $\fF$  are compact.
In the literature the definition is often formulated in terms of functions
$f_1,\dots, f_r, f_{r+1},\dots f_s$ so that they generate $\fF$ and the first $r$ of them generate $\fH$.
If the definition is modified by setting $r=s=m$ and $\fH = \fF$, then one obtains
the notion of Liouville integrability.
We here focused on the $\Cinf$ case,
but one can talk about real-analytic and algebraic integrable systems as well.

We now adapt the above definition to systems on quasi-Poisson manifolds that possess an additional
regularity property.  Let us call a quasi-Poisson manifold $(M, P_M)$ \emph{regular}
if the rank of the bivector field $P_M$ is equal to $\dim(M)$ on a dense open subset of $M$.
Note that $\dim(M)$ is then necessarily even,   and that
the non-degeneracy condition given before Theorem \ref{Thm:Corr}  does not imply regularity\footnote{The conjugacy class $\mathrm{C}_D$ of $D:=\diag(1,\ldots,1,-1)\in \UU(n)$ with the moment map $A\mapsto A\in \UU(n)$ is non-degenerate by \cite[Prop. 3.4]{AKSM}.
A simple computation (see e.g. the end of Section 3 in \cite{AKSM}) shows that the quasi-Poisson bivector $P$ of $\UU(n)$ evaluated at $g\in \UU(n)$ reads $P_g=\frac14 \sum_{a,b\in \mathtt{A}} (\Ad_{g^{-1}}-\Ad_g)_{ab} (e_a)^L_g \wedge (e_b)^L_g$, where $(\rho_{ab})_{a,b\in \mathtt{A}}$ denotes the decomposition of a linear map $\rho:\uu(n)\to \uu(n)$ in the orthonormal basis $(e_a)_{a\in \mathtt{A}}$.
This formula can be used for any conjugacy class $\mathrm{C}$, and in particular the bivector $P$ vanishes on $\mathrm{C}$ if $g^2=\1_n$ for all $g\in \mathrm{C}$.
This is the case for $\mathrm{C}_D$, which is therefore not regular because the rank of its quasi-Poisson bivector is identically zero while $\dim \mathrm{C}_D>0$.}.
According to Remark \ref{Rem:Reg},
 our master phase space   $(\cM_d, P_{\underline{z}})$  is regular  for arbitrary parameters $\underline{z}$.

By definition, a \emph{degenerate integrable system on a regular quasi-Poisson $G$-manifold} $(M, \br{-,-} )$
is a quadruple  $(M, \br{-,-}, \fH, \fF)$ with two subrings
$\fH$ and $\fF$ of $\Cinf(M)$ having the  properties:
\begin{enumerate}[label=\Roman*., itemsep=0pt]
\item The conditions 1., 2. and 3. of the symplectic case hold,
where Poisson brackets, Poisson subalgebras and Hamiltonian vector fields are replaced by
quasi-Poisson brackets, quasi-Poisson subalgebras and quasi-Hamiltonian vector fields, respectively.
\item The ring $\fH$ consists of invariant functions, i.e. it is a subring of $\Cinf(M)^G$ with respect to
the action of the Lie group $G$ used in the definition of the quasi-Poisson structure.
\end{enumerate}
The second part of the definition ensures that the quasi-Hamiltonian vector fields of the elements of $\fH$ commute
due to \eqref{Eq:JacPhi}, as noticed e.g. in \cite[Lemma 18]{LB}.
Liouville integrability on a regular quasi-Poisson manifold  is also defined similarly to the symplectic case, adding condition II.

We see from Theorems \ref{Thm:B1} and \ref{Thm:B5} that our master system associated with $\UU(n)$
is  a degenerate integrable system of rank $n$ on $(\cM_d,\br{-,-}_{\underline{z}})$
in the sense of the definition just presented.

\medskip

The extension of the notions of Liouville and degenerate integrability
to \emph{bona fide} Poisson manifolds is also interesting, and had been studied in the
literature \cite{LGMV}.

Let $(N,\br{-, -})$  be a Poisson manifold whose Poisson tensor has maximal rank $2m\leq \dim(N)$
on a dense open submanifold.  By definition,  $(N,\br{-,-},\fH,\fF)$ is a degenerate integrable
system if the following requirements are met:
\begin{enumerate}[label=\roman*., itemsep=0pt]
\item
The subrings $\fH$ and $\fF$  of $\Cinf(N)$ have functional dimensions $r<m$ and  $s$, respectively, such that
$r + s = \dim(N)$  and  $r<m$.
\item The conditions 2. and 3. of the symplectic case hold.
\item Additionally,  the Hamiltonian vector fields of $\fH$ span an $r$-dimensional subspace of
the tangent space over a dense open submanifold of $N$.
\end{enumerate}

Notice that in the definition of degenerate integrability on Poisson manifolds
we  could have unified part of the latest conditions by simply requiring that the functional
dimension of the space of Hamiltonian vector fields  associated with $\fH$ is $r$, and  $r+s=\dim(N)$, $r<m$ hold.
Note also that Liouville integrability in the Poisson case is obtained by imposing $r=m$
instead of $r<m$ in the definition.

These notions of integrability `have the right to exist' because there is a plethora of interesting examples,
and because they support deep theorems on the structure of integrable systems on symplectic and
(to some extent also) on Poisson manifolds. For such theorems we refer to the literature \cite{J,LGMV,RudS}, and
 raise the issue
of their generalization to the quasi-Poisson/quasi-Hamiltonian  setting as a problem for future work.

\section{Hamiltonian reduction of the master system}
\label{S:Red}

In the previous section we presented a degenerate integrable system on the
non-degenerate, regular quasi-Poisson manifold $(\cM_d,\br{-,-}_{\underline{z}})$.
Now we study the reduction of this master system defined by taking quotient
by the symmetry group $\UU(n)$.
To avoid technical difficulties, we shall focus on the subset
$\cM_{d*}$ of $\cM_d$ on which the action of $\UU(n)$ is free.
For the underlying general reduction theory, the reader may consult \cite{AKSM,AMM},
which contain the following fundamental results.

\begin{proposition}[\cite{AKSM}] \label{Pr:Red}
Assume that $(M,\br{-,-},\Phi)$ is a Hamiltonian quasi-Poisson manifold for an action of $G$.
For each conjugacy class $\mathrm{C}\subset G$ such that $\Phi^{-1}(\mathrm{C})\neq \emptyset$ and
$G$ acts freely on
$\Phi^{-1}(\mathrm{C})$, the quotient  $\Phi^{-1}(\mathrm{C})/G$ inherits a structure of Poisson manifold.
\end{proposition}

Proposition \ref{Pr:Red} holds because the quasi-Poisson bracket $\br{-,-}$ on $\Cinf(M)$ restricts
to a Poisson bracket on $\Cinf(M)^G$ since $\phi_M$ in \eqref{Eq:JacPhi} vanishes on $G$-invariant functions.
Moreover, \eqref{Eq:momap} guarantees that each $F\in \Cinf(M)^G$ induces a quasi-Hamiltonian vector field  $X_F$ that is tangent to
the level sets of the moment map.
Similarly,  the following result holds.

\begin{proposition}[\cite{AMM}] \label{Pr:Red-omega}
Assume that $(M,\omega,\Phi)$ is a quasi-Hamiltonian manifold for an action of $G$.
For each conjugacy class $\mathrm{C}\subset G$ such that $\Phi^{-1}(\mathrm{C})\neq \emptyset$ and $G$ acts freely on $\Phi^{-1}(\mathrm{C})$,
the quotient    $\Phi^{-1}(\mathrm{C})/G$ inherits a structure of symplectic manifold.
\end{proposition}

Furthermore, in the case where $M$ is equipped with a non-degenerate quasi-Poisson bivector $P$ and a $2$-form $\omega$
corresponding to each other through \eqref{Eq:corrPOm},  the Poisson bivector and symplectic
 form on  $\Phi^{-1}(\mathrm{C})/G$ appearing in the previous two propositions correspond to each other in the usual sense \cite[Prop. 10.6]{AKSM}.
The Hamiltonian vector field induced  by a reduced Hamiltonian arising from $F\in \Cinf(M)^G$ results also via
  projection of the original quasi-Hamiltonian vector field $X_F$, quite in the same way as in
  classical Poisson/symplectic reduction \cite{OR}.

\subsection{The reduced Poisson space and its generic symplectic leaves}

We start with a basic lemma.

\begin{lemma}\label{Lm:LR1}
The subset $\cM_{d*}\subset \cM_d$ on which the $\UU(n)$-action is free is a dense, open, $\UU(n)$-invariant submanifold that is also invariant
under the flows of all the quasi-Hamiltonian vector fields associated with the elements of $\Cinf(\cM_d)^{\UU(n)}$.
\end{lemma}

\begin{proof}
It is easily seen that $\cM_{d*}$ is non-empty. Thus its open and dense character follows from
a general result about the principal orbit type for the actions of compact Lie groups \cite{Br,DK}; and it
can also be seen directly.  The flow $\varphi_t$ of $X_H$ associated with any $H\in \Cinf(\cM_d)^{\UU(n)}$ by \eqref{B1} satisfies
\be
\varphi_t \circ \cA_g = \cA_g \circ \varphi_t, \quad \forall g\in \UU(n),
\label{R1}\ee
where $\cA_g: \cM_d \to \cM_d$ is the action of $g\in \UU(n)$ \eqref{Eq:act-Mast}.
This implies that the isotropy group does not change along the integral curves.
\end{proof}

The quotient space
\be
\cM_{d*}^\red:= \cM_{d*}/\UU(n)
\label{R2}\ee
is a smooth Poisson manifold,
whose Poisson algebra results by means of the natural identification
\be
\Cinf(\cM_{d*}^\red) \equiv \Cinf(\cM_{d*})^{\UU(n)}.
\label{R3}\ee
We denote the resulting Poisson manifold $(\cM_{d*}^\red, \br{-,-}_{\underline{z}}^\red)$.
Let
\be
\pi: \cM_{d*} \to \cM_{d*}^\red
\label{R4}\ee
be the projection to the quotient space.
The restriction of any Hamiltonian $H\in \fH$ \eqref{B3} on $\cM_{d*}$ is still denoted by $H$, and it gives
rise to the reduced Hamiltonian $H_\red \in \Cinf(\cM_{d*}^\red)$ obeying the identity
\be
H= H_\red \circ \pi.
\label{R5}\ee
The Hamiltonian vector field $X_{H_\red}$ induced by the reduced Poisson structure
satisfies
\be
X_{H_{\red}} = \pi_* (X_H).
\label{R6}\ee
Its flow is  obtained by applying the projection $\pi$ to the flow of $X_H$ on $\cM_{d*}$.

We below explain that a dense open subset of $\cM_{d*}^\red$ is filled by symplectic leaves
of codimension $n$.
Observe from the moment map condition \eqref{Eq:momap} that the $\UU(n)$-invariant functions of the form
\be
h \circ \Phi \quad\hbox{with}\quad h \in \Cinf(\UU(n))^{\UU(n)}
\label{R7}\ee
Poisson commute with all invariant functions.  Therefore, they yield Casimir functions
of the reduced Poisson bracket. The corresponding joint level surfaces in $\cM_{d*}^\red$ are
the quotient spaces of the form
\be
\cM_{d*}^\red(\cC) := (\Phi^{-1}(\cC)\cap \cM_{d*})/\UU(n),
\label{R8}\ee
where $\cC$ is a conjugacy class of $\UU(n)$ that lies in the image of $\cM_{d*}$ under the moment map $\Phi$.
The general theory  tells us  that the sets $\cM_{d*}^\red(\cC)$ are smooth symplectic manifolds, whose
connected components are the symplectic leaves of the Poisson manifold
 $(\cM_{d*}^\red, \br{-,-}_{\underline{z}}^\red)$.

\begin{lemma}\label{Lm:LR2}
 The subset $\Phi^{-1}(\UU(n)_\reg) \cap \cM_{d*}$ is dense and open in $\cM_{d*}$.
 \end{lemma}

 \begin{proof}
 It is known \cite{Go} that any element of a connected, compact, semisimple Lie group can be written as a
 group commutator.
 By using this, it is not difficult to see that
 the moment map $\Phi:\cM_d \to \UU(n)$  \eqref{Eq:Mast-phi-1}
is surjective (cf. the beginning of \S\ref{ss:Cmpct}).
 Consider the (complex-valued) real-analytic function on $\cM_d$ defined by
 the discriminant\footnote{Recall that the discriminant of a polynomial $p(\lambda)$
 is a certain polynomial in the coefficients of $p(\lambda)$ that vanishes if and only if not all the roots of $p(\lambda)$
 are simple \cite[Ch.~2]{Pr}.}
 of the characteristic polynomial $\det(\Phi - \lambda \1_n)$. This is not identically zero,
 and thus the complement of its zero set is dense and open.
 The complement in question is equal to $\Phi^{-1}(\UU(n)_\reg)\subset \cM_d$.
 Since the intersection of two dense open sets is also dense open, the claim follows.
 \end{proof}

\begin{proposition}
\label{Pr:LR3}
A dense open subset of $\cM_{d*}^\red$ is filled by symplectic leaves of codimension $n$,
which are the connected components of the manifolds \eqref{R8}, where $\cC$ runs over the
conjugacy classes of $\UU(n)$ that lie in $\Phi(\cM_{d*})\cap \UU(n)_\reg$.
\end{proposition}

\begin{proof}
It is clear from Lemma \ref{Lm:LR2}  that $\pi( \Phi^{-1}(\UU(n)_\reg) \cap \cM_{d*} )$ is dense and open in $\cM_{d*}^\red$.
Thus we only need to compute the dimension of the manifolds \eqref{R8}.
 For this we use the alternative
realization \cite{AMM} of $\cM_{d*}^\red(\cC)$  given by
\be
\cM_{d*}^\red(\cC) \equiv (\Phi^{-1}(\mu)\cap \cM_{d*})/\UU(n)_\mu,
\label{R9}\ee
where $\mu$ is an arbitrarily chosen element of $\cC$, with isotropy group $\UU(n)_\mu$.
By the free action property, $\mu$ is a regular value of the restriction of the moment map on $\cM_{d*}$.
This implies that $\Phi^{-1}(\mu)\cap \cM_{d*}$ is a submanifold satisfying
\be
\dim( \Phi^{-1}(\mu)\cap \cM_{d*}) = \dim(\cM_{d*}) - n^2 = \dim(\cM_{d*}^\red) = 2nd + n^2.
\label{R10}\ee
 Consequently,
\be
\dim((\Phi^{-1}(\mu)\cap \cM_{d*})/\UU(n)_\mu) = \dim(\cM_{d*}^\red ) - \dim(\UU(n)_\mu).
\label{R11}\ee
If $\mu \in \UU(n)_\reg$, then $\UU(n)_\mu$ is a maximal torus, whose dimension is $n$,
whereby the proof is complete.
\end{proof}

\begin{remark}
The dense open subset  $( \cM_{d*} \cap \Phi^{-1}(\UU(n)_\reg) ) /\UU(n)$ of $\cM_{d*}/\UU(n)$ can be characterized by the property
that the differentials of the Casimir functions
 arising from \eqref{R7} span an $n$-dimensional subspace  of the cotangent space at each
point of this subset.
(One sees this by using Lemma \ref{Lm:LR4} presented below.)
The example treated in Section \ref{S:spinRS}  shows that this is a proper subset of $\cM_{d*}/\UU(n)$.
We suspect that $\Phi(\cM_{d*})$ contains the whole of $\UU(n)_\reg$, but have not proved this.
It is also worth remarking that $\cM_{d*}^\red$ is a dense open subset of the full quotient space
$\cM^\red_d=\cM_d/\UU(n)$,  simply  because the canonical projection from $\cM_d$ to $\cM_d^\red$ is both continuous and open.
Moreover,  it follows  from the principal orbit theorem (Theorem 2.8.5 in \cite{DK}) that $\cM_{d*}^\red$ is connected.
\end{remark}

\begin{remark}
For an arbitrary conjugacy class $\cC$,  the quotient space $\Phi^{-1}(\cC)/\UU(n)$  is in general a disjoint union of smooth
strata corresponding to the different orbit types
of the $\UU(n)$-action.
The connected components of these strata are the symplectic leaves of the singular Poisson space $\cM_d/\UU(n)$,
whose smooth functions descend from $\Cinf(\cM_d)^{\UU(n)}$.
This follows from the theory of singular Hamiltonian reduction \cite{OR,RudS,SL,Sn} that works
for quasi-Hamiltonian spaces as well.
\end{remark}

\subsection{The commuting reduced Hamiltonians and their Hamiltonian vector fields}  \label{ss:RedVF}

First, we wish to show that the commuting reduced Hamiltonians coming from $\fH$ \eqref{B3} form a ring of functional
dimension $n$ on $\cM_{d*}^\red$. This hinges on the well-known  result formulated in the  following lemma. To be self-contained,
we shall prove this in Appendix \ref{sec:J}.

\begin{lemma}\label{Lm:LR4}
The exterior derivatives of the elements of $\Cinf(\UU(n))^{\UU(n)}$ span an $n$-dimensional
subspace of $T_g^*{\UU(n)}$ if and only if  $g$ belongs to $\UU(n)_\reg$.
Thus, fixing any $g\in \UU(n)_\reg$, there exist $h_1,\dots, h_n$ in $\Cinf(\UU(n))^{\UU(n)}$ such that
$dh_1(g),\dots, dh_n(g)$ are linearly independent, and $dh(g)$ is a linear combination of them for any
$h\in \Cinf(\UU(n))^{\UU(n)}$.
Consequently, for any $h\in \Cinf(\UU(n))^{\UU(n)}$, there exists an open neighbourhood of $g$ where $h$
can be expressed as a function $h=f(h_1,\dots, h_n)$ with a smooth real function $f$ defined locally
in a neighbourhood of  $(h_1(g),\dots, h_n(g)) \in \R^n$.
\end{lemma}

For $i=1,2$, let us introduce the sets
\be
\cM_{d*}^{i,\reg}:= \cE_i^{-1}(\UU(n)_\reg) \cap \cM_{d*},
\label{R12}\ee
with the projections $\cE_i$ defined above \eqref{B1}, and
\be
\cM_{d*}^{i,\red}:= \pi(\cM_{d*}^{i,\reg}).
\label{R13}\ee
Clearly, these are dense open subsets of $\cM_{d*}$ and $\cM_{d*}^\red$, respectively.

\begin{proposition}\label{Pr:LR5}
Fix $i\in \{1,2\}$ and pick an arbitrary  $x\in \cM_{d*}^{i,\reg}$.
Then  there exist $n$ functions $h_k \in \Cinf(\UU(n))^{\UU(n)}$, $k=1,\dots, n$,
such that the exterior derivatives of the functions $H_k := \cE_i^*(h_k)  \in \Cinf(\cM_d)^{\UU(n)}$ are linearly independent at $x$.
These functions give rise to the reduced Hamiltonians $H_{k,\red} \in \Cinf(\cM_{d*}^\red)$ defined by
\be
H_{k,\red} (\pi(y)) = H_k(y), \qquad
\forall y\in \cM_{d*},
\label{R14}\ee
and the exterior derivatives of these reduced Hamiltonians are linearly independent at $\pi(x)\in \cM_{d*}^{i,\red}$.
As a consequence, the functional dimension of the ring of the reduced Hamiltonians on $\cM_{d*}^\red$ that
descends from the ring $\fH_i :=\cE_i^*\left( \Cinf(\UU(n))^{\UU(n)}\right)$ equals $n$.
\end{proposition}

\begin{proof}
We start by noting that the projection $\cE_i: \cM_{d} \to \UU(n)$ is a surjective submersion, and therefore the transpose
of its derivative gives an injective linear map from $T^*_{\cE_i(x)}\!\UU(n)$ to $T^*_x\!\cM_d$ at every $x\in \cM_d$.
If $x\in \cM_d^{i,\reg}$, then $\cE_i(x)\in \UU(n)_\reg$, and thus there exist smooth class functions $h_1,\dots, h_n$ on $\UU(n)$
whose exterior derivatives are independent at $\cE_i(x)$. Since $d \cE_i^*(h_k) = \cE_i^* (dh_k)$,
the covectors $dH_k(x)$ are linearly independent for $k=1,\dots,n$.  Since $\pi: \cM_{d*}\to \cM_{d*}^\red$ is also
a surjective submersion, the linear independence of the exterior derivatives of the functions $H_{k,\red}$ at $\pi(x)$ follows
from the linear independence of the exterior derivatives of the $H_k$ at $x$.
The final conclusion holds since $\cM_{d*}^{i,\red}$ is a  dense open subset of $\cM_{d*}^\red$.
\end{proof}

Our next goal is to develop a convenient characterization  of the reduced Hamiltonian vector fields
associated with the Hamiltonians $H\in \fH$ \eqref{B3}.
To this end, we restrict ourselves to the dense open submanifold $\cM_{d*}^{2,\reg}$ \eqref{R12} and introduce
also its submanifold
\be
\cM_{d*}^{2,\reg,0}:= \{ (A,Q, v_1,\ldots,v_d)\in \cM_{d*}\mid Q\in \bT(n)_\reg \},
\label{R15}\ee
where $\bT(n)$ is the diagonal subgroup of $\UU(n)$ and $\bT(n)_\reg = \bT(n) \cap \UU(n)_\reg$.
Observe that every $\UU(n)$ orbit in $\cM_{d*}^{2,\reg}$ intersects $\cM_{d*}^{2,\reg,0}$, and the intersection
is given by an orbit of the normalizer $\fN(n) < \UU(n)$ of $\bT(n)$,
\be
\fN(n)= \{ g\in \UU(n)\mid g \bT(n) g^{-1} = \bT(n)\}.
\label{R16}\ee
 Consequently, we obtain the
identification
\be
\cM_{d*}^{2,\red}:= \cM_{d*}^{2,\reg}/\UU(n) \equiv \cM_{d*}^{2,\reg,0}/\fN(n).
\label{R17}\ee
The quotient  by $\fN(n)$  can be taken in two consecutive steps, which gives
\be
\cM_{d*}^{2,\red} = \tilde \cM_{d*}^{2,\red}/S_n
\quad \hbox{with}\quad
\tilde \cM_{d*}^{2,\red}:=  \cM_{d*}^{2,\reg,0}/\bT(n),
\label{R38}\ee
using   the permutation group $S_n = \fN(n)/\bT(n)$.
We shall present the reduced Hamiltonian vector fields on $\cM_{d*}^{2,\red}$ as the projections of vector fields on
$\cM_{d*}^{2,\reg,0}$;
the different spaces constructed so far are related through the commutative Diagram \ref{Diag:big}.

\begin{figure}[h]
\centering
  \captionsetup{width=.8\linewidth}
   \begin{tikzpicture}
 \node (A)  at (-1.5,1.2) {$\UU(n)$};  
 \node (B)  at (1.5,1.2) {$\UU(n)_{\reg}$}; 
 \node (R)  at (4.5,1.2) {$\mathbb{T}(n)_{\reg}$}; 
 \node (C)  at (-1.5,-1.2) {$\cM_{d \ast}$}; 
 \node (D)  at (1.5,-1.2) {$\cM_{d \ast}^{2,\reg}$};  
 \node (S)  at (4.5,-1.2) {$\cM_{d \ast}^{2,\reg,0}$};  
 \node (E)  at (-1.5,-3.6) {$\cM_{d \ast}^{\red}$}; 
 \node (F)  at (1.5,-3.6) {$\cM_{d \ast}^{2,\red}$};  
 \node (T)  at (4.5,-3.6) {$\tilde{\cM}_{d \ast}^{2,\red}$};  
  \path[{Hooks[left]}->] (B) edge (A);
  \path[{Hooks[left]}->] (D) edge (C);
  \path[{Hooks[left]}->] (R) edge (B);
  \path[{Hooks[left]}->] (S) edge (D);
  \path[{Hooks[left]}->] (F) edge (E);
  \path[->] (C) edge node[left] {${\cE}_2$}  (A);
  \path[->] (D) edge node[left] {${\cE}_2$}  (B);
  \path[->] (S) edge node[left] {${\cE}_2^0$}  (R);
  \path[->] (C) edge node[right,gray,font=\scriptsize] {$/\UU(n)$}  (E);
  \path[->] (D) edge node[right,gray,font=\scriptsize] {$/\UU(n)$}  (F);
  \path[->] (S) edge node[right,gray,font=\scriptsize] {$/\mathcal{N}(n)$}  (F);
  \path[->] (S) edge node[right,gray,font=\scriptsize] {$/\mathbb{T}(n)$}  (T);
  \path[->] (T) edge node[below,gray,font=\scriptsize] {$/S_n$}  (F);
 \end{tikzpicture}
 \caption{The open dense submanifold $\MM_{d\ast}\subset \MM_d$ on which $\UU(n)$ acts freely is represented together with two of
  its subsets in the middle row. These subsets are defined using the projection $\cE_2:(A,B,v_\alpha)\mapsto B$ by requiring
   that the image sits in the regular part $\UU(n)_{\reg}$ of $\UU(n)$ or in the regular part $\mathbb{T}(n)_{\reg}$ of $\mathbb{T}(n)$, as depicted in
    the upper-half part of the diagram. The bottom row represents the subsets obtained by performing reduction with respect to the free action \eqref{Eq:act-Mast} of $\UU(n)$
    (or one of its subgroups, as indicated in gray along the projection maps).}
 \label{Diag:big}
\end{figure}

In Theorem \ref{B1} we emphasized the derivation aspect of the unreduced vector fields, but now it
will be more useful to view a vector field as a section of the tangent bundle. A vector field
$V$ on $\cM_d$ is a map $V: \cM_d \to T\cM_d$ that assigns to $m\in \cM_d$ a tangent vector $V_m\in T_m \cM_d$, which
can be presented as  $V_m = (V_m^1, V_m^2, V_m(1),\dots, V_m(d))$ with
\be
V_m^1\in T_A\UU(n), \,\, V_m^2 \in T_B\UU(n),\,\, V_m(\alpha) \in T_{v_\alpha} \CC^n,
\,\, \hbox{at}\,\, m = (A,B, v_1,\dots, v_d) \in \cM_d,
\label{R18}\ee
where $\alpha =1,\dots, d$.
In particular, the vector field $V:= X_H$ of \eqref{B5} takes the form
\be
V_m^1 =0,\,\, V_m^2 = - B \nabla h(A),\,\, V_m(\alpha) = 0\,\,\hbox{for}\,\,\alpha=1,\dots, d.
\label{R19}\ee
Similarly, a vector field $W$ on $\cM_{d*}^{2,\reg, 0}$ is given by
$W_x=(W_x^1, W_x^2, W_x(1),\dots, W_x(d))$ with
\be
W_x^1\in T_A\UU(n), \,\, W_x^2 \in T_Q\bT(n),\,\, W_x(\alpha) \in T_{v_\alpha} \CC^n,
\,\, \hbox{at}\,\,
x= (A,Q, v_1,\dots, v_d) \in \cM_{d*}^{2,\reg, 0}.
\label{R20}\ee

Let us remark that two vectors $V_m, U_m\in T_m \cM_{d*}$ at $m=(A,B,v_1,\dots, v_d) \in \cM_{d*}$
 satisfy
\be
D\pi(m)(V_m) = D\pi(m)(U_m)
\label{R21}\ee
if and only if the relations
\be
U_m^1 = V_m^1 + [\xi_m,A],\,\,
U_m^2 = V_m^2 + [\xi_m, B],\,\,
U_m(\alpha) = V_m(\alpha) + \xi_m v_\alpha,\,\,\alpha=1,\dots, d,
\label{R22}\ee
hold with some $\xi_m \in \uu(n)$.
Indeed, the kernel of the derivative map $D\pi(m)$ is
the tangent space to the $\UU(n)$ orbit through $m$.

Let $\fT(n)$ denote the Lie algebra of $\bT(n)$, and consider the orthogonal decomposition
\be
\uu(n) = \fT(n) + \fT(n)_\perp.
\label{R23}\ee
Correspondingly, we can decompose any $\xi \in \uu(n)$ as
\be
\xi = \xi_\fT + \xi_\perp
\quad\hbox{with}\quad\xi_\fT\in \fT(n),\,\, \xi_\perp \in \fT(n)_\perp.
\label{R24}\ee
Then, for any $Q\in \bT(n)_\reg$, we introduce the linear operator $\cR(Q)$ on $\uu(n)$ by the
formula
\be
\cR(Q) \xi:= \frac{1}{2}  \left((\Ad_Q - \id)\vert_{\fT(n)_\perp}\right)^{-1}\circ (\Ad_Q + \id) ( \xi_\perp),
\label{R25}\ee
where we used that $(\Ad_Q -\id) \in \End(\uu(n))$
is invertible on $\fT(n)_\perp$ for any $Q\in \bT(n)_\reg$.  Note that $\cR(Q)$ is a well-known classical dynamical $r$-matrix \cite{EV}.

\begin{proposition}\label{Pr:LR6}
Let $\cE_i^0$ $(i=1,2)$ and $\cE^0(\alpha)$ $(\alpha=1,\dots, d)$ be the restrictions on $\cM_{d*}^{2,\reg,0}$ \eqref{R15}
of the maps $\cE_i$ and $\cE(\alpha)$ defined above equation \eqref{B1}, and let
 $\pi^0: \cM_{d*}^{2,\reg,0} \to \cM_{d*}^{2, \red}$  \eqref{R17} be the natural projection.
 For any $H\in \fH$ \eqref{B3},
 consider the restriction of the quasi-Hamiltonian vector field $X_H$ \eqref{B5} on $\cM_{d*}^{2,\reg}$ and denote
 by $X_H^{2,\red}$ its projection on $\cM_{d*}^{2,\red}$. Then, we have
 \be
 X_H^{2,\red} = (\pi^0)_*(Y_H),
 \label{R26}\ee
 where $Y_H$ is the vector field on $\cM_{d*}^{2,\reg,0}$ defined as follows:
  \bea
&& Y_H(\cE_1^0) = [ \cE_1^0, (\cR\circ \cE_2^0)(\nabla h \circ \cE_1^0) + \zeta_\fT] , \nonumber\\
&&  Y_H(\cE_2^0) = -\cE_2^0 \cdot \left(\nabla h \circ \cE_1^0\right)_\fT  ,
\label{R27}\\
&& Y_H(\cE^0(\alpha))= -\bigl((\cR\circ \cE_2^0 + \frac{1}{2})(\nabla h \circ \cE_1^0) + \zeta_\fT\bigr)\cdot \cE^0(\alpha),
\quad  \alpha=1,\dots,d,
\nonumber
 \eea
with an arbitrary $\fT(n)$-valued function  $\zeta_\fT$ on $\cM_{d*}^{2,\reg,0}$.
In other words, the vector field $W:= Y_H$ on $\cM_{d*}^{2,\reg,0}$ has the components, at $m=(A,Q, v_1,\dots, v_d)$,
 \bea
&& W_m^1 =[A, \cR(Q)(\nabla h(A)) + \zeta_\fT(m)] , \nonumber\\
&&  W_m^2  =- Q (\nabla h(A))_\fT  ,
\label{R28}\\
&& W_m(\alpha)= -\bigl((\cR(Q) + \frac{1}{2})(\nabla h(A)) + \zeta_\fT(m)\bigr) v_\alpha,\quad  \alpha=1,\dots,d.
\nonumber
 \eea
 \end{proposition}
\begin{proof}
The restriction of the vector field $V\equiv X_H$ \eqref{B5} on $\cM_{d*}^{2,\reg}$ is projectable
on $\cM_{d*}^{2,\red}$.  On account of \eqref{R17},  the projected vector field can be obtained by applying the derivative
$D\pi(m)$ to $V_m$ \eqref{R19}  at all $m\in \cM_{d*}^{2,\reg,0}$, i.e., we have
\be
X_H^{2,\red}(\pi(m)) = D\pi(m) (V_m).
\label{R29}\ee
Since $D\pi^0(m)$ is the restriction of $D\pi(m)$ to the linear subspace
\be
T_m \cM_{d*}^{2,\reg, 0} < T_m \cM_{d*}^{2,\reg} = T_m \cM_{d*},
\label{R30}\ee
we obtain
\be
(\pi^0)_*(Y_H)(\pi(m)) = D\pi(m) (W_m).
\label{R31}\ee
Thus, we have to show that the tangent vectors $V_m$ \eqref{R19} and $W_m$ \eqref{R28} satisfy
\be
D\pi(m)(V_m) = D\pi(m)(W_m), \quad \forall m = (A,Q,v_1,\dots, v_d) \in \cM_{d*}^{2,\reg,0}.
\label{R32}\ee
Define
\be
\xi_m:= - (\cR(Q) + \frac{1}{2})(\nabla h(A)) + \zeta_\fT(m),
\label{R33}\ee
which is an element of $\uu(n)$.
For \eqref{R31},
 it is enough to check that  the following relations hold:
\be
W_m^1 = V_m^1 + [\xi_m,A],\,\,
W_m^2 = V_m^2 + [\xi_m, Q],\,\,
W_m(\alpha) = V_m(\alpha) + \xi_m v_\alpha.
\label{R34}\ee
Now, the third relation is obvious, for all $\alpha$, and the first one follows since $[\nabla h(A), A]=0$ because $h\in \Cinf(\UU(n))^{\UU(n)}$.
As for the second relation, we must check that
\be
- Q (\nabla h(A))_\fT  = - Q (\nabla h(A))  + [Q, (\cR(Q) + \frac{1}{2}) (\nabla h(A))_\perp].
\label{R35}\ee
This is equivalent to the identity
\be
(\nabla h(A))_\perp = (1 - \Ad_{Q^{-1}}) \circ (\cR(Q) + \frac{1}{2}) (\nabla h(A))_\perp.
\label{R36}\ee
But on $\fT(n)_\perp$ we can write
\be
(1 - \Ad_{Q^{-1}}) \circ (\cR(Q) + \frac{1}{2}) = \Ad_{Q^{-1}} \circ (\Ad_Q -1)\circ (\cR(Q) + \frac{1}{2})  = \id,
\label{R37}\ee
and hence the proof is complete. Incidentally, the proof also shows that $Y_H$ given above is the most general
vector field that verifies the identity \eqref{R32}, which is equivalent to \eqref{R26}.
\end{proof}

\begin{remark}  \label{Rem:Flow}
Fix $H\in \fH$ \eqref{B3}, which we write as $H=\cE_1^*h$ for $h\in \Cinf(\UU(n))^{\UU(n)}$.
The integral curve starting at $(A^0,Q^0,v_\alpha^0)\in \cM_{d\ast}^{2,\reg,0}$ of the vector field $Y_H$ \eqref{R27}
with $\zeta_{\fT}=-\frac12 \nabla h(A^0)$ can be written as follows.
Let $\epsilon>0$ (possibly $\epsilon=\infty$) be such that
\begin{equation} \label{Eq:FlowBt}
 B(t):=Q^0\, \exp(-t \nabla h(A^0)) \in \UU(n)_{\reg}\,, \qquad \text{for } -\epsilon<t<\epsilon\,.
\end{equation}
Consider the unique smooth curve $C(t)\in \UU(n)$, $-\epsilon<t<\epsilon$, satisfying
\begin{equation} \label{Eq:FlowQt}
 Q(t):=C(t)B(t)C(t)^{-1}\in \mathbb{T}(n)_{\reg} \,, \quad
 C(0)=\1_n,\,\,\, (\dot{C}(t)C(t)^{-1})_{\fT}=0\,.
\end{equation}
Such a curve $C(t)$ can be found by quadrature.
The vanishing of $(\dot{Q}(t)Q(t)^{-1})_\perp$ is equivalent to the condition
\begin{equation} \label{Eq:FlowCond}
 (\dot{C}(t)C(t)^{-1})_\perp = \left((\id-\Ad_{Q(t)})\vert_{\fT(n)_\perp}\right)^{-1} \circ \Ad_{Q(t)}
 (\nabla h(C(t) A^0 C(t)^{-1}))_\perp\,,
\end{equation}
since $C\nabla h(A)C^{-1}=\nabla h(CAC^{-1})$ for any $A,C\in \UU(n)$.
Thanks to \eqref{Eq:FlowCond}, we easily verify that $(A(t),Q(t),v_\alpha(t))$
provides the desired integral curve for $Q(t)$ given by \eqref{Eq:FlowQt} and
\begin{equation}
 A(t):=C(t) A^0 C(t)^{-1}\,, \quad v_\alpha(t):= C(t) v_\alpha^0\,.
\end{equation}
\end{remark}

The explicit formula \eqref{R28} can be used to show that the vector fields $X_H^{2,\red}$ \eqref{R26} span an
$n$-dimensional subspace of the tangent space at every element of a dense open subset of $\cM_{d*}^{\red,2}$,
which is required for the degenerate integrability of the reduced system. This is omitted in favour of a more powerful
argument that we are going to present in what follows.
However, the formula will be crucial for our considerations in Section \ref{S:spinRS}.

\subsection{Constants of motion and degenerate integrability after reduction}\label{ss:redDegInt}

The constants of motion of the reduced system arise from the $\UU(n)$-invariant constants of motion of the master system.
 We can display a large set of such constants of motion explicitly.
Indeed, for $(A,B, v_1,\dots, v_d)\in \cM_d$ consider the $n\times n$ matrices
\be
A^i,\quad  (BA B^{-1})^j, \quad v_\alpha v_\beta^\dagger A^k,\quad v_\gamma v_\delta^\dagger  ( BA B^{-1})^l,
\label{R42}\ee
for any meaningful, non-negative integer values of the indices.
Then, take the trace of an arbitrary repeated product of matrices chosen from this set, in any order.
 By Theorem \ref{Thm:B1},
the real and imaginary parts of such trace functions are $\UU(n)$-invariant constants of motion.

We can  also apply averaging to obtain
$\UU(n)$-invariant constants of motion. That is, take an arbitrary $f\in \Cinf(\cM_d)$
and define the averaged function $(f\circ \Psi)^{\mathrm{av}} $ by
\be
(f\circ \Psi)^{\mathrm{av}}(x) := \int_{\UU(n)} (f\circ \Psi) (\cA_g(x)) \, \text{d}_{\UU(n)}g,\qquad \forall x\in \cM_d,
\label{R43}\ee
where $\text{d}_{\UU(n)}$ denotes the probability Haar measure on $\UU(n)$, $\cA_g$ is the action of $g\in \UU(n)$ \eqref{Eq:act-Mast}, and $\Psi$ is the map
defined in \eqref{B7}.
Then, $(f\circ \Psi)^{\mathrm{av}}$ belongs to $\Cinf(\cM_d)^{\UU(n)}$ and is a constant of motion of the master system, for any $f$.

We conjecture that the so far exhibited $\UU(n)$-invariant constants of motions are sufficient for the degenerate integrability
of the reduced system (in the sense of the definition presented at the end of Section \ref{S:Integr-Md})
on the Poisson manifold $\cM_{d*}^\red$ \eqref{R2}.
For this, it remains to prove that they generate a subring of $\Cinf(\cM_{d*}^\red)$ having functional dimension $\dim(\cM_{d*}^\red) - n$, but presently
we cannot prove this. However,  using a different method,
we can exhibit the sufficient number of constants of motion
on a certain dense open subset of $\cM_{d*}^\red$.
The construction expounded below was inspired by a reasoning applied by Reshetikhin in \cite{Re1}.

\medskip

We need some preparation.
To begin, consider the $\UU(n)$-manifold
\be
M:= \UU(n)_\reg \times
\CC^{n\times d} = \{ (A, v_1,\dots, v_d) \mid A\in \UU(n)_\reg,\,v_\alpha \in \C^n\,\, (\alpha=1,\dots,d)\},
\label{R44}\ee
with $\UU(n)$ acting similarly as on  $\cM_d$, via conjugating $A$ and multiplying all vectors $v_\alpha$ by $g\in \UU(n)$.
Next, let $M_*$ be the dense open subset of $M$ on which \emph{the $\UU(n)$-action is free}. Then,
introduce the subset $\cM_{d**}\subset \cM$ by
\be
\cM_{d**}:= \{ (A, B, v_1,\dots, v_d) \in \cM_d \mid (A,v_1,\dots, v_d)\in M_*\}.
\label{R45}\ee
Since $B\in \UU(n)$ is arbitrary, $\cM_{d**}$ has the structure $M_* \times \UU(n)$. Thus, it is
a dense open subset of $\cM_d$. Consequently, it is also a dense open subset of $\cM_{d*}$.

To continue, introduce the subset  $\fC_{**}\subset \cM_{d**}$ by the following definition:
\be
\fC_{**}:= \{ (A,\tilde B, v_1,\dots, v_d)\in \cM_{d**}\mid \chi(A) - \chi(\tilde B) =0,\, \forall \chi\in \Cinf(\UU(n))^{\UU(n)}\}.
\label{R46}\ee
By using Lemma \ref{Lm:LR4}, it is easy to show
that $\fC_{**} \subset \cM_{d**}$ is an embedded (regular) submanifold of codimension $n$.
(For a proof, see  Appendix  \ref{sec:J}.)
Finally, let us  observe that
\be
\fC_{**} = \Psi(\cM_{d**}),
\label{R47}\ee
i.e.,  $\fC_{**}$  is the image of $\cM_{d**}$ under the map $\Psi$ \eqref{B7}.
Indeed, $\fC_{**}$ is precisely the subset of $\cM_{d**}$ for which the second $\UU(n)$ component is conjugate
 to the first one, denoted $A$,
and thus it can be written as $B A B^{-1}$ with some $B\in \UU(n)$.

Restriction of the $\UU(n)$-action to $\cM_{d**}$ and to $\fC_{**}$ makes these $\UU(n)$-manifolds.
Since these actions of $\UU(n)$ are free, the quotient spaces
\be
\cM_{d**}^\red:= \cM_{d**}/\UU(n)
\quad\hbox{and}\quad
\fC_{**}^\red := \fC_{**}/\UU(n)
\label{R48}\ee
become smooth manifolds.  The corresponding  $\UU(n)$ principal bundle projections
 \be
 p_1: \cM_{d**} \to \cM_{d**}^\red
 \quad\hbox{and}\quad
 p_2: \fC_{**} \to \fC_{**}^\red.
 \label{R49}\ee
 are smooth submersions, and thus are also open maps.
 Notice that $p_1$ is  the restriction of the bundle projection $\pi$ \eqref{R4}.
 Since $\cM_{d**} \subset \cM_{d*}$ is a dense open subset,  $\cM_{d**}^\red \subset \cM_{d*}^\red$
 is also a dense open subset.

 Now, let $\psi: \cM_{d**} \to \fC_{**}$ be the restriction of the map $\Psi$ \eqref{B7}.
 It is easily seen from the previous discussion that $\psi$ is an $\UU(n)$-equivariant, smooth,  surjective submersion.
As a result, it descends to a smooth, surjective submersion $\psi_{\red}: \cM_{d**}^\red \to \fC_{**}^\red$, in such a way
that we obtain the commutative Diagram \ref{Diag-DIS}.

\begin{figure}[ht]
\centering
  \captionsetup{width=.8\linewidth}
   \begin{tikzpicture}
 \node (A)  at (-1.7,1.2) {$\cM_{d \ast \ast}$};
 \node (B)  at (1.7,1.2) {$\fC_{\ast \ast}$};
 \node (C)  at (-1.7,-1.2) {$\cM_{d \ast \ast}^{\red}$};
 \node (D)  at (1.7,-1.2) {$\fC_{\ast \ast}^{\red}$};
  \path[->] (A) edge  node[above] {$\psi$}  (B); \path[->] (B) edge node[right] {$p_2$}  (D);
  \path[->] (A) edge node[left] {$p_1$}  (C); \path[->] (C) edge node[above] {$\psi_{\red}$}  (D);
 \end{tikzpicture}
 \caption{Spaces and smooth maps used to establish degenerate integrability in Theorem \ref{Thm:DegInt}.}
\label{Diag-DIS}
\end{figure}

Now we are ready to state one of the main results of the paper.

\begin{theorem} \label{Thm:DegInt}
Consider the master system of Theorem \ref{Thm:B1}, associated with the Hamiltonians in $\fH=\cE_1^* \Cinf(\UU(n))^{\UU(n)}$,
and restrict the reduced system on the dense open subset $\cM_{d**}^\red$ (defined in \eqref{R45}, \eqref{R48}).
Then, the map $\psi_\red$ of Diagram \ref{Diag-DIS} is constant along the integral curves of the
restricted reduced system. Thus, the reduced system on $\cM_{d**}^\red$
possesses the constants of motion
\be
F \circ \psi_\red,
\qquad \forall F \in \Cinf(\fC_{**}^\red).
\label{R50}\ee
The functional dimension of the so-obtained ring of constants of motion equals
\be
\dim(\fC_{**}^\red) = \dim(\cM_{d**}^\red) - n,
\label{R51}\ee
and the reduced Hamiltonian vector fields associated with  $\fH$ span an $n$-dimensional subspace
of the tangent space at every point of $\cM_{d**}^\red$.  Consequently,
the restriction of the reduced system to the Poisson manifold $\cM_{d**}^\red$ is a degenerate integrable
system in the sense of the definition given in \S\ref{ss:Def-Int}.
\end{theorem}
\begin{proof}
Consider a reduced Hamiltonian $H_\red$ on $\cM_{d**}^\red$ that arises from $H\in \fH$ \eqref{B3}. The integral curves of $H_\red$ in $\cM_{d**}^\red$
are of the form $p_1(x(t))$, where $x(t)$ is an integral curve of $H$ in $\cM_{d**}$. We have
\be
\psi_\red( p_1(x(t)) ) = p_2( \psi(x(t)) ) = p_2( \psi(x(0)) ),
\label{R52}\ee
since $\psi$ is constant along $x(t)$, by Theorem \ref{Thm:B1}. Therefore, all functions $F\circ \psi_\red$ are
constants of motion for the restricted reduced system.
The differentials $d (F\circ \psi_\red) $ span a   subspace of codimension $n$  in $T^*_y \cM_{d**}^\red $  at every $y\in \cM_{d**}^\red$,
for $\psi_\red$ is a surjective submersion and $ \dim(\fC_{**}^\red) = \dim(\cM_{d**}^\red) - n $.
Thus,  the functional dimension of the ring $\psi_\red^\ast\Cinf(\fC_{**}^\red)$ of these constants of motion equals $\dim(\cM_{d**}^\red) - n$.

Fix an arbitrary point $x:=(A,B,v_1,\ldots,v_d) \in \cM_{d**}$ and note that the kernel of $Dp_1(x)$ consists of the tangent
vectors of the form
\be
([X,A], [X,B], Xv_1,\dots, X v_d), \qquad X\in \uu(n).
\ee
The definition of $\cM_{d**}$ \eqref{R45} ensures that
\be
([X,A], Xv_1,\dots, X v_d)\neq 0 \quad \hbox{if}\quad X\neq 0.
\ee
The value of the quasi-Hamiltonian vector field $X_H$ of $H= \cE_1^* h$ is given at $x$ by
\be
(0, - B \nabla h(A), 0).
\ee
These tangent vectors span an $n$-dimensional space since $A\in \UU(n)_\reg$, also by the
definition of $\cM_{d**}$, and this  subspace of the tangent space obviously has trivial intersection
with $\mathrm{Ker}\left( Dp_1(x) \right)$.
The reduced Hamiltonian vector field of $H_\red\in \Cinf(\cM_{d**}^\red)$ is just the pushforward by $p_1$
of the quasi-Hamiltonian vector field $X_H$ on $\cM_{d**}$. Thus, we see that the reduced Hamiltonian
vector fields span an $n$-dimensional space at the arbitrary point $p_1(x)\in \cM_{d**}^\red$.
Of course, this implies that the space of the covectors $d H_\red(p_1(x))$ also has dimension $n$.
(We have  seen this in Proposition \ref{Pr:LR5}, too.)

The functional dimension of the ring of reduced Hamiltonians $H_\red$ is indeed smaller than half the dimension of the generic symplectic leaves in $\cM_{d**}^\red$,
and they generate complete flows in $\cM_{d**}^\red$,
which fits the definition of a degenerate integrable system on a Poisson manifold.
\end{proof}

Consider a symplectic leaf of codimension $n$, given by
\be
\cM_{d**}^\red(\cC) =  (\Phi^{-1}(\cC) \cap \cM_{d**})/\UU(n) \subset  \cM_{d**}^\red,
\label{leaf}\ee
 where $\cC \subset \UU(n)_\reg$ is a conjugacy class.  Let
\be
E_{\cC}: \cM_{d**}^\red(\cC)  \to \cM_{d**}^\red
\label{embed}\ee
be the canonical embedding.
Theorem \ref{Thm:DegInt} has the following important consequence.

\begin{corollary} \label{Cor:IntSympl}
Denote $\fH_\red$ and
 $\fF_\red:= \psi_\red^* \Cinf(\fC_{**}^\red)$ the rings of the commuting reduced Hamiltonians and
 their joint constants of motion on $\cM_{d**}^\red$,
and consider the embedding \eqref{embed}.
Then, the differentials of the elements of $E_{\cC}^* \fH_\red$ and of $E_{\cC}^* \fF_\red$ span
subspaces of  dimension $n$ and of codimension $n$, respectively, in  $T^*_y \cM_{d**}^\red(\cC)$
for every $y \in \cM_{d**}^\red(\cC)$.
 Consequently, the reduced system enjoys the property of degenerate integrability on every
 symplectic leaf  $\cM_{d**}^\red(\cC)$ \eqref{leaf} of codimension $n$.
 \end{corollary}

\begin{proof}
The statements  follow by combining what we proved above with
 well known results about the construction of local coordinate systems on manifolds \cite{Wa}.

We start by fixing an arbitrary point $y = p_1(x)\in \cM_{d**}^\red(\cC)$, and denote
\be
D= \dim(\fC_{**}^\red)  = \dim(\cM_{d**}^\red(\cC)) = \dim(\cM_{d**}^{\red}) - n= 2nd + n(n-1).
\ee
We can choose $n$ functions $\cF_1,\dots, \cF_n\in \fF_\red$ in such a way that
\be
\cF_i  \circ p_1 = h_i \circ \Phi,
\ee
and the differentials of the functions $h_i \in \Cinf(\UU(n))^{\UU(n)}$ are independent at $\Phi(x)$.
The functions $\cF_i$ $(i=1,\dots, n)$ are Casimir functions that are independent at $y$.
Next, we can select further $(D-n)$ elements $\cF_{n+1},\dots, \cF_D$ from $\fF_\red$ so that
the $D$ functions $\cF_k$ $(k=1,\dots, D)$ are independent at $y$.  Then, we can choose $n$ additional functions
$\cF_{D+1},\dots, \cF_{D+n} \in \Cinf(\cM_{d**}^\red)$ in such a way that altogether we obtain a local
coordinate system on $\cM_{d**}^\red$ around the point $y$.
By construction, the intersection of $\cM_{d**}^\red(\cC)$ with the open set where these coordinates are valid is
a joint level surface of the coordinate functions $\cF_1,\dots, \cF_n$.
Hence, the functions $E_{\cC}^* \cF_i$ ($i=n+1,\dots,  D+n)$ may serve as local coordinates on $\cM_{d**}^\red(\cC)$ around $y$.
In particular, the differentials of the functions
\be
E_{\cC}^* \cF_i, \qquad i=n+1,\dots, D,
\ee
are independent at $y$, and these are constants of motion for the restriction of the reduced system on $\cM_{d**}^\red(\cC)$.
Therefore, we proved the claim regarding $\fF_\red$.

The validity of the claim about $E_{\cC}^*\fH_\red$ is a consequence of the facts that
the reduced Hamiltonian vector fields on $\cM_{d**}^\red$ are tangent to $ \cM_{d**}^\red(\cC)$,
span an $n$-dimensional space at every point, and their restrictions on $ \cM_{d**}^\red(\cC)$ are the Hamiltonian
vector fields of the elements of $E_{\cC}^* \fH_\red$.
\end{proof}

\begin{remark}
Let us recall from  \eqref{lacirel} that $\Psi$ \eqref{B7} is a quasi-Poisson map.
By using this, it is possible to show that $\fC_{**}^\red$ is a Poisson manifold, and $\psi_\red$ in Diagram \ref{Diag-DIS}
is a Poisson map. It could be worthwhile to further study the Poisson structure on $\fC_{**}^\red$.
\end{remark}

We have proved the integrability of the reduced system on the dense open subset
$\cM_{d**}^\red$ of $\cM_{d*}^\red$. However, further work is needed to show integrability
on $\cM_{d*}^\red$.  We know that the reduced Hamiltonians descending from $\fH$ \eqref{B7} and their Hamiltonian vector fields
are well-defined and smooth on $\cM_{d*}^\red$, and thus have the same functional dimensions on $\cM_{d*}^\red$ as on $\cM_{d**}^\red$.
Nevertheless, there is no guarantee that the constants of motion $\psi_\red^* \Cinf(\fC_{**}^\red)$
extend to smooth functions on $\cM_{d*}^\red$.
One might be able  to overcome this difficulty by constructing such
 a functional generating set of the constants of motion over $\cM_{d**}^\red$  whose elements have suitably small supports
 so that extending them by zero on $\cM_{d*}^\red \setminus \cM_{d**}^\red$ would yield smooth functions.

\subsection{Compactness of \texorpdfstring{$\Phi^{-1}(\mC)/\UU(n)$}{Phi-1(M)/U(n)} for \texorpdfstring{$d=1$}{d=1} and non-compactness for \texorpdfstring{$d>1$}{d>1}}
\label{ss:Cmpct}

We below exhibit a qualitative difference between the reduced phase spaces for $d=1$ and $d>1$.

Let $\mC \subset \UU(n)$ be a conjugacy class and $\mu\in \mC$ an arbitrary element.
Then, the quotient space
\be
\cM_d^\red(\cC):=\Phi^{-1}(\mC)/\UU(n) \equiv \Phi^{-1}(\mu)/ \UU(n)_\mu
\label{R53}\ee
is a disjoint union of symplectic
manifolds resulting from the different $\UU(n)$ orbit types.
The inverse image refers to the moment map $\Phi$ \eqref{Eq:Mast-phi-1} on $\cM_d$.
Let us note that $\Phi^{-1}(\mu)$ is never empty. This holds since every element of $\SU(n)$ can be written
as  a group commutator \cite{Go}, and for any $\mu \in \UU(n)$ there exists $v_d\in \disk(x_d)$ such that
\be
\mu \exp(-\ic x_d v_d v_d^\dagger) \in \SU(n).
\label{Rextra}\ee

For $(A,B, v_1,\dots,v_d)\in \Phi^{-1}(\mu)$, the equality of  the determinants implied
by the  relation $\Phi(A,B,v_1,\dots, v_d)=\mu$ \eqref{Eq:Mast-phi-1} gives
\be
\exp(\ic (x_1 \vert v_1 \vert^2 + \cdots + x_d \vert v_d \vert ^2)) = \det(\mu) =: e^{\ic \Gamma}
\quad \hbox{with some}\quad  0\leq \Gamma < 2\pi.
\label{R54}\ee
This is equivalent to
\be
x_1 \vert v_1\vert^2 + \cdots + x_d \vert v_d \vert^2 = \Gamma \,\, \hbox{modulo $2\pi$}.
\label{R55}\ee

Now, suppose that $d=1$. If $\Gamma =0$, i.e., $\mu \in \mathrm{SU}(n)$, then
 we must have $v_1 =0$, and $\Phi^{-1}(\mu)$ is a closed subset of $\UU(n) \times \UU(n)$.
For $0<\Gamma  < 2\pi$,
the $d=1$ case of \eqref{R55} implies the condition
\be
\vert v_1 \vert^2 =  \Gamma/x_1 \quad \hbox{if}\quad x_1>0
\quad\hbox{or}\quad
 \vert v_1 \vert^2 = (2\pi - \Gamma)/\vert x_1\vert
\quad \hbox{if}\quad x_1<0.
\label{R56}\ee
It follows that $\Phi^{-1}(\mu)$ is a closed subset of the compact set
\be
\UU(n) \times \UU(n) \times S^{2n}_r \subset \UU(n) \times \UU(n) \times \disk(x_1) \equiv \cM_{d=1},
\label{R57}\ee
where $S^{2n-1}_r$ is the sphere in $\R^{2n} \simeq \CC^n$ of radius $0<r< \sqrt{2\pi/\vert x_1\vert}$  determined by \eqref{R56}.
Since closed subsets of compact sets are  compact, we see that $\Phi^{-1}(\mu)$ is always compact in the $d=1$ case.
The following consequence is worth emphasizing.

\begin{proposition}
In the $d=1$ case the reduced spaces $\Phi^{-1}(\mu)/\UU(n)_\mu$ are compact topological spaces.
\end{proposition}
\begin{proof}
This follows from the compactness of $\Phi^{-1}(\mu)$.
Indeed, according to Theorem 3.1 in \cite{Br}, for a compact group $G$ acting continuously on a Hausdorff space $X$,
$X$ is compact if and only if $X/G$ is compact.
 \end{proof}

An important special case is when $\mu = e^{\ic \gamma} \1_n$ with some $0\leq \gamma < 2\pi$.
Then, the moment map condition can be re-arranged as
\be
A B A^{-1} B^{-1} = \exp( \ic (\gamma \1_n - x_1 v_1 v_1^\dagger)).
\label{R58}\ee
Let us suppose for definiteness that $x_1>0$ and write $\gamma$ in the form
\be
\gamma = \frac{\beta}{n} + l \frac{2\pi}{n} \quad\hbox{with some fixed}\quad 0\leq \beta < 2\pi\quad\hbox{and}\quad  l\in\{0,1,\dots, (n-1)\}.
\label{R59}\ee
 We can transform $v_1$ into a vector whose only non-vanishing
component is the last one.
Then, since the determinant of the group commutator is $1$, we must have
\be
x_1 v_1 v_1^\dagger = \diag(0,\dots, 0, \beta).
\label{R60}\ee
In this case, we can also write
\be
\exp( \ic (\gamma \1_n - x_1 v_1 v_1^\dagger)) = \exp(\ic \diag(\gamma,\dots, \gamma, - (n-1) \gamma)).
\label{R61}\ee
Further assuming that $e^{\ic \gamma}$ is not a $k$-th root of unity for any $k=1,\dots, n$, and
working with $\SU(n)$ instead of $\UU(n)$, the structure of the `constraint surface'  \eqref{R58} in the
internally fused double $\SU(n)\times \SU(n)$ has been studied in detail in \cite{FK,FKl}.
The corresponding reduced phase space  turned out to be a smooth symplectic manifold.
For a certain range of the parameter $\gamma$, it is $\CC \mathbb{P}^{n-1}$ equipped with
a multiple of the Fubini--Study symplectic form, while for other $\gamma$-values  it has not been identified
with an already known symplectic manifold (the qualitatively different possibilities were termed
type I and type II cases in \cite{FKl}).
The resulting reduced system has been characterized
as a compactified trigonometric Ruijsenaars--Schneider system.
Using $\UU(n)$ instead of $\SU(n)$  causes a few changes in the structure of the reduced
phase space, which one might wish to study; a related investigation can be found in  \cite{Ru}.
Similar conclusions hold in the $x_1<0$ case as well.

Our next result shows  a qualitative difference between the $d=1$ and $d>1$ cases.

\begin{proposition}
For any $d\geq 2$ and $\mu \in \UU(n)$, $\Phi^{-1}(\mu)$, and thus also $\Phi^{-1}(\mu)/\UU(n)_\mu$, is non-compact.
\end{proposition}

\begin{proof}
First,  consider an arbitrary $\mu \in \SU(n)$ and $d\geq 2$.  Then, we fix $A,B\in \SU(n)$ for which
$A B A^{-1}B^{-1} = \mu$.
Next, for any parameter $0\leq \epsilon < \pi$, define the vectors $v_1(\epsilon)\in \disk(x_1)$ and $v_2(\epsilon)\in \disk(x_2)$
as follows:
\be
v_1(\epsilon)_j := \vert x_1\vert^{-\frac{1}{2}} \delta_{1,j}  \sqrt{\pi - s \epsilon},
\quad
v_2(\epsilon)_j :=   \vert x_2\vert^{-\frac{1}{2}} \delta_{1,j}  \sqrt{\pi + \epsilon}, \qquad j=1,\dots, n,
\label{R62}\ee
where  $s:= {x_1 x_2}/{\vert x_1 x_2 \vert}$.
In addition, if $d>2$, then set $v_\alpha := 0$ for $\alpha =3,\dots, d$. It is easily verified that
\be
\exp( \ic x_1  v_1(\epsilon) v_1(\epsilon)^\dagger) \exp (\ic x_2 v_2(\epsilon) v_2(\epsilon)^\dagger ) = \1_n,
\label{R63}\ee
and
\be
m(\epsilon):= (A,B,  v_1(\epsilon), v_2(\epsilon), v_3,\dots, v_d)  \in \Phi^{-1}(\mu).
\label{R64}\ee
Of course, the vectors starting with $v_3$ do not appear if $d=2$.

Now, take any sequence $\epsilon_k$ $(k\in \Z_{> 0}$) from $[0, \pi)$ for which $\lim_{k\to \infty} \epsilon_k = \pi$.
Since $\lim_{k\to \infty} v_2(\epsilon_k)$ is outside $\disk(x_2)$,
$m(\epsilon_k)$ does not have a subsequence that converges to a point of $\Phi^{-1}(\mu)$.
Consequently, $\Phi^{-1}(\mu)$ is not compact.

For $\mu \in \UU(n)\setminus \SU(n)$ and $d\geq 2$,   we take $v_2(\epsilon)$ and $v_\alpha$ for $\alpha =3,\dots, d$ precisely as above,
but pick a different vector $v_1(\epsilon) \in \disk(x_1)$ so that it obeys the condition
\be
\mu \exp(-\ic x_2 v_2(\epsilon) v_2(\epsilon)^\dagger) \exp(-\ic x_1 v_1(\epsilon) v_1(\epsilon)^\dagger)  \in \SU(n).
\ee
Then, we can find matrices $A(\epsilon), B(\epsilon) \in \SU(n)$ for which
\be
m(\epsilon):= (A(\epsilon), B(\epsilon), v_1(\epsilon),v_2(\epsilon), \dots, v_d) \in \Phi^{-1}(\mu).
\ee
Again, if $\epsilon_k \to \pi$, then  there does not exist any subsequence of $m(\epsilon_k)$  that possesses a limit in $\Phi^{-1}(\mu)$.
\end{proof}

\section{New real form of the spin RS system of Krichever and Zabrodin} \label{S:spinRS}

In this section, we assume that $d\geq 2$; the case $d=1$ goes back to \cite{FK,FKl} as discussed in \S\ref{ss:Cmpct}.
Fix $\gamma\in (0,2\pi)$ such that $k\gamma\notin 2\pi\Z$ for any $1< k \leq n$.
\begin{proposition}
 The Lie group $\UU(n)$ acts freely on $\Phi^{-1}(e^{\ic \gamma}\1_n)$ through \eqref{Eq:act-Mast},
therefore it is a subset of $\MM_{d\ast}$ as defined in Lemma \ref{Lm:LR1}.
\end{proposition}
\begin{proof}
We write $\Phi^{-1}(e^{\ic \gamma}\1_n)$ explicitly as the tuples $(A,B,v_1,\ldots,v_d)\in \MM_d$ satisfying
\begin{equation} \label{Eq:Md-gam}
 ABA^{-1}B^{-1}\exp(\ic x_1 v_1v_1^\dagger) \cdots \exp(\ic x_d v_dv_d^\dagger) =e^{\ic \gamma}\1_n\,.
\end{equation}
Meanwhile, note that
\begin{equation}
 \MM_d^\CC := \{(X,Z,V_1,W_1,\ldots , V_d,W_d) \mid X,Z\in \Gl_n(\CC),\,\, V_\alpha\in \CC^n,\,\, W_\alpha\in \Mat(1\times n,\CC)\}
\end{equation}
is equipped with the $\Gl_n(\CC)$-action
\begin{equation} \label{Eq:GlnC-act}
 g \cdot (X,Z,V_1,W_1,\ldots , V_d,W_d)  := (gXg^{-1},gZg^{-1},gV_1,W_1 g^{-1},\ldots , gV_d,W_dg^{-1})\,.
\end{equation}
Consider the complement of the zero loci $\{1+W_\alpha V_\alpha =0 \}$ with $1\leq \alpha \leq d$, in which we define  the following subset
\begin{equation} \label{Eq:MdC-gam}
\MM_{d,\gamma}^\CC := \big\{ XZX^{-1}Z^{-1} = e^{\ic \gamma} (\1_n+W_dV_d) \cdots (\1_n+W_1V_1) \big\}\,.
\end{equation}
The action of $\Gl_n(\CC)$ given by \eqref{Eq:GlnC-act} is defined on $\MM_{d,\gamma}^\CC$ \eqref{Eq:MdC-gam}.
By \cite[\S3.2]{CF2} based on \cite[\S2.6]{CF1}, $\Gl_n(\CC)$ acts freely on $\MM_{d,\gamma}^\CC$ because $q:=e^{\ic \gamma}$ is not a $k$-th root of unity ($1\leq k \leq n$).
Note that we can embed $\Phi^{-1}(e^{\ic \gamma}\1_n)$ inside $\MM_{d,\gamma}^\CC$ through
\begin{equation} \label{Eq:MdC-transf}
 X=A,\,\, Z=B,\,\, V_\alpha=v_\alpha,\,\, W_\alpha=-\ic x_\alpha \ctt(-\ic x_\alpha |v_\alpha|^2) v_\alpha^\dagger\,, \quad \text{for }\,
 \ctt(t):=\frac{e^{\ic t}-1}{\ic t}\,,
\end{equation}
where the $\UU(n)$-action \eqref{Eq:act-Mast} becomes a special instance of the $\Gl_n(\CC)$-action \eqref{Eq:GlnC-act}. Thus $\UU(n)$ must act freely on $\Phi^{-1}(e^{\ic \gamma}\1_n)$.
\end{proof}

Now we explain the following key result: the reduced dynamics of the invariant functions $\Re \tr(A)$ and $\Im \tr(A)$ described in \S\ref{ss:RedVF}  combine to yield
(after complex-linear extension) the equations of motion of the spin RS system of Krichever and Zabrodin \cite{KZ}. Hence our construction yields a real form of the spin RS system.

Introduce the following subset of $\cM_{d*}^{2,\reg,0}$ \eqref{R15}:
\begin{equation}
\label{Eq:Model-ig}
\cM_{d*}^{2,\reg,0}\cap \Phi^{-1}(e^{\ic \gamma}\1_n) \,.
\end{equation}
This set parametrizes tuples $(A,Q, v_1,\ldots,v_d)\in \MM_d$ where $Q\in \bT(n)_\reg$ and the moment map condition \eqref{Eq:Md-gam} is satisfied (with $B:=Q$).
The later condition can be rewritten as
\begin{equation}
 QAQ^{-1}=e^{-\ic \gamma} A + e^{-\ic \gamma} \sum_{\alpha=1}^d v_\alpha w_\alpha^\dagger\,,
\end{equation}
for the smooth $\CC^n$-valued function $w_\alpha=w_\alpha(A,v_\alpha,v_{\alpha+1},\ldots,v_d)$, $1\leq \alpha \leq d$, defined by
\begin{equation} \label{Eq:walpha}
 w_\alpha^\dagger
 :=\ic x_\alpha \ctt(\ic x_\alpha |v_\alpha|^2) \, v_\alpha^\dagger \exp(\ic x_{\alpha+1} v_{\alpha+1} v_{\alpha+1}^\dagger) \cdots \exp(\ic x_{d} v_{d} v_{d}^\dagger) A\,,
\end{equation}
where $\ctt(t)$ is as in \eqref{Eq:MdC-transf}. Indeed, this follows from the case $\beta=1$ of the equalities
\begin{equation}
 e^{-\ic \gamma} A+e^{-\ic \gamma} \sum_{\alpha=\beta}^d v_\alpha w_\alpha^\dagger
 = e^{-\ic \gamma} \exp(\ic x_{\beta} v_{\beta} v_{\beta}^\dagger) \cdots \exp(\ic x_{d} v_{d} v_{d}^\dagger) A\,, \qquad 1\leq \beta \leq d\,.
\end{equation}
After introducing the matrix-valued function
\begin{equation} \label{Eq:F-coll}
 F:= \sum_{\alpha=1}^d v_\alpha w_\alpha^\dagger \,,
\end{equation}
which is a sum of $d$ matrices of rank $1$, we can finally write
\begin{equation} \label{Eq:A-kl}
 A_{kl}= (e^{\ic \gamma} Q_k Q_l^{-1}- 1)^{-1}\, F_{kl}\,, \qquad 1\leq k,l\leq n\,.
\end{equation}
The form of $A$ is reminiscent of the Lax matrix for the trigonometric spin RS system obtained in \cite{FFM} and its complex analogues \cite{AF,AO,CF2,KZ}.
As $Q\in  \bT(n)_\reg$, we set $Q_k=e^{\ic q_k}$ for $q_1,\ldots,q_n$ that are distinct modulo $2\pi \Z$.

\begin{theorem} \label{Thm:Dyn-RS}
 Consider $H^\Re,H^\Im \in \fH$ defined by
\begin{equation}
 H^{\Re}= \cE_1^\ast (h^{\Re}), \,\,  h^{\Re}(A)=\Re \tr(A); \qquad
H^{\Im}= \cE_1^\ast (h^{\Im}),\,\, h^{\Im}(A)= \Im \tr(A).
\end{equation}
Then the evolution equation on $\cM_{d*}^{2,\reg,0}\cap \Phi^{-1}(e^{\ic \gamma}\1_n)$ corresponding to the vector field
$Y_{H^\Re}$ given in \eqref{R27} with $\zeta_{\fT}=-\frac12 (\nabla h^{\Re} \circ \cE_1^0)_{\fT}$ is such that
\begin{equation}
 \begin{aligned} \label{Eq:dot-Hre}
\dot{q}_j&:=-\ic Q_j^{-1} \,Y_{H^\Re}(Q_j) =-\frac{\ic}{2}(e^{\ic\gamma}-1)^{-1} (F_{jj}-\overline{F}_{jj})\,, \\
(\dot{v}_{\alpha})_j&:=Y_{H^\Re}((v_{\alpha})_j)= \frac14 \sum_{k\neq j} \left[1-\ic \cot\left(\frac{q_j-q_k}{2} \right) \right] (A_{jk}-\bar{A}_{kj}) (v_\alpha)_k\,,
 \end{aligned}
\end{equation}
and similarly for $Y_{H^\Im}$ with $\zeta_{\fT}=-\frac12 (\nabla h^{\Im} \circ \cE_1^0)_{\fT}$ the evolution equation satisfies
\begin{equation}
 \begin{aligned} \label{Eq:dot-Him}
\dot{q}_j&:=-\ic Q_j^{-1} \, Y_{H^\Im}(Q_j)=-\frac{1}{2}(e^{\ic\gamma}-1)^{-1} (F_{jj}+\overline{F}_{jj})\,, \\
(\dot{v}_{\alpha})_j&:=Y_{H^\Im}((v_{\alpha})_j)=-\frac{\ic}{4} \sum_{k\neq j} \left[1-\ic \cot\left(\frac{q_j-q_k}{2} \right) \right] (A_{jk}+\bar{A}_{kj}) (v_\alpha)_k\,;
 \end{aligned}
\end{equation}
In both cases, $\dot{w}_{\alpha}$ is obtained from $\dot{v}_{\alpha}$ after replacing the vector $v_\alpha$ by $w_\alpha$.
In particular, the complex-linear combination of the above evolution equations corresponding to the ``formal vector field'' $2(e^{\ic\gamma}-1) (Y_{H^\Re}+\ic Y_{H^\Im})$
(associated with $\cE_1^\ast (h)$ for $h(A)=2(e^{\ic\gamma}-1) \tr(A)$ and $\zeta_{\fT}=-(e^{\ic\gamma}-1) [ (\nabla h^{\Re} \circ \cE_1^0)_{\fT}+\ic  (\nabla h^{\Im} \circ \cE_1^0)_{\fT}]$)
yields
\begin{equation}
 \begin{aligned} \label{Eq:dot-RS}
\frac12\dot{q}_j&=-\ic F_{jj}\,, \\
(\dot{v}_{\alpha})_j&=-\ic \sum_{k\neq j} (v_\alpha)_k\, F_{jk}  \operatorname{V}\left(\frac{q_j-q_k}{2}\right)\,,\\
(\dot{\bar{w}}_{\alpha})_j&=\ic \sum_{k\neq j} (\bar{w}_\alpha)_k\, F_{kj}  \operatorname{V}\left(\frac{q_k-q_j}{2}\right)\,,
 \end{aligned}
\end{equation}
where we introduced the `potential function'
\begin{equation}
\operatorname{V}(q):=\cot\left(q \right) - \cot\left(q +\frac{\gamma}{2}\right)\,.
\end{equation}
\end{theorem}

\begin{proof}
As a consequence of Proposition \ref{Pr:LR6}, the restriction $\Phi:\cM_{d*}^{2,\reg,0}\to \UU(n)$  of the moment map  is conjugated by a matrix $C(t)$ (see Remark \ref{Rem:Flow}) along the integral curve of $Y_{H^\Re}$ and $Y_{H^\Im}$. Thus the evolution of these vector fields can be restricted to $\cM_{d*}^{2,\reg,0}\cap \Phi^{-1}(e^{\ic \gamma}\1_n)$.

Next, note from our choice of pairing \eqref{Eq:ipU} that \eqref{B4} yields
$\nabla h^{\Re}(A)=-\frac12(A-A^\dagger)$ and $\nabla h^{\Im}(A)=\frac{\ic}{2}(A+A^\dagger)$.
By \eqref{R25}, we can write for any $\xi\in \uu(n)$ and $1\leq j,k\leq n$,
$(\cR(Q)(\xi))_{jk}=-\frac{\ic}{2} (1-\delta_{jk}) \cot\left(\frac{q_j-q_k}{2}\right) \xi_{jk}$.
Hence the expressions for $(\dot{v}_{\alpha})_j$ in \eqref{Eq:dot-Hre} and \eqref{Eq:dot-Him} directly follow\footnote{Our choice of $\zeta_{\fT}$ guarantees that the matrix multiplying $v_\alpha$ in $Y_{H^{\Re/\Im}}(v_\alpha)$
has zero diagonal part.} from \eqref{R27}, while for $\dot{q}_j$ we further use the form of $A$ given by \eqref{Eq:A-kl}.
To write down the evolution equation of $w_{\alpha}$ \eqref{Eq:walpha}, we note thanks to Remark \ref{Rem:Flow} that the evolution of $w_\alpha$ and $v_\alpha$ is the same (it corresponds to left multiplication by $C(t)$), hence they satisfy the same evolution equations.

For the second part, we multiply \eqref{Eq:dot-Hre} by $2(e^{\ic\gamma}-1)$ and \eqref{Eq:dot-Him} by $2\ic (e^{\ic\gamma}-1)$, then add the results together. We directly get $\dot{q}_j$ in \eqref{Eq:dot-RS}, while for $(\dot{v}_{\alpha})_j$ we need to remark that
$\left[1-\ic \cot\left(\frac{q_j-q_k}{2} \right) \right] A_{jk} = -\ic (e^{\ic\gamma}-1)^{-1} \operatorname{V}(q_j-q_k) F_{jk}$.
To get $(\dot{\bar{w}}_{\alpha})_j$, we first need the complex conjugation of the evolution under the real vector fields $Y_{H^{\Re}}$ and $Y_{H^{\Im}}$, and after that we take the complex-linear combination of the resulting equations.
\end{proof}

Up to a multiplicative constant, the formal evolution equation \eqref{Eq:dot-RS} reproduces the real form of the trigonometric spin RS system obtained by Marshall and the 2 authors, see Corollary 4.3 in \cite{FFM}.
However, there are important discrepancies between these real models.
First, it is only after taking a complex-linear combination of \eqref{Eq:dot-Hre} and \eqref{Eq:dot-Him} that we can recognize the spin RS evolution equations since the
function $(A,B,v_\alpha)\mapsto \tr(A)$ is not real-valued on $\MM_d$.
Second, our matrix of collective spins $F$ \eqref{Eq:F-coll} involves two types of vectors: the original $v_\alpha$ that enter the definition of the space
$\MM_d$, and the new $w_\alpha$ \eqref{Eq:walpha}. In the work \cite{FFM}, the analogue of the matrix $F$ has the symmetric
form $\tilde{F}=\sum_{\alpha=1}^d \tilde{v}_\alpha \tilde{v}_\alpha^\dagger$ for some vectors $\tilde{v}_\alpha$. Nevertheless, the
lack of symmetry between the two type of spin vectors $v_\alpha,w_\alpha$ that we used is closer to the original complex model of Krichever and Zabrodin \cite{KZ}.
In particular, we can derive from \eqref{Eq:F-coll} and \eqref{Eq:dot-RS} that
\begin{equation}
 \frac12\ddot{q}_j=- \sum_{k\neq j} F_{jk} F_{kj}  \left[ \operatorname{V}\left(\frac{q_j-q_k}{2}\right) - \operatorname{V}\left(\frac{q_k-q_j}{2}\right) \right]\,,
\end{equation}
so that the evolution of the variables $(x_j:=\frac12 q_j,a_j^\alpha:=(v_\alpha)_j,b_j^\alpha:=(\bar{w}_\alpha)_j)$
can be readily identified with the equations of motion (1.21)-(1.23) therein (up to multiplying the time parameter by $-\ic $).
This completes our claim that we have unveiled a real form of the spin RS system of Krichever and Zabrodin \cite{KZ}.

\section{Summary and open problems}
\label{S:ccl}

Here, we summarize the most essential points of our study and draw attention to open problems.

We first introduced a master dynamical system on the extended double $\cM_d$ \eqref{I4} that we equipped with a pencil of
quasi-Poisson brackets admitting compatible quasi-Hamiltonian structures, as described in Section \ref{S:Master}.
The commuting flows of the master system and their joint constants of motion
are characterized by Theorems \ref{Thm:B1} and \ref{Thm:B5}, which show that this system is a natural generalization of the
degenerate integrable system of free motion carried by the extended cotangent bundle $M_d$ \eqref{I1}.
We then investigated the  reduction of the master system obtained by descending
to the quotient space $\cM_d^\red = \cM_d/\UU(n)$.
 We considered the chain of dense open subsets
\be
\cM_{d**}^\red \subset \cM_{d*}^\red \subset  \cM_d^\red,
\ee
where $\cM_{d*}^\red$ is the `big cell' associated with the principal orbit type and
$\cM_{d**}^\red$ is defined in \eqref{R45}, and
proved that the reduced system yields a multi-Hamiltonian, degenerate
integrable system after restriction on  $\cM_{d**}^\red$.
This is the content of Theorem \ref{Thm:DegInt},
and we  demonstrated in Corollary \ref{Cor:IntSympl} that integrability is inherited on any maximal
symplectic leaf $\cM_{d**}^\red(\mathrm{C})$ \eqref{leaf}  of codimension $n$, where $\mathrm{C}\subset \UU(n)_{\reg}$ is a conjugacy class.
These are new results even for $d=1$,  in which case the pencil collapses
to the main quasi-Poisson structure of Theorem \ref{Thm:qH-Master}.

The track of thought followed in our proof of degenerate integrability  may  be adopted
to fill some gaps that were left open in previous investigations of reductions of other master systems, see, e.g.,  \cite{Fe} and references therein.
It is also worth noting that a slight modification of our arguments shows the integrability of the reduced systems
in the real analytic category as well.

One of our interesting results is  the identification  of special symplectic leaves in $\cM_{d*}^\red$ $(d\geq 2)$ that support a new real form
of the complex, trigonometric spin RS models of Krichever and Zabrodin.  These models were described at the level
of the reduced equations of motion (Theorem \ref{Thm:Dyn-RS}), leaving
the proof of their integrability for  future work.
Besides this, several other questions remain open, and we finish with a couple of remarks on those and other issues:

\begin{enumerate}

\item We conjecture that degenerate integrability  holds on the Poisson manifold $\cM_{d*}^\red$, too.

\item It would be important to explore the integrability properties of the reduced system on arbitrary
symplectic leaves of the full quotient space $\cM_d^\red$, in particular to establish integrability of the real form of the spin RS model considered in Section \ref{S:spinRS}.

\item The quantization of the integrable systems living on (symplectic leaves of) $\cM_d^\red$ is a wide open question.

\item There are $\UU(n)$-invariant subsets $\operatorname{A}_{k}(x) \subset \CC^n$ ($k\geq 1$) which are analogues of the quasi-Poisson ball $\disk(x)$, see Remark \ref{Rem:qP-disk}. To highlight the similarity with the extended cotangent bundle \eqref{I1}, we chose to define the master phase space $\MM_d$ \eqref{I4} using $d$ quasi-Poisson balls as we can recover $\CC^n\simeq T^\ast \R^n$ from $\disk(x)$ in the limit $x\to 0$. (The quasi-Poisson bracket \eqref{Eq:qPB-C2} must be rescaled while taking the limit.) However, one can replace any or all of the balls $\disk(x_\alpha)$ by
 arbitrary regions  $A_{k_\alpha}(x_\alpha) \subset \CC^n$  ($\forall k_\alpha\in \Z_{>0}$ for $1\leq \alpha\leq d$).
The so-obtained systems and their reductions
enjoy similar properties as is the case for $\cM_d$, but they could have more complicated topological structure.

\item In analogy with the pencil of quasi-Poisson structures on $\MM_d$ used throughout the text,
we explained in Remark \ref{Rem:Cotang} that we can  also define
compatible Poisson structures  on the extended cotangent bundle $M_d$  \eqref{I1}, for any $d\geq 2$.
This raises the issue of finding analogous structures in the related complex holomorphic models \cite{AO,CF2}, and in the real forms based
on Poisson--Lie groups \cite{FFM}. This problem should be particularly fruitful to investigate since it could lead to proving Conjecture \ref{Conj:PenC} from Appendix \ref{sec:M}.

\item  It could be interesting to study the structure of degenerate integrable systems on regular quasi-Poisson manifolds in general.
Other examples fitting the definition formulated in \S\ref{ss:Def-Int} could arise from
the quasi-Poisson description of moduli spaces of flat connections \cite{AKSM,ARe,LB}.

\end{enumerate}

We hope to be able to return to at least some of the above issues in the future.

\medskip

\subsubsection*{Acknowledgements}

We are grateful to B.G.~Pusztai and L.~Stach\'o for help on a technical  point.
M.F. was partly supported by a Doctoral Prize Fellowship of the University of Loughborough.
This project has received funding from the European Union’s Horizon 2020 research and innovation programme under
the Marie Sk\l{}odowska-Curie grant agreement No. 101034255.
The work of L.F. was supported in part by the NKFIH research grant K134946.

\medskip

 \appendix
\section{An alternative approach to the quasi-Poisson ball}
\label{sec:A}

We outline our original approach to discover the quasi-Poisson ball $\disk(1)$ given in Proposition \ref{Pr:qP-Cn} which plays a prominent role in this work.

\medskip

Fix integers $n_0,n_1\geq 1$. In the work of Van den Bergh \cite{VdB1}, it is proved that the complex manifold
$\Mat(n_0\times n_1,\CC)\times \Mat(n_1\times n_0,\CC)$
on which $\Gl_{n_0}(\CC)\times \Gl_{n_1}(\CC)$ acts by $(g_0,g_1)\cdot (V,W)=(g_0 V g_1^{-1},g_1 W g_0^{-1})$
is equipped with a holomorphic Hamiltonian quasi-Poisson structure on a dense open submanifold.
For $n_1=1$ and $n:=n_0$, the quasi-Poisson bracket reads
\begin{equation*}
 \br{V_{i},V_{k}}=0\,, \quad \br{W_{i},W_{k}}=0\,, \quad
\br{V_{i},W_{k}}=\frac12 (2+ WV) \delta_{ik} + \frac12  V_i W_{k}\,,
\end{equation*}
and the moment map is $\Phi(V,W)=(\1_{n}+VW, \det(\1_n + VW)^{-1})$ on the dense open submanifold $\{(V,W)\mid  \det(\1_n + VW) = 1+WV\neq 0\}$.
As a real model, we want to consider (a suitable submanifold of) $\CC^n$
with the natural action of $\UU(n)\times \UU(1)$ given by $(g,\lambda)\cdot v:=gv \lambda^{-1}$.
We make the following ansatz
\begin{equation}
 \br{v_i,v_k}=0\,, \quad \br{\bar{v}_i,\bar{v}_k}=0\,, \quad
\br{v_i,\bar{v}_k}=\frac{\ic}{2}\att(|v|^2)\, \delta_{ik} - \frac{\ic}{2} \btt(|v|^2) \, v_i \bar{v}_k\,,
\label{Eq:ansatz}
\end{equation}
where $\att(t),\btt(t)$ are real-analytic functions in a  neighborhood of the origin $0\in \R$,
and $(v_i,\bar{v}_i)_{i=1}^n$ denote the complex-valued coordinate functions on $\CC^n$.
Of course, anti-symmetry of the bracket is understood.
Reality of the bracket follows from the identity $\br{\bar{f_1},\bar{f_2}}=\overline{\br{f_1,f_2}}$ for any
 $f_1,f_2\in \Cinf(\CC^n,\CC)$. This holds for the evaluation functions $(v_i,\bar{v}_i)$, hence for all functions.
Below,  we start with a technical lemma that will be utilized subsequently.

\begin{lemma}\label{lem:Laci}
Let us extend the inner product of $\uu(n)$ to a complex bilinear form on $\gl_n(\CC)$ and take an
arbitrary pair of dual bases of $\gl_n(\CC)$, denoted $\{E_\alpha\}$ and $\{E^\beta\}$, that satisfy
\be
\langle E_\alpha, E^\beta \rangle := - \tr(E_\alpha E^\beta) = \delta_{\alpha\beta}.
\ee
Then the action of the Cartan trivector $\phi$ \eqref{Eq:Cartan3}  of $\uu(n)$ on three evaluation function on $\CC^n$
can be expressed as follows:
\begin{equation}
 \begin{aligned}
\phi_{\C^n}(v_i,v_k,v_l) &=
- \frac{1}{2} \sum_{\alpha,\beta,\gamma} \langle E^\alpha, [E^\beta, E^\gamma] \rangle
 (E_\alpha v)_i (E_\beta v)_k  (E_\gamma v)_l, \\
\phi_{\C^n}(v_i,v_k,{\bar v}_l) &=
  \frac{1}{2} \sum_{\alpha,\beta,\gamma} \langle E^{\alpha}, [E^\beta, E^\gamma] \rangle
  (E_\alpha v)_i (E_\beta v)_k
 ( v^\dagger E_\gamma)_l,
 \end{aligned}
\end{equation}
 where all the summation indices run over $n^2= \dim_\CC(\gl_n(\CC))$ values.
 \end{lemma}
\begin{proof}
By definition of the infinitesimal action of $\uu(n)$ through \eqref{EqinfVectM}, we have for an  arbitrary
$\xi \in \uu(n)$ that
\begin{equation}
 \xi_{\C^n}(v_i) = - (\xi v)_i,
\qquad
\xi_{\C^n}({\bar{v}}_i) = (v^\dagger \xi)_i.
\end{equation}
Using the definition of  $\phi$ \eqref{Eq:Cartan3}, we can write for an orthonormal basis $(e_a)_{a\in \mathtt{A}}$ of $\uu(n)$ that
\begin{equation*}
 \begin{aligned}
 \phi_{\C^n}(v_i,v_k,v_l) &= - \frac{1}{2} \sum_{a,b,c\in \mathtt{A}} \langle e_{a}, [e_b, e_c] \rangle \, (e_av)_i  (e_b v)_k  (e_c v)_l,\\
 \phi_{\C^n}(v_i,v_k,{\bar v}_l) &=
 \frac{1}{2} \sum_{a,b,c\in \mathtt{A}} \langle e_{a}, [e_b, e_c] \rangle \,
 (e_av)_i  (e_b v)_k ( v^\dagger e_c)_l.
 \end{aligned}
\end{equation*}
Noting that \eqref{Eq:shortEE} yields $\sum_{a\in \mathtt{A}} \langle e_a, E_\alpha\rangle \langle e_a, E^\beta\rangle=\delta_{\alpha\beta}$,
substitution of the expansions
 \be
 e_a = \sum_\alpha \langle e_a, E_\alpha\rangle E^\alpha\,, \quad  e_a = \sum_\alpha \langle e_a, E^\alpha\rangle E_\alpha\,,
 \ee
 in the previous equations leads immediately to the claim.
\end{proof}

Below, we fix $r\in \R_{>0}\cup \{\infty\}$ and impose the condition that the ansatz \eqref{Eq:ansatz} should define
a quasi-Poisson bracket with regard to the natural action of $\UU(n)$ on $\CC^n$.

\begin{proposition}  \label{Pr:Ans}
Suppose that $\att(t),\btt(t)$ are real-analytic functions on $\{t\in \R \mid |t|<r\}$.
The bracket \eqref{Eq:ansatz} defines a real quasi-Poisson bracket on $\{v\in \CC^n \mid |v|^2 < r\}$ for the $\UU(n)$-action
 $g\cdot v:=gv$ if and only if the functions
 $\att(t),\btt(t)$ satisfy
\begin{equation} \label{Eq:Ans1}
 \att(t)\btt(t) + \att'(t) \big[ \att(t) - t \btt(t)  \big] = - t\,.
\end{equation}
\end{proposition}
\begin{proof}
It is sufficient to verify the quasi-Poisson property \eqref{Eq:JacPhi} on the evaluation functions $(v_i,\bar{v}_i)$,
with the Cartan trivector of $\uu(n)$ on the right-hand side. Using the reality of the bracket \eqref{Eq:ansatz}, we only
 need to check the two identities
\begin{equation} \label{Eq:pf-qP1}
\begin{aligned}
 \br{v_i,\br{v_k,v_l}} +\br{v_k,\br{v_l,v_i}} + \br{v_l,\br{v_i,v_k}} &= \frac12 \phi_{\CC^n}(v_i,v_k,v_l) \,,  \\
 \br{v_i,\br{v_k,\bar{v}_l}} +\br{v_k,\br{\bar{v}_l,v_i}} +
 \br{\bar{v}_l,\br{v_i,v_k}} &= \frac12 \phi_{\CC^n}(v_i,v_k,\bar{v}_l)\,,
 \end{aligned}
\end{equation}
with $\phi$ the Cartan trivector of $\uu(n)$ and $1\leq i,k,l \leq n$.
To compute the right-hand sides, we make use of Lemma \ref{lem:Laci} and the obvious basis
$(E_{ab})_{1\leq a,b\leq n}$ of $\gl_n(\CC)$. Then, for $E_\alpha = E_{ab}$ we have $E^\alpha = - E_{ba}$.
Now the only non-zero values of $\langle E^\alpha, [E^\beta, E^\gamma]\rangle$ come from
\be
 \langle E_{aa}, [E_{cd}, E_{dc}] \rangle \quad \hbox{for}\quad c\neq d,\,\,a=c\,\,\hbox{or}\,\, a= d,
 \label{E2}\ee
 and from
 \be
 \langle E_{ab}, [E_{bc}, E_{ca}] \rangle \quad \hbox{for $a,b,c$ distinct},
 \label{E3}\ee
 and the permutations of the 3 basis vectors.
 For convenient bookkeeping, we below write
 \begin{equation}
  \begin{aligned}
 \phi_{\CC^n}(v_i,v_k,v_l) &=  \phi_{\CC^n}^{(2)}(v_i,v_k,v_l) +  \phi_{\CC^n}^{(3)}(v_i,v_k,v_l)\,, \\
  \phi_{\CC^n}(v_i,v_k,\bar v_l) &= \phi_{\CC^n}^{(2)}(v_i,v_k,\bar v_l) +  \phi_{\CC^n}^{(3)}(v_i,v_k,\bar v_l)\,,
  \end{aligned}
 \end{equation}
where $\phi_{\CC^n}^{(2)}$ contains the non-zero terms coming from \eqref{E2} and  $\phi_{\CC^n}^{(3)}$ contains
the contributions associated with the non-zero values in \eqref{E3}.
With these preparations in hand, we proceed to the inspection of the required identities.

\medskip

\noindent $\bullet$ \textit{First identity in \eqref{Eq:pf-qP1}.}
Since $\br{v_i,v_j}=0$ for any $1\leq i,j\leq n$, the left-hand side is zero.
Next, we compute using \eqref{E2} that
\begin{equation*}
 \begin{aligned}
  {  } {\phi}^{(2)}_{\CC^n}(v_i,v_k,v_l)=&
 \frac12 \delta_{ik}\delta_{(l\neq i)}v_i v_k v_l
 +\frac12 \delta_{kl}\delta_{(i\neq k)}v_i v_k v_l
 +\frac12 \delta_{il}\delta_{(k\neq l)}v_i v_k v_l \\
 &
 -\frac12 \delta_{ik}\delta_{(l\neq i)}v_i v_k v_l
 -\frac12 \delta_{kl}\delta_{(i\neq k)}v_i v_k v_l
 -\frac12 \delta_{il}\delta_{(k\neq l)}v_i v_k v_l\,=0\,.
 \end{aligned}
\end{equation*}
(Here and below, we use the notation $\delta_{(i\neq j)}:=1-\delta_{ij}$ for $1\leq i,j\leq n$.)
Similarly, we see that ${\phi}^{(3)}_{\CC^n}(v_i,v_k,v_l)=0$, and the right-hand side vanishes.

\noindent $\bullet$ \textit{Second identity in \eqref{Eq:pf-qP1}.}
To calculate the left-hand side, we first note that
\begin{equation} \label{Eq:pf-qP5}
  \br{v_i,|v|^2}=\sum_{1\leq l\leq n} v_l \br{v_i,\bar{v}_l}
 =\frac{\ic v_i}{2} \left[\att(|v^2|)- |v|^2 \btt(|v|^2) \right]\,.
\end{equation}
Therefore, we have
\begin{equation*}
 \begin{aligned}
   \br{v_i,\br{v_k,\bar{v}_l}}=&
\frac14 \att(|v|^2) \btt(|v|^2) v_k \delta_{il}
-\frac14 \btt(|v|^2)^2 v_i v_k \bar{v}_l \\
&-\frac14 \left[\att'(|v^2|) \delta_{kl} - \btt'(|v|^2) v_k \bar{v}_l \right] \, v_i\,
\left[\att(|v^2|) - |v|^2 \btt(|v|^2) \right] \,,
 \end{aligned}
\end{equation*}
from which we obtain
\begin{equation*}
 \begin{aligned}
& \br{v_i,\br{v_k,\bar{v}_l}} +\br{v_k,\br{\bar{v}_l,v_i}} + \br{\bar{v}_l,\br{v_i,v_k}} \\
 =& \frac14
\left(\att(|v|^2)\btt(|v|^2) + \att'(|v|^2) \big[ \att(|v^2|) - |v|^2 \btt(|v|^2)  \big] \right)
\,(v_k \delta_{il} - v_i \delta_{kl})\,.
 \end{aligned}
\end{equation*}
For the right-hand side, we inspect the terms coming from \eqref{E2} and \eqref{E3}, and find
\begin{equation*}
 \begin{aligned}
{\phi}^{(2)}_{\CC^n}(v_i,v_k,\bar{v}_l)=&
-\frac12 \delta_{il} \delta_{(i\neq k)}v_k (|v_k|^2+|v_i|^2)
+\frac12 \delta_{kl} \delta_{(i\neq k)}v_i (|v_i|^2+|v_k|^2)\,,\\
{\phi}^{(3)}_{\CC^n}(v_i,v_k,\bar{v}_l)=&
-\frac12 \delta_{il} \delta_{(i\neq k)}v_k (|v|^2-|v_k|^2-|v_i|^2)
+\frac12 \delta_{kl} \delta_{(i\neq k)}v_i (|v|^2-|v_i|^2-|v_k|^2)\,.
 \end{aligned}
\end{equation*}
Therefore,
\begin{equation}
 \frac12 \phi_{\CC^n}(v_i,v_k,\bar{v}_l)=
-\frac14 |v|^2 (\delta_{il} \delta_{(i\neq k)}v_k - \delta_{kl} \delta_{(i\neq k)}v_i)
=
-\frac14 |v|^2 (v_k  \delta_{il} - v_i \delta_{kl})\,.
\end{equation}
Then the two sides of the equality coincide if and only if \eqref{Eq:Ans1} holds.
\end{proof}

\begin{remark}
We are interested in real-analytic structure functions $\att$ and $\btt$, but
the condition given by \eqref{Eq:Ans1} is valid for $\Cinf$ structure functions as well.
\end{remark}

To study the existence of a moment map for the quasi-Poisson structures given in Proposition \ref{Pr:Ans}, we need the following result.
\begin{lemma} \label{Lem:AnsMomap}
 A smooth function $\Phi:\{v\in \CC^n \mid |v|^2 < r\}\to \UU(n)$ is a moment map for the action given by $g \cdot v:=gv $ if and only if
 \begin{equation} \label{Eq:Lem-Momap}
 \br{\Phi_{ij},v_{k}}=\frac12 \delta_{kj} (\Phi v)_{i} + \frac12 \Phi_{kj} v_{i}\,.
\end{equation}
\end{lemma}
\begin{proof}
The moment map property given in \eqref{Eq:momap} is satisfied if and only if \eqref{Eq:momap} holds
 with $F=g_{ij}$ (where $g_{ij}\in \Cinf(\UU(n),\CC)$ is the coordinate function $g_{ij}(\eta)=\eta_{ij}$
 for $\eta\in \UU(n)$) when applied to the coordinate functions $v_k$ and $\bar{v}_k$.
Since we consider a real quasi-Poisson bracket, \eqref{Eq:Lem-Momap} is equivalent to
 \begin{equation*}
 \br{\Phi_{ij},\bar{v}_{k}}=-\frac12 (v^\dagger \Phi)_{j} \delta_{ik} - \frac12 \bar{v}_{j} \Phi_{ik}\,,
\end{equation*}
so it suffices to check that \eqref{Eq:momap} with $F=g_{ij}$  evaluated on $v_k$ verifies the first identity
 in \eqref{Eq:Lem-Momap}.

The action induces for any $\xi\in \uu(n)$ that $\xi_{\CC^n}(v_{k})=-(\xi v)_{k}$. Moreover, on the coordinate function $g_{ij}$ of $\UU(n)$, the left- and right-invariant vector fields give $\xi^L(g_{ij})=(g \xi)_{ij}$, $\xi^R(g_{ij})=(\xi g)_{ij}$. For any $1\leq i,j,k\leq n$, the condition \eqref{Eq:momap} with $F=g_{ij}$ applied to $v_{k}$ yields
\begin{equation*}
 \br{\Phi_{ij},v_{k}}=- \frac12 \sum_{a\in \mathtt{A}} (\Phi  e_a + e_a \Phi)_{ij} \,(e_a v)_{k}\,,
\end{equation*}
for any orthonormal basis $(e_a)_{a\in \mathtt{A}}$ of $\uu(n)$.
Using \eqref{Eq:shortEE}, we recover  \eqref{Eq:Lem-Momap}.
\end{proof}

\noindent Now we search for a Hamiltonian quasi-Poisson structure on $\CC^n$ with a prescribed moment map.

\begin{proposition} \label{Pr:Ans2}
On $\{v\in \CC^n \mid |v|^2 < r\}$,
consider a quasi-Poisson bracket of the form \eqref{Eq:ansatz} with two functions $\att(t),\btt(t)$ which are real-analytic on $\{t\in \R \mid |t|<r\}$.
The smooth function $\Phi:\{v\in \CC^n \mid |v|^2 < r\}\to \UU(n)$ given by
$\Phi(v)=\exp(\ic vv^\dagger)$  defines a moment map if and only if we have
\begin{equation} \label{Eq:abc-rel}
 \ctt(t) \att(t)= 2 + \ic t \ctt(t)\,, \qquad
\ctt'(t) (\att(t)-t \btt(t))-\ctt(t) \btt(t)=\ic \ctt(t)\,,
\end{equation}
where $\ctt(t)$ is the complex-valued analytic function $\ctt:\R\to\CC$, $\ctt(t):=\frac{e^{\ic t}-1}{\ic t}$.
This leads uniquely to the functions
 \begin{equation} \label{Eq:ab-fct-BIS}
\att(t)=t\,\cot\left(\frac{t}{2}\right)\,, \quad \btt(t)= \cot\left(\frac{t}{2}\right)-\frac{2}{t}\,,
\end{equation}
which are real-analytic for $\vert t\vert < 2\pi$ and satisfy the relations \eqref{Eq:Ans1}. Then the ansatz \eqref{Eq:ansatz} reproduces the quasi-Poisson structure of
Proposition \ref{Pr:qP-Cn}  with $x=1$, that results by exponentiation of the standard
Poisson structure on $\CC^n$.
\end{proposition}

\begin{proof}
By Lemma \ref{Lem:AnsMomap}, we have to show that the equality in \eqref{Eq:Lem-Momap} holds if and only if \eqref{Eq:abc-rel} is satisfied. Noting that $\Phi=\exp(\ic vv^\dagger)=\1_n+\ic \ctt(|v|^2) vv^\dagger$, we can compute using \eqref{Eq:ansatz} that
\begin{equation}
 \begin{aligned} \label{Eq:pf-Ans1}
&\br{\Phi_{ij},v_{k}}=\ic \,\br{\ctt(|v|^2) v_i \bar{v}_j,v_{k}} \\
 =& \frac12 \ctt(|v|^2) \att(|v|^2)\, v_i \delta_{kj}
 +\frac{1}{2} \left[ \ctt'(|v|^2)(\att( |v|^2) -|v|^2 \btt(|v|^2)  )   - \ctt(|v|^2)\btt(|v|^2) \right] v_i \bar{v}_j v_k\,.
 \end{aligned}
\end{equation}
Plugging the expression of $\Phi$ in terms of $\ctt(|v|^2)$ into the moment map equation \eqref{Eq:Lem-Momap} gives
\begin{equation}  \label{Eq:pf-Ans2}
\br{\Phi_{ij},v_{k}}=\frac12 \left[2+ \ic |v|^2 \ctt(|v|^2)\right]\, v_i \delta_{kj}
 +\frac{\ic}{2} \ctt(|v|^2) v_i \bar{v}_j v_k\,.
\end{equation}
Since \eqref{Eq:pf-Ans1} and \eqref{Eq:pf-Ans2} coincide if and only if  \eqref{Eq:abc-rel} holds,
we have derived the relations  \eqref{Eq:abc-rel}.
The rest of the statements of the proposition is then readily checked.
\end{proof}

\begin{remark} \label{Rem:App-U1}
The brackets of Proposition \ref{Pr:Ans} can be viewed also as quasi-Poisson brackets
associated with the group $\UU(n)\times \UU(1)$ acting on $\CC^n$ according to
 $(g,\lambda)\cdot v:=gv \lambda^{-1}$. This holds  because the Cartan trivector of $\UU(1)$ vanishes.
With the functions \eqref{Eq:ab-fct-BIS}, it turns out that $(\Phi,\widehat{\Phi}):
\disk(1) \to \UU(n)\times \UU(1)$ given by $\Phi(v) = \exp(\ic v v^\dagger)$ and
\begin{equation}
\widehat{\Phi}(v) = (\det \Phi(v))^{-1} = \exp(-\ic |v|^2)
\end{equation}
 is a moment map for this action.
Indeed, one finds that the $\UU(1)$ moment map condition is
$\{ \widehat{\Phi}, v_k\} = - v_k \widehat{\Phi}$.
With the ansatz \eqref{Eq:ansatz}, this requires the relation
$\att(t)-t\btt(t)=2$,
which is satisfied by the functions in \eqref{Eq:ab-fct-BIS}.
\end{remark}

\begin{remark}
Fix an arbitrary polynomial function $\att:\R\to \R$ with $\att(0)\neq 0$.
Set $\btt(t)=-(t+\att'(t)\att(t))/(\att(t)-t\att'(t))$, which is real-analytic in an open neighborhood $U\subset \R$ of $0$.
The pair $(\att(t),\btt(t))$ defines a $\UU(n)$ quasi-Poisson bracket through \eqref{Eq:ansatz} on $\{v\in \CC^n \mid |v|^2 < r\}$ whenever $(-r,r)\subset U$ because \eqref{Eq:Ans1} is satisfied.
Taking any constant $a\in \R\setminus\{0\}$, the easiest cases are
\begin{equation*}
 \begin{aligned}
\att(t)=&a\,, \quad \btt(t)=-\frac{t}{a}\,, \quad U=\R\,, \\
\att(t)=&1+at\,, \quad \btt(t)=-a-(1+a^2)t\,, \quad U=\R\,, \\
\att(t)=&1+t+at^2\,, \quad \btt(t)=-\frac{1+2(1+a)t+3at^2+2a^2t^3}{1-at^2}\,,\quad
U=\left\{
\begin{array}{ll}
\R&a<0, \\
(-a^{-1/2},a^{-1/2})&a>0.
\end{array}
 \right.
 \end{aligned}
\end{equation*}
A moment map is not known for these examples, and it might not exist.
\end{remark}

\begin{remark}[Relation to Van den Bergh's example] \label{Rem:VdB}
We began this appendix by highlighting a holomorphic  Hamiltonian quasi-Poisson structure
 due to Van den Bergh \cite{VdB1} on
$$\{(V,W)\mid 1+WV\neq 0\} \subset T^\ast \CC^n = \{(V,W)\mid V\in \CC^n,\,\, W\in \Mat(1\times n,\CC)\}$$
for the action of
$\Gl_n(\CC)\times \CC^\times$ by $(g,\lambda)\cdot(V,W)=(gV \lambda^{-1},\lambda W g^{-1})$.
For any $\chi\in \CC^\times$ and for any sufficiently small invariant open subset $M_\reg\subset \{(V,W)\mid WV\neq 0\}$, the smooth map
\begin{equation}
  \psi:M_\reg \to T^\ast \CC^n\,, \quad
  (\tilde{V},\tilde{W}):=\psi(V,W)=\left( \frac{\ln(1+WV)}{\chi \, WV}\, V, \,W\right)\,,
 \end{equation}
is a well-defined biholomorphism between $M_\reg$ and its image, denoted $\widetilde M_\reg$, with inverse
 \begin{equation} \label{Eq:psiInv}
  \psi^{-1}:\widetilde M_\reg \to M_\reg\,, \quad
  (V,W):=\psi(\tilde V,\tilde W)=\left( \frac{\exp(\chi \tilde{W}\tilde{V})-1}{\tilde{W}\tilde{V}}\, \tilde V, \,\tilde W\right)\,.
 \end{equation}
We see that $(\1_n+VW)\circ \psi^{-1} =\exp(\chi \tilde{V}\tilde{W})$.
In fact, a standard calculation allows us to check that $\psi:M_\reg \to \widetilde{M}_\reg$ defines an isomorphism of complex Hamiltonian quasi-Poisson manifolds if we endow $\widetilde{M}_\reg$ with the quasi-Poisson bracket
 \begin{equation} \label{qPB-tild-VdB}
 \begin{aligned}
    \br{\tilde V_{i},\tilde V_{k}}&=0\,,\quad \br{\tilde W_{i},\tilde W_{k}}=0\,, \\
 \br{\tilde V_{i},\tilde W_{k}}&=
 \frac12 \tilde{W}\tilde{V}\, \coth\left(\frac{\chi}{2}\tilde{W}\tilde{V}\right) \delta_{ik}
 -\frac12 \left[\coth\left(\frac{\chi}{2}\tilde{W}\tilde{V}\right) - \frac{2}{\chi\tilde{W}\tilde{V}}\right] \tilde{V}_i\tilde{W}_k\,,
 \end{aligned}
 \end{equation}
and moment map $\tilde{\Phi}(\tilde V,\tilde W)= (\exp(\chi \tilde{V}\tilde{W}), \exp(-\chi \tilde{W}\tilde{V}))$.
Setting $\chi=\ic$ in \eqref{qPB-tild-VdB} shows that the structure on $\widetilde{M}_\reg$ can be regarded as
a complexification of the structure of Proposition \ref{Pr:Ans2}.
\end{remark}

\section{Auxiliary technical results}
\label{sec:J}

We here explain some results on invariant functions
that are used in the main text, and also prove the submanifold property
of $\fC_{**}\subset \cM_{d**}$ introduced in equation \eqref{R46}.

We begin by noting  that  for $h\in \Cinf(\UU(n))^{\UU(n)}$  the definition \eqref{B4} implies
the relation
\be
\nabla h(\eta g \eta^{-1}) = \eta \nabla h(g) \eta^{-1},
\qquad \forall \eta, g\in \UU(n).
\label{J1}\ee
It follows that $\nabla h(g)$ \emph{belongs to the center of the Lie algebra of the isotropy subgroup} (stabilizer) of $g$
with respect to the conjugation action of $\UU(n)$ on itself.
If $g\in \bT(n)_\reg$, then the stabilizer is just $\bT(n)$, and $\nabla h(g)$ belongs to its Lie algebra $\fT(n)$.
(For $g\in \bT(n) \setminus \bT(n)_\reg$, the center of the isotropy Lie algebra is a proper subalgebra of $\fT(n)$.)
By taking these remarks into account, for the $\uu(n)$-valued derivative $\nabla h $ is just
a device for encoding the exterior derivative $dh$,
we see that the main statement of Lemma \ref{Lm:LR4} is equivalent to the following result.

\begin{lemma}\label{Lm:LR4+}
The derivatives $\nabla h$ of the elements of $\Cinf(\UU(n))^{\UU(n)}$ span $\fT(n)$
at every regular element $Q\in \bT(n)_\reg$.
In other words, fixing any $Q\in \bT(n)_\reg$, there exist $h_1,\dots, h_n$ in $\Cinf(\UU(n))^{\UU(n)}$ such that
$\nabla h_1(Q),\dots, \nabla h_n(Q)$ are  linearly independent.
\end{lemma}
\begin{proof}
Consider the following $2n$ elements of $\Cinf(\UU(n))^{\UU(n)}$:
\be
h_k^\mathrm{r}(g):= \frac{1}{2k} \Re \tr(g^k), \qquad h_k^\ri(g):= \frac{1}{2k} \Im \tr(g^k),\qquad k=1,\dots, n.
\label{J2}\ee
Their derivatives are
\be
\nabla h_k^\mathrm{r}(g)= \frac{1}{2}(g^{-k} - g^k), \qquad \nabla h_k^\ri(g)= \frac{\ic}{2} (g^k + g^{-k}),\qquad k=1,\dots, n.
\label{J3}\ee
Now, inspect the derivatives at a fixed $Q=\diag(Q_1,\dots, Q_n)\in \bT(n)_\reg$.
 With respect to the standard basis $\ic E_{ll}$ ($l=1,\dots, n)$ of $\fT(n)$, the derivatives \eqref{J3} taken at $Q$ have the components
  \be
(\nabla h_k^\mathrm{r}(Q))_l = \frac{\ic}{2}(Q_l^{k} - Q_l^{-k}), \qquad (\nabla h_k^\ri(Q))_l= \frac{1}{2} (Q_l^k + Q_l^{-k}),\qquad k=1,\dots, n,
\label{J4}\ee
which represent them as elements of $\R^n$. An $n$-element subset of these $2n$ vectors in $\R^n$ is linearly dependent if and only if the corresponding
$n$ by $n$ determinant is zero.  If all such real determinants were zero, then also the determinant  of the matrix formed by
the complex vectors $V_k(Q)\in \CC^n$,  $(V_k(Q))_l:= Q_l^k$ were zero. (This holds since $V_k(Q)\in \CC^n$ is a complex linear combination
of two real vectors in \eqref{J4}.)
But the determinant formed out of the vectors $V_k(Q)$ is equal to
$\prod_{l=1}^n Q_l \prod_{1\leq j < k \leq n} (Q_k-Q_j)$,
which is non-zero if $Q\in \bT(n)_\reg$. This shows that at any fixed $Q\in \bT(n)_\reg$ there exist at least one $n$-element subset of the
$2n$ functions \eqref{J2}  such that their derivatives span $\fT(n)$.
\end{proof}

To complete the explanation of Lemma  \ref{Lm:LR4},  let us elucidate why every $h\in \Cinf(\UU(n))^{\UU(n)}$ can be expressed in terms of
the $n$ class functions  $h_1,\dots, h_n$
in a neighbourhood of $g\in \UU(n)_\reg$, where $dh_1(g),\dots, dh_n(g)$ are independent.
To see this, notice that $h_1,\dots, h_n$ can be enlarged \cite{Wa}  to a local coordinate system that is valid
on an open set containing $g$.  Then, $h$ can be expressed in terms of $h_1, \dots, h_n$ and the $(n^2 - n)$ additional coordinates.
 However, it cannot depend on the additional coordinates since we know that $dh$
can be expanded in $d h_1,\dots, dh_n$ alone (because the span of the differentials of the invariant functions
is at most $n$-dimensional at any $\eta \in \UU(n)$).

The following statement is also worth noting.

\begin{lemma}\label{Lm:B1+1}
Any $n$ of the $2n$ functions in \eqref{J2} are functionally independent on  a dense open subset of $\UU(n)$.
\end{lemma}
\begin{proof}
It is enough to verify that the determinant $D_n$ of the matrix whose
columns are the derivatives \eqref{J4} of the $n$ functions
 does not vanish identically on $\bT(n)$. Since we are dealing with real-analytic functions, this implies
 that the complement
of the zero set of the determinant is  dense open in $\bT(n)$.
By the relation \eqref{J1}, the derivatives are then independent on a dense open subset of $\UU(n)$.
The proof is completed by the $l=n$ case of the subsequent Lemma \ref{Lm:B1+2}.
\end{proof}

\begin{lemma}\label{Lm:B1+2}
For any $l=1,\dots, n$, let $f_1,\dots, f_l$
be a subset of the $2n$ real functions
\be
\cos (x),\,\sin (x),\,\cos (2x),\, \sin (2x),\, \dots, \cos (nx),\, \sin (nx),
\ee
which appear in \eqref{B4} if one puts $Q_l = e^{-\ic x}$.
Then, there exist real numbers $x_1,\dots, x_l$ for which
the determinant $D_l(x_1, \dots x_l)$ of the $l\times l$ matrix
$Y_{jm}(x_1,\dots, x_l) := f_m(x_j)$ ($j,m=1,\dots, l$) is not zero.
\end{lemma}
\begin{proof}
We proceed by induction on $l=1,\dots, n$. The $l=1$ case is plain.
The expansion of the determinant $D_l(x_1,\dots, x_l)$ according to the last row
of the matrix $Y$  gives
\be
D_l(x_1,\dots, x_l) = \sum_{k=1}^l (-1)^{k+l} M_k(x_1,\dots, x_{l-1}) f_k(x_l),
\label{Dlsum}\ee
where $M_k$ is the determinant $D_{l-1}$ built from the $l-1$ functions
$f_1,\dots, f_{k-1}, f_{k+1},\dots, f_l$. By the induction hypothesis, we may choose
$x_1,\dots, x_{l-1}$ so that, for example, $M_l(x_1,\dots, x_{l-1})\neq 0$.
Then, the linear independence
of the functions $f_1,\dots, f_l$ (say in $\Cinf(\R)$) guarantees that $x_l$ can be chosen
in such a way that
the sum \eqref{Dlsum} is not zero.
\end{proof}

In  \S\ref{ss:redDegInt} we introduced the manifold
\be
\cM_{d**}:= \{ (A, B, v_1,\dots, v_d) \in \cM_d \mid (A,v_1,\dots, v_d)\in M_*\},
\label{R45Y}\ee
 where $M_*$ is a certain open subset of $\UU(n)_\reg \times \CC^{n\times d}$, and its closed subset
\be
\fC_{**}:= \{ (A,\tilde B, v_1,\dots, v_d)\in \cM_{d**}\mid \chi(A) - \chi(\tilde B) =0,\, \forall \chi\in \Cinf(\UU(n))^{\UU(n)}\}.
\label{R46Y}\ee
For completeness, now we prove the claim that was stated after equation \eqref{R46}.

\begin{lemma}\label{Lm:subman}
The  subset $\fC_{**} \subset \cM_{d**}$ is an embedded (regular) submanifold of codimension $n$.
\end{lemma}
\begin{proof}
The condition in \eqref{R46Y} on the pairs $(A,\tilde B)$ is equivalent to the requirement that $A$ and $\tilde B$ have the same eigenvalues.
In turn, this is equivalent to the relations
\be
h_i(A) - h_i(\tilde B) = 0, \qquad i=1,\dots, 2n,\ee
where  $\{h_i\}$ denotes the set of $2n$ functions listed in \eqref{J2}.
Let us take an arbitrary point $(A_0, \tilde B_0, v_1,\ldots,v_d) \in \fC_{**}$  and consider an open neighbourhood of
this point of the form
\be
W:= U \times (\eta U \eta^{-1}) \times V,
\ee
where $U$ is an open neighbourhood of $A_0$ in $\UU(n)_\reg$,  $\eta\in\UU(n)$ is chosen so that $\eta A_0 \eta^{-1} = \tilde B_0$, and $V$
is a neighbourhood of $(v_1,\ldots,v_d)\in \CC^{n\times d}$.
We can select $n$ functions out of the $h_i$ ($i=1,\dots, 2n$) so that their exterior derivatives are linearly independent at $A_0$,
and then we can also choose $U$ in such a way that the same holds at every element of $U$.
By re-labeling, we may assume that the $n$ independent functions are $h_1,\dots, h_n$.
Then, it follows from Lemma \ref{Lm:LR4} that the values of $h_{n+1},\dots, h_{2n}$ at any $A\in U$ are
determined by $h_1(A),\dots, h_n(A)$, and similarly for $\tilde B  \in \eta U \eta^{-1}$.
Thus, we have proved that the intersection $\fC_{**}\cap W$ is  the joint zero set of the functions
$F_i \in \Cinf(W,\R)$ defined by
\be
F_i(A,\tilde B, v_1,\ldots,v_d) := h_i(A) - h_i(\tilde B), \quad i=1,\dots, n,
\ee
whose differentials are linearly independent on $W$.
By a standard result of differential geometry  \cite{Wa}, this implies that the subset $\fC_{**}$ inherits a unique manifold
structure from $\cM_{d**}$, such that it also has the subspace topology.
\end{proof}

\section{Poisson algebra of invariant first integrals and a conjecture}
\label{sec:M}

Consider $\MM_d$ equipped with the quasi-Poisson bracket associated with the quasi-Poisson
bivector $P_{\underline{z}}$ defined in   \eqref{Eq:PencCor} for a fixed $\underline{z}=(z_{\alpha\beta})_{\alpha<\beta} \in \R^{d(d-1)/2}$, for any $d\geq 2$.
Recall from \S\ref{ss:Unred} that $\UU(n)$-invariant functions of the matrix $A$ are Poisson commuting on $\MM_d$, and  together with their first integrals (constants of motion) they satisfy a property analogous to being degenerately integrable. Motivated by the
original considerations of Gibbons--Hermsen \cite{GH} and the investigation of Marshall and the authors in the Poisson--Lie setting \cite{FFM}, we are led to consider among all the $\UU(n)$-invariant first integrals those of a peculiar form: they are obtained by taking traces of matrices involving solely $A$ and the spin variables $v_1,\ldots,v_d$.
The Poisson algebra   $\fI$ \eqref{Eq:fItr} generated by such first integrals contains in its Poisson center the algebra $\fH_{\tr}$ of polynomials in the functions  $\Re(\tr A^k),\Im(\tr A^k)$, $k\in \Z$; in analogy with \cite{FFM} we conjecture that the pair $(\fI,\fH_{\tr})$ yields a degenerate integrable system on the reduced phase space $\Phi^{-1}(e^{\ic \gamma})/\UU(n)$ considered in Section \ref{S:spinRS} (for generic $\gamma\in \R$), though we can not prove this result at present.

\medskip

We start by working in the Poisson algebra $\Cinf(\MM_d)^{\UU(n)}$ obtained by restricting the quasi-Poisson bracket (associated with the quasi-Poisson bivector $P_{\underline{z}}$ from \eqref{Eq:PencCor}) on $\Cinf(\MM_d)$ to invariant functions.
We let  $\fH_{\tr}:=\R[\Re(\tr A^k),\Im(\tr A^k) \mid k\in \Z]$, and we define
\begin{equation}
\fI_0:=\{F(x_1 |v_1|^2,\ldots,x_d |v_d|^2) \mid F(t_1,\ldots,t_d) \text{ is real analytic on }(-2\pi,2\pi)^{\times d}\,\}\,.
\end{equation}
In particular,  $e^{|v_\alpha|^2}$ or $\btt(x_\alpha |v_\alpha|^2)$ with $\btt(t)$ defined in \eqref{Eq:ab-fct}  belong to $\fI_0$ for any $1\leq \alpha\leq d$.
We put $I^k_{\alpha \beta}:=v_\alpha^\dagger A^k v_\beta$ (=$\tr(v_\beta v_\alpha^\dagger A^k)$) and we introduce the polynomial algebra generated by the real and imaginary parts of the $I^k_{\alpha \beta}$ and $\tr(A^k)$ with coefficients in $\fI_0$:
\begin{equation} \label{Eq:fItr}
 \fI:=\fI_0[\Re(\tr A^k),\Im(\tr A^k), \Re(I^k_{\alpha \beta}),\Im(I^k_{\alpha \beta}) \mid k\in \Z_{\geq0},\, 1\leq \alpha,\beta \leq d]\,.
\end{equation}
Note that $I^0_{\alpha\alpha}=|v_\alpha|^2\in \fI_0\subset \fI$ for all $1\leq \alpha \leq d$ and $\fH_{\tr}\subset \fI$.
When defining $\fI$, we could allow any $k\in \Z$ since $A$ is unitary and $I^{-k}_{\alpha \beta}=\overline{I_{\beta \alpha}^k}$.
We also remark that, while $\fI$ is infinite-dimensional, it is finitely generated over $\fI_0$ : we only need to consider  functions $\tr A^k,I^k_{\alpha \beta}$ with exponent $k\in \{1,\ldots,n\}$ due
to the Cayley--Hamilton theorem for $A$.

\begin{proposition} \label{Pr:FInt-A}
For the parameters  $\underline{z}=(z_{\alpha\beta})_{\alpha<\beta}$ defining the quasi-Poisson bivector $P_{\underline{z}}$ as in \eqref{Eq:PencCor}, consider their antisymmetric extension $(z_{\alpha\beta}^\ast)_{1\leq \alpha, \beta\leq d}$ obtained by setting $z_{\alpha\alpha}^\ast=0$ and $z_{\alpha\beta}^\ast=-z_{\beta\alpha}^\ast:=z_{\alpha\beta}$ for $\alpha<\beta$.
Denote by $\br{-,-}_{\underline{z}}$ the quasi-Poisson bracket associated with  $P_{\underline{z}}$.
\begin{enumerate}
 \item[(1)] $\br{\tr(A^k),\tr(A^l)}_{\underline{z}}=0$ and $\br{I_{\alpha\beta}^k,\tr(A^l)}_{\underline{z}}=0$ for any $k,l\in \Z_{\geq0}$ and $1\leq \alpha,\beta \leq d$.
 \item[(2)] For any $k,l\in \Z_{\geq0}$ and $1\leq \alpha,\beta,\gamma,\epsilon \leq d$,  the following holds
 \begin{equation}
  \begin{aligned}
\br{I_{\alpha\beta}^k,I_{\gamma\epsilon}^l&}_{\underline{z}}=
\frac12 I_{\gamma\beta}^0 I_{\alpha\epsilon}^{k+l} - \frac12 I_{\gamma \beta}^{k+l} I_{\alpha\epsilon}^{0}
+\frac12 \left(\sum_{1\leq r\leq k} - \sum_{1\leq r \leq l}\right)
 \left(I_{\gamma\beta}^{r} I_{\alpha\epsilon}^{k+l-r} - I_{\gamma\beta}^{k+l-r}I_{\alpha\epsilon}^{r} \right)  \\
 &+\frac12 \sgn(\gamma-\alpha) I_{\gamma\beta}^{k}I_{\alpha\epsilon}^l - z_{\alpha \gamma}^\ast\,I_{\alpha\beta}^{k}I_{\gamma\epsilon}^l
 +\frac12 \sgn(\epsilon-\beta) I_{\gamma\beta}^{l}I_{\alpha\epsilon}^k - z_{\beta \epsilon}^\ast\,I_{\alpha\beta}^{k}I_{\gamma\epsilon}^l   \\
 &-\frac12 \sgn(\epsilon-\alpha) I_{\gamma\beta}^{k+l}I_{\alpha\epsilon}^0 + z_{\alpha \epsilon}^\ast \,I_{\alpha\beta}^{k}I_{\gamma\epsilon}^l
 -\frac12 \sgn(\gamma-\beta) I_{\gamma\beta}^{0}I_{\alpha\epsilon}^{k+l} + z_{\beta\gamma}^\ast \,I_{\alpha\beta}^{k}I_{\gamma\epsilon}^l  \\
&+\frac{\ic}{x_\beta} \, \delta_{\beta\gamma}  I_{\alpha\epsilon}^{k+l}
+\frac{\ic}{2}\delta_{\beta\gamma} b(x_\beta |v_\beta|^2)
 \left[|v_\beta|^2  I_{\alpha\epsilon}^{k+l} - I_{\alpha\beta}^{k}I_{\gamma\epsilon}^l\right]  \\
 &-\frac{\ic}{x_\alpha} \,\delta_{\alpha\epsilon}I_{\gamma\beta}^{k+l}
-\frac{\ic}{2}\delta_{\alpha\epsilon}\btt(x_\alpha |v_\alpha|^2) \left[|v_\alpha|^2 I_{\gamma\beta}^{k+l} - I_{\alpha\beta}^{k}I_{\gamma\epsilon}^l\right] \,,
 \label{Eq:FInt-A-2}
  \end{aligned}
 \end{equation}
 where we omit the first line if $k=0$ or $l=0$.
\item[(3)] $\br{|v_\alpha|^2,I_{\gamma\epsilon}^l}_{\underline{z}}= \frac{\ic}{x_\alpha} (\delta_{\alpha\gamma}-\delta_{\alpha\epsilon}) I_{\gamma\epsilon}^{l}$, hence $\br{|v_\alpha|^2,|v_\gamma|^2}_{\underline{z}}=0$ for any $l\in \Z_{\geq0}$, $1\leq \alpha,\gamma,\epsilon \leq d$.
\end{enumerate}
\noindent In particular, the subalgebra $\fI\subset \Cinf(\MM_d)^{\UU(n)}$ is a Poisson algebra whose center contains $\fH_{\tr}$.
\end{proposition}
\begin{proof}
We have \emph{(1)} directly from \S\ref{ss:Unred}.
We deduce \emph{(3)} from \emph{(2)} by taking $k=0$ then $l=0$  and $\gamma=\epsilon$.
In view of these three items, $\fI\subset \Cinf(\MM_d)^{\UU(n)}$ is a Poisson algebra. Moreover, the Poisson center of $\fI$ contains $\fH_{\tr}$ by \emph{(1)}.

To get \emph{(2)}, recall that we have the decomposition
$\br{-,-}_{\underline{z}}=\br{-,-}+\br{-,-}_{\psi_{\underline{z}}}$,
where $\br{-,-}$ is the bracket explicitly presented in Theorem \ref{Thm:qH-Master},
while $\br{-,-}_{\psi_{\underline{z}}}$ satisfies \eqref{Eq:underz-1} and \eqref{Eq:underz-vv}.
Deriving \eqref{Eq:FInt-A-2} is therefore a standard computation after decomposing the elements
$I^k_{\alpha \beta}=\sum_{1\leq i,j\leq n} (\bar{v}_\alpha)_i  (A^k)_{ij} (v_\beta)_j$ in terms of entries of the matrices  defining $\MM_d$.
\end{proof}

When $x_1=\ldots=x_d=1$ and $z_{\alpha \beta}^\ast=\sgn(\alpha-\beta)$ for any $1\leq \alpha,\beta \leq d$, Proposition \ref{Pr:FInt-A} yields a real version of the algebra of first integrals given by Chalykh and the first-named author in \cite[\S5.1]{CF2}.
This choice of parameters $(z_{\alpha \beta}^\ast)$ corresponds to performing successive fusions of the residual $\UU(1)$-actions as we explained in \S\ref{sss:Comm-pen}, which correspond to analogous fusions of $\CC^\times$-actions performed in \cite{CF2}.
In fact, these Poisson algebras can be compared through the formulas \eqref{Eq:FInt-A-2} and \cite[(5.1)]{CF2}: they match up to an overall sign (due to a different choice of trace pairing) except for the terms in $\delta_{\alpha \epsilon}$ and $\delta_{\beta\gamma}$ whose discrepancy is explained by the (local) change of variables given in Remark \ref{Rem:VdB}.

On the other hand, we recall that in the complex   case,   besides the one in \cite{CF2}, there
 is another Poisson algebra that was
derived by Arutyunov and Olivucci \cite{AO} from a Poisson--Lie perspective.
While both algebras describe the first integrals for the  trigonometric/hyperbolic spin RS system in the complex setting,
it is a puzzling observation that they do not match. Nevertheless, if we take our
formula \eqref{Eq:FInt-A-2} in the case where $x_1=\ldots=x_d=1$ and all $z_{\alpha\beta}=0$,
then comparison  with \cite[Eq.~(7.5)]{AO} (choosing $\kappa=1$  and the `\textit{minus}' Poisson bracket)
shows  that the two Poisson algebras have the same form.\footnote{The comparison with \cite{AO} is rather subtle:
one not only needs to complexify and make use of Remark \ref{Rem:VdB}, but also needs to use an
alternative choice of first integrals after complexification.
See also \cite[Remark 5.2]{CF2}.}
Building on these observations, we expect that the Poisson algebras of \cite{AO} and \cite{CF2} can be simultaneously considered in the following way.

\begin{conjecture} \label{Conj:PenC}
There exists a pencil of compatible Poisson brackets on the phase space
of the complex trigonometric spin RS system (denoted $\mathcal{M}_{n,d,q}^\times/\!/\Gl_n(\CC)$ in \cite{CF2}) which
contains the Poisson structures of Arutyunov--Olivucci \cite{AO} and of Chalykh--Fairon \cite{CF2} as special instances.
Furthermore, for all such Poisson brackets, one can obtain a Hamiltonian formulation of the
trigonometric spin RS system of Krichever and Zabrodin \cite{KZ}.
\end{conjecture}

We are confident that the results presented in \S\ref{ss:Pencil} can be adapted to the complex setting to yield a pencil of compatible quasi-Poisson brackets on $\mathcal{M}_{n,d,q}^\times$, a special case of which is given in \cite{CF2}. We also believe that Conjecture \ref{Conj:PenC}
can be approached from a Poisson--Lie perspective by endowing the oscillator manifold of \cite{AO} with a pencil of compatible Poisson brackets.


\newpage





\begin{thebibliography}{99}

\addcontentsline{toc}{section}{References}

    \setlength{\parskip}{0em}

 \bibitem{AKSM}
A.~Alekseev, Y.~Kosmann-Schwarzbach and E.~Meinrenken,
{\it Quasi-Poisson manifolds}.
Canad. J. Math. {\bf 54} (2002) 3-29;
\href{https://arxiv.org/abs/math/0006168}{\tt arXiv:math/0006168}

\bibitem{AMM}
A.~Alekseev, A.~Malkin and E.~Meinrenken,
{\it Lie group valued moment maps}.
J. Differential Geom. {\bf 48} (1998) 445-495;
 \href{https://arxiv.org/abs/dg-ga/9707021}{\tt arXiv:dg-ga/9707021}

\bibitem{ARe}
S.~Arthamonov and N.~Reshetikhin,
{\it Superintegrable systems on moduli spaces of flat connections}.
Commun. Math. Phys. {\bf 386} (2021) 1337-1381;
 \href{https://arxiv.org/abs/1909.08682}{\tt arXiv:1909.08682}

\bibitem{AF}
G.~Arutyunov and S.~Frolov,
{\it On the Hamiltonian structure of the spin Ruijsenaars--Schneider model}.
J. Phys. A {\bf 31} (1998) 4203-4216;
 \href{https://arxiv.org/abs/hep-th/9703119}{\tt arXiv:hep-th/9703119}

\bibitem{AFM}
G.E.~Arutyunov, S.A.~Frolov and P.B.~Medvedev,
{\it Elliptic Ruijsenaars--Schneider model from the cotangent bundle over the two-dimensional current group}.
J. Math. Phys. {\bf 38} (1997) Paper No. 5682;
\href{https://arxiv.org/abs/hep-th/9608013}{\tt arXiv:hep-th/9608013}


\bibitem{AO}
G.~Arutyunov and E.~Olivucci,
{\it Hyperbolic spin Ruijsenaars--Schneider model from Poisson reduction}.
Proc. Steklov Inst. Math. {\bf 309} (2020) 31-45;
 \href{https://arxiv.org/abs/1906.02619}{\tt arXiv:1906.02619}

\bibitem{Br}
G.E.~Bredon,
Introduction to Compact Transformation Groups.
Academic Press, 1972

\bibitem{C}
 F.~Calogero,
{\it Solution of the one-dimensional N-body problem with quadratic and/or
inversely quadratic pair potentials}.
J. Math. Phys. {\bf 12} (1971) 419-436

\bibitem{CF1}
O.~Chalykh and M.~Fairon,
{\it Multiplicative quiver varieties and generalised Ruijsenaars--Schneider models}.
J. Geom. Phys. {\bf 121} (2017) 413-437;
\href{https://arxiv.org/abs/1704.05814}{\tt arXiv:1704.05814}

 \bibitem{CF2}
 O.~Chalykh and M.~Fairon,
{\it On the Hamiltonian formulation of the trigonometric spin Ruijsenaars--Schneider system}.
Lett. Math. Phys. {\bf 110} (2020) 2893-2940;
\href{https://arxiv.org/abs/1811.08727}{\tt arXiv:1811.08727}

\bibitem{DK}
J.J.~Duistermaat and J.A.C.~Kolk,
Lie Groups. Universitext. Springer-Verlag, 2000

\bibitem{EV}
P.~Etingof and A.~Varchenko,
{\it Geometry and classification of solutions of the classical dynamical Yang-Baxter equation}.
Commun. Math. Phys. {\bf 192} (1998) 77-120;
 \href{https://arxiv.org/abs/q-alg/9703040}{\tt arXiv:q-alg/9703040}


\bibitem{Fa}
M.~Fairon,
{\it Integrable systems on multiplicative quiver varieties from cyclic quivers}.
Preprint,
\href{https://arxiv.org/abs/2108.02496}{\tt arXiv:2108.02496}

\bibitem{FFM}
M.~Fairon, L.~Feh\'er and I.~Marshall,
{\it Trigonometric real form of the spin RS model of Krichever and Zabrodin}.
Annales Henri Poincar\'{e} {\bf 22} (2021) 615-675;
\href{https://arxiv.org/abs/2007.08388}{\tt arXiv:2007.08388}

\bibitem{FG}
M.~Fairon and T.~G\"{o}rbe,
{\it Superintegrability of Calogero--Moser systems associated with the cyclic quiver}.
Nonlinearity {\bf 34} (2021) 7662--7682;
\href{https://arxiv.org/abs/2101.05520}{\tt arXiv:2101.05520}

\bibitem{Fe1}
L.~Feh\'er, {\it Poisson--Lie analogues of spin Sutherland models}.
Nucl. Phys. B {\bf 949} (2019) Paper No. 114807;
 \href{https://arxiv.org/abs/1809.01529}{\tt arXiv:1809.01529}

\bibitem{Fe2}
L.~Feh\'er,
{\it Bi-Hamiltonian structure of Sutherland models coupled to two $\mathfrak{u}(n)^\ast$-valued spins from Poisson reduction}.
Nonlinearity {\bf 35} (2022) 2971--3003;
\href{https://arxiv.org/abs/2109.07391}{\tt arXiv:2109.07391}

\bibitem{Fe} L.~Feh\'er, {\it Poisson reductions of master integrable systems on doubles of compact Lie groups}.
Annales Henri Poincar\'{e} {\bf 24} (2023)  1823-1876;
\href{https://arxiv.org/abs/2208.03728}{\tt arXiv:2208.03728}

\bibitem{FK2}
L.~Feh\'er and C.~Klim\v{c}\'\i k,
{\it Poisson--Lie interpretation of trigonometric Ruijsenaars duality}.
Commun. Math. Phys. {\bf 301} (2011) 55-104;
  \href{https://arxiv.org/abs/0906.4198}{\tt arXiv:0906.4198}


\bibitem{FK}
L.~Feh\'er and C.~Klim\v{c}\'\i k,
{\it Self-duality of the compactified Ruijsenaars--Schneider system from quasi-Hamiltonian reduction}.
Nucl. Phys. B {\bf 860} (2012) 464-515;
\href{https://arxiv.org/abs/1101.1759}{\tt arXiv:1101.1759}


\bibitem{FKl}
L.~Feh\'er and T.J.~Kluck,
{\it New compact forms of the trigonometric Ruijsenaars--Schneider system}.
Nucl. Phys. B {\bf 882} (2014), 97-127;
  \href{https://arxiv.org/abs/1312.0400}{\tt arXiv:1312.0400}

\bibitem{FGNR}
V.~Fock,  A.~Gorsky, N.~Nekrasov and V.~Rubtsov,
{\it Duality in integrable systems and gauge theories}.
JHEP {\bf 07} (2000) Paper No. 028;  \href{https://arxiv.org/abs/hep-th/9906235}{\tt arXiv:hep-th/9906235}

\bibitem{GH}
J.~Gibbons and T.~Hermsen,
{\it A generalisation of the Calogero--Moser system}.
Physica D {\bf 11} (1984) 337-348

\bibitem{Go}
M.~Got\^{o},
{\it A theorem on compact semi-simple groups}.
J. Math. Soc. Japan {\bf 1} (1949) 270-272

\bibitem{HJS}
J.~Hurtubise, L.~Jeffrey and R.~Sjamaar,
{\it Group-valued implosion and parabolic structures}.
Amer. J. Math. {\bf 128} (2006) 167-214;
\href{https://arxiv.org/abs/math/0402464}{\tt arXiv:math/0402464}

\bibitem{J}
B.~Jovanovic,
{\it Symmetries and integrability}.
Publ. Institut Math. {\bf 49} (2008) 1-36;
 \href{https://arxiv.org/abs/0812.4398}{\tt arXiv:0812.4398}


\bibitem{KKS}
D.~Kazhdan, B.~Kostant and S.~Sternberg,
{\it Hamiltonian group actions and dynamical systems of Calogero type}.
Comm. Pure Appl. Math. {\bf 31} (1978) 481-507

\bibitem{KLOZ}
 S.~Kharchev, A.~Levin, M.~Olshanetsky and  A.~Zotov,
 {\it Quasi-compact Higgs bundles and Calogero--Sutherland systems with two types spins}.
 J. Math. Phys. {\bf 59} (2018) Paper No. 103509;
  \href{https://arxiv.org/abs/1712.08851}{\tt arXiv:1712.08851}

\bibitem{KMK}
K.~Kowalczyk-Murynka and M.~Ku\'{s},
{\it Matrix and vectorial generalized Calogero--Moser models.} Physica D {\bf 440} (2022) 133491

\bibitem{KZ}
 I.~Krichever and A.~Zabrodin,
{\it Spin generalization of the Ruijsenaars--Schneider model, non-abelian 2D
Toda chain and representations of Sklyanin algebra}.
Russian Math. Surveys {\bf 50} (1995) 1101-1150;
 \href{https://arxiv.org/abs/hep-th/9505039}{\tt arXiv:hep-th/9505039}

\bibitem{LGMV}
C.~Laurent-Gengoux, E.~Miranda and P.~Vanhaecke,
{\it Action-angle coordinates for integrable systems on Poisson manifolds}.
Int. Math. Res. Not. IMRN {\bf 2011}, 1839-1869;
 \href{https://arxiv.org/abs/0805.1679}{\tt arXiv:0805.1679}

\bibitem{LB}
A.~Le~Blanc,
{\it Quasi-Poisson structures and integrable systems related to
the moduli space of flat connections on a punctured Riemann sphere}.
J. Geom. Phys. {\bf 57} (2007) 1631-1652

\bibitem{LX}
L.-C.~Li and P.~Xu,
{\it A class of integrable spin Calogero--Moser systems}.
Commun. Math. Phys. {\bf 231} (2002) 257-286;
 \href{https://arxiv.org/abs/math/0105162}{\tt arXiv:math/0105162}

\bibitem{MF}
A.S.~Mischenko and A.T.~Fomenko,
{\it  Generalized Liouville method for integrating Hamiltonian systems}.
Funct. Anal. Appl. {\bf 12} (1978) 113-125

 \bibitem{Mo}
J.~Moser,
{\it Three integrable Hamiltonian systems connected with isospectral deformations}.
Adv. Math. {\bf 16} (1975) 197-220

\bibitem{Nek}
N.N.~Nekhoroshev,
{\it  Action-angle variables and their generalizations}.
Trans. Moscow Math. Soc. {\bf 26} (1972) 180-197

\bibitem{N}
N.~Nekrasov,
{\it Infinite-dimensional algebras, many-body systems and gauge theories.}
 Moscow Seminar in Mathematical Physics,
AMS Transl. Ser. 2, Vol.~191, Amer. Math. Soc., Providence, RI, pp. 263-299, 1999

\bibitem{OP}
M.A. Olshanetsky and A.M. Perelomov, {\it Classical integrable
finite-dimensional systems related to Lie algebras.}
Phys. Rept. {\bf 71} (1981) 313-400

\bibitem{OR}
J.-P.~Ortega and T.~Ratiu,
Momentum Maps and Hamiltonian Reduction. Birkh\"auser, 2004

\bibitem{Pr}
C.~Procesi,
 Lie Groups. An Approach through Invariants and Representations.
Universitext. Springer, 2007

\bibitem{Pu}
B.G.~Pusztai,
{\it The hyperbolic $BC(n)$ Sutherland and the rational $BC(n)$
Ruijsenaars--Schneider--van Diejen models: Lax matrices and duality}.
Nucl. Phys. B  {\bf 856} (2012) 528-551;
\href{https://arxiv.org/abs/1109.0446}{\tt arXiv:1109.0446}

\bibitem{Re1}
N.~Reshetikhin,
{\it Degenerate integrability of spin Calogero--Moser systems and the
duality with the spin Ruijsenaars systems}.
Lett. Math. Phys. {\bf 63} (2003) 55-71;
\href{https://arxiv.org/abs/math/0202245}{\tt arXiv:math/0202245}

\bibitem{Re2}
 N.~Reshetikhin,
{\it Degenerately integrable systems}.
J. Math. Sci. {\bf 213} (2016) 769-785;
\href{https://arxiv.org/abs/1509.00730}{\tt arXiv:1509.00730}

\bibitem{Re3}
N.~Reshetikhin,
{\it Spin Calogero--Moser models on symmetric spaces}.
Integrability, Quantization, and Geometry. I. Integrable Systems, Proc. Sympos. Pure Math., 103.1, Amer. Math. Soc.,
Providence, RI, pp. 377-402, 2021;
\href{https://arxiv.org/abs/1903.03685}{\tt arXiv:1903.03685}

\bibitem{ReS}
N.~Reshetikhin and G.~Schrader,
{\it Superintegrability of generalized Toda models on symmetric spaces}.
Int. Math. Res. Not. IMRN {\bf 2021}, 12993-13010;
\href{https://arxiv.org/abs/1802.00356}{\tt arXiv:1802.00356}

\bibitem{RudS}
G.~Rudolph and M.~Schmidt,
Differential Geometry and Mathematical Physics. Part I. Manifolds, Lie Groups and
Hamiltonian Systems. Springer, 2013

\bibitem{Ru}
S.N.M.~Ruijsenaars,
{\it Action-angle maps and scattering theory for some finite-dimensional
integrable systems III. Sutherland type systems and their duals}.
Publ. RIMS {\bf 31} (1995) 247-353

\bibitem{RS}
S.N.M.~Ruijsenaars and H.~Schneider,
{\it A new class of integrable systems and its relation to solitons}.
Ann. Phys. (NY) \textbf{170} (1986)  370-405

\bibitem{SL}
R.~Sjamaar and E.~Lerman,
{\it Stratified symplectic spaces and reduction}.
Ann. of Math.  {\bf 134} (1991) 375-422

\bibitem{Sn}
J.~\'{S}niatycki,
Differential Geometry of Singular Spaces and Reduction of Symmetries.
Cambridge University Press,  2013

\bibitem{STS}
M.A.~Semenov-Tian-Shansky,
{\it Integrable systems: an r-matrix approach}.
Kyoto preprint RIMS-1650, 2008;
\href{http://www.kurims.kyoto-u.ac.jp/preprint/file/RIMS1650.pdf}{kurims.kyoto-u.ac.jp/preprint/file/RIMS1650.pdf}

\bibitem{Su}
B.~Sutherland,
{\it Exact results for a quantum many-body problem in one dimension}.
Phys. Rev. A {\bf 4} (1971) 2019-2021

\bibitem{TZ}
E.~Trunina and A.~Zotov,
{\it Lax equations for relativistic GL($NM,\C$) Gaudin models on elliptic curve}.
J. Phys. A {\bf 55} (2022), Paper No. 395202, 38 pp.;
\href{https://arxiv.org/abs/2204.06137}{\tt arXiv:2204.06137}

\bibitem{VdB1}
M.~Van den Bergh,
{\it Double Poisson algebras}.
Trans. Amer. Math. Soc. {\bf 360} (2008) 5711--5769;
\href{https://arxiv.org/abs/math/0410528}{\tt arXiv:math/0410528}

\bibitem{Wa}
F.W.~Warner, Foundations of Differentiable Manifolds and Lie Groups.
Springer, 1983

\bibitem{Wo}
S.~Wojciechowski,
{\it An integrable marriage of the Euler equations with the Calogero--Moser system.}
Phys. Lett. A \textbf{111}  (1985) 101-103

\end{thebibliography}
\end{document}